\documentclass[12pt,english,numsec]{extarticle}

\usepackage[T1]{fontenc}
\usepackage[latin9]{inputenc}
\usepackage[a4paper]{geometry}
\usepackage{amsmath}
\usepackage{amssymb}
\usepackage{amsfonts}
\usepackage{amsthm}
    \newtheorem{theorem}{Proposition}[section]
    \newtheorem{corollary}{Corollary}[theorem]
    \newtheorem{lemma}[theorem]{Lemma}
    
    \newtheorem{remark}{Remark}[section]
    \theoremstyle{definition}
    \newtheorem{definition}{Definition}[section]
\usepackage{mathrsfs}
\usepackage{tocbibind}
\usepackage{dsfont}
\usepackage{braket}
\usepackage{commath}
\usepackage{ragged2e}
\usepackage{bigints}
\usepackage{babel}
\usepackage{hyperref}
    \hypersetup{colorlinks=true, linktocpage=true, pageanchor=true, allcolors=black}

\makeatletter
\date{}
\makeatother

\begin{document}

\title{n Distinguishable Particles interacting via Two-Body Delta Potentials in One Spatial Dimension}
\author{Antonio Moscato\footnote{antonio\_moscato@ymail.com}}
\maketitle

\justifying

\begin{abstract}
    \justifying
    \noindent This paper studies a system of $n \in \mathbb{N}: \, n \geq 2$ non-relativistic, spinless quantum particles moving on the real line and interacting via a two-body delta potential. The Hamiltonian of such a system is proved to be affiliated to the resolvent algebra of the case, $\mathcal{R}\left( \mathbb{R}^{2n},\sigma \right)$; it is further shown the existence of a $\text{C}^{\ast}-$dynamical system and of a subalgebra $\pi_S\left( \mathfrak{S}_0 \right)^{-1} \subset \mathcal{R}\left( \mathbb{R}^{2n},\sigma \right)$, stable under time evolution, where $\pi_S$ is the Schr{\"o}dinger representation of the resolvent algebra.
\end{abstract}

\section{Introduction}
\justifying
The system investigated is made up of $n \in \mathbb{N}: \, n \geq 2$ distinguishable particles, interacting via a two-body delta potential in one spatial dimension. The symbolic Hamiltonian governing the system is

\begin{equation}\label{Symbolic Hamiltonian two-body}
    H = -\sum_{i=1}^n \, \frac{1}{2m_i} \frac{\partial^2}{\partial x_i^2} - g\sum_{1 \leq i < j \leq n} \, \delta\left( x_i - x_j \right) \equiv H_0 - g\sum_{1 \leq i < j \leq n} \, \delta\left( x_i - x_j \right),
\end{equation}

\noindent where $g \in \mathbb{R} \setminus \{0\}$ is the coupling constant, $m_i \in \mathbb{R} \setminus \{0\}$ the mass of the $i^{\text{th}}-$particle, $\left(H_0,\mathcal{D}_{H_0}\right)$ the free Hamiltonian. Purpose of the paper is showing that (\ref{Symbolic Hamiltonian two-body}) is affiliated to $\mathcal{R}\left( \mathbb{R}^{2n},\sigma \right)$, the resolvent algebra on the symplectic space $\left( \mathbb{R}^{2n}, \sigma \right)$. Such a $\text{C}^{\ast}-$algebraic formalism was introduced by D. Buchholz and H. Grundling in \cite{01_BG} and has proved useful, since then, for both the finite and the infinite dimensional quantum mechanical modeling cases (see \cite{01_BG}, \cite{02_B}, \cite{03_B}, \cite{04_BG}, \cite{06_M}). The adopted strategy to prove the announced purpose is briefly sketched: given an even smooth function of compact support $v$\footnote{It does not harm generality assuming $\int_{\mathbb{R}} \, v^2 = 1$.}, said $V \doteq v^2$ and
\begin{equation*}
    V_\epsilon : \, x \in \mathbb{R} \longmapsto V_{\epsilon}\left( x \right) \doteq \frac{1}{\epsilon} V\left( \frac{x}{\epsilon} \right) \in \mathbb{R}, \quad \epsilon > 0,
\end{equation*}

\noindent the Schr{\"o}dinger Hamiltonians
\begin{equation}\label{Approximating Hamiltonians}
    H_\epsilon = - \sum_{i=1}^n \, \frac{1}{2m_i} \frac{\partial^2}{\partial x_i^2} - g \sum_{1 \leq i < j \leq n} \, V_\epsilon \left( x_i - x_j \right) \equiv H_0 - g \sum_{1 \leq i < j \leq n} \, V_\epsilon^{(ij)}, \quad \epsilon > 0,
\end{equation}

\noindent self-adjoint on $\mathcal{D}_{H_0}$, are considered. \cite{04_BG}, prop. 4.1 guarantees that $R_{H_{\epsilon}}\left( z \right) \in \pi_S\left[ \mathcal{R}\left( \mathbb{R}^{2n},\sigma \right) \right]$ for all $\epsilon > 0, \; z \in i\mathbb{R}\setminus\{0\}$, where $\pi_S$ is the Schr{\"o}dinger representation of the resolvent algebra. the result then follows from $H_{\epsilon}$ converging to $H$ in the norm resolvent sense. Starting from section 2, what required to prove the result is exposed. $\hfill \square$

\vfill

\section{Preliminaries}
\justifying

\begin{remark}
    Given $i,j \in \{1,\ldots,n\}$ such that $i<j$, the coordinate transformation $\vec{f}_{\left(ij\right)}: \textbf{x} \in \mathbb{R}^n \longmapsto \vec{f}_{\left(ij\right)}\left( \textbf{x} \right) \equiv \left(r_{(ij)}, R_{(ij)}, y_1, \ldots, \hat{y}_i, \ldots, \hat{y}_{j}, \ldots, y_n \right) \in \mathbb{R}^n$, (hats represent omission) where
    \begin{equation}\label{Coordinate Transformation}
        \vec{f}\left( \textbf{x} \right) = 
        \begin{cases}
            R_{(ij)} & = \frac{m_i x_i + m_j x_j}{m_i + m_j} \\
            r_{(ij)} & = x_i - x_j \\
            y_k & = x_k, \quad k \neq i,j
        \end{cases},
    \end{equation}
    \noindent is considered. The corresponding jacobian is identically equal to $1$, as can be easily verified, and the inverse transformation is
    \begin{align*}
        \vec{f}_{\left(ij\right)}^{-1}: \, \left(r_{(ij)}, R_{(ij)}, y_1, \ldots, \hat{y}_i, \ldots, \hat{y}_{j}, \ldots, y_n \right) \in \mathbb{R}^{n} \longmapsto & \vec{f}_{\left(ij\right)}^{-1} \left(r_{(ij)}, R_{(ij)}, y_1, \ldots, \hat{y}_i, \ldots, \hat{y}_{j}, \ldots, y_n \right) \\
        & \equiv \left( x_1, \ldots, x_n \right) \in \mathbb{R}^n
    \end{align*}
    
    \noindent where
    
    \begin{equation}
        \vec{f}_{\left(ij\right)}^{-1} \left(r_{(ij)}, R_{(ij)}, y_1, \ldots, \hat{y}_i, \ldots, \hat{y}_{j}, \ldots, y_n \right) = 
        \begin{cases}
            x_i = R_{(ij)} + \frac{m_j}{m_i + m_j} r_{(ij)} \\
            x_j = R_{(ij)} - \frac{m_i}{m_i + m_j} r_{(ij)} \\
            x_k = y_k, \quad k \neq i,j
        \end{cases}.
    \end{equation} $\hfill \square$
\end{remark}

\begin{definition}
    \justifying
    Given (\ref{Coordinate Transformation}), by introducing the Hilbert space
    \begin{equation*}
        \chi_{(ij)} \doteq L^2 \left( \mathbb{R}^n, \, dr_{(ij)}dR_{(ij)} dy_1 \cdots \widehat{dy_i} \cdots \widehat{dy_j} \cdots dy_n \right),
    \end{equation*}
    the (unitary) operator implementing (\ref{Coordinate Transformation}) is
    \begin{equation}\label{Coordinate transformation operator}
        U_{(ij)}: \, \psi \in L^2 \left( \mathbb{R}^n, dx_1 \cdots dx_i \cdots dx_j \cdots dx_n \right) \longmapsto U_{(ij)}\psi \doteq \psi \circ \vec{f}_{\left(ij\right)}^{-1} \in  \chi_{(ij)}.
    \end{equation}
    $\hfill \square$
\end{definition}

\begin{definition}
    \justifying
    Given $i,j \in \{1,\ldots,n\}: \, i<j, \; \epsilon > 0$, set $\underline{Y}_{(ij)} = \left( R_{(ij)}, y_1, \ldots, \hat{y}_i, \ldots, \hat{y}_j, \ldots, y_n \right) \in \mathbb{R}^{n-1}$, the scaling operator
    \begin{equation*}
        U_\epsilon^{(ij)}: \, \Tilde{\psi} \in \chi_{(ij)} \longmapsto U_\epsilon^{(ij)}\tilde{\psi} \in \chi_{(ij)} 
    \end{equation*}
    is introduced, where, should $\tilde{\psi}$ be continuous, $\left( U_\epsilon^{(ij)}\tilde{\psi} \right)\left(r_{(ij)}, \underline{Y}_{(ij)}\right) \doteq \sqrt{\epsilon} \Tilde{\psi} \left( \epsilon r_{(ij)}, \, \underline{Y}_{(ij)} \right)$. $\hfill \square$
\end{definition}

\begin{remark}
    \justifying
    By introducing the Hilbert space
    \begin{equation*}
        \chi_{(ij)}^{(red)} \doteq L^2\left( \mathbb{R}^{n-1}, \, dR_{(ij)} dy_1 \cdots \widehat{dy_i} \cdots \widehat{dy_j} \cdots dy_n \right),
    \end{equation*}
    \noindent on the one hand, $\chi_{(ij)} = L^2\left( \mathbb{R}, dr_{(ij)} \right) \otimes \chi_{(ij)}^{(red)}$, on the other hand, $U_{\epsilon}^{(ij)} \equiv u_{\epsilon}^{(ij)} \otimes \mathds{1}$, where
    \begin{equation}\label{Little scaling operator}
        u_{\epsilon}^{(ij)}: \, \varphi \in L^2\left( \mathbb{R},\ dr_{(ij)} \right) \longmapsto u_{\epsilon}^{(ij)}\varphi \in L^2\left( \mathbb{R}, dr_{(ij)} \right)
    \end{equation}
    \noindent and $\left(u_{\epsilon}^{(ij)}\varphi\right) \left(r_{(ij)}\right) = \sqrt{\epsilon}\varphi\left( \epsilon r_{(ij)} \right)$, should $\varphi$ be a continuous function. $u^{(ij)}_{\epsilon}$ is well-defined and unitary because of
    \begin{align*}
        \int_{\mathbb{R}} \, \abs{\left( u_{\epsilon}^{(ij)} \varphi \right) \left( r_{(ij)} \right)}^2 dr_{(ij)} \equiv \int_{\mathbb{R}} \, \epsilon \abs{\varphi \left(\epsilon r_{(ij)} \right)}^2 dr_{(ij)} = \left( \tilde{r}_{(ij)} = \epsilon r_{(ij)} \right) = \int_{\mathbb{R}} \, \abs{\varphi \left( \tilde{r}_{(ij)} \right)}^2 d\tilde{r}_{(ij)} \equiv \norm{\varphi}_2^2.
    \end{align*}
    $\hfill \square$
\end{remark}

\begin{definition}
    \justifying
    Let $v \in C_{0}^{\infty}\left( \mathbb{R} \right)$ be\footnote{The same letter will be used to denote the corresponding multiplication operator on $L^2\left( \mathbb{R},dr_{(ij)} \right)$.} even and such that $\int_{\mathbb{R}} v^2 = 1$. For all $i,j \in \{1,\ldots,n\}: \, i<j, \, \epsilon >0$, the bounded linear operator
    \begin{equation*}
        A_{\epsilon}^{(ij)} \doteq \left( v \otimes \mathds{1} \right) \frac{U_{\epsilon}^{(ij)}}{\sqrt{\epsilon}} U_{(ij)} \equiv \left(\frac{vu_{\epsilon}^{(ij)}}{\sqrt{\epsilon}} \otimes \mathds{1} \right) U_{(ij)}: \, L^2\left( \mathbb{R}^n, dx_1 \cdots dx_n \right) \longrightarrow \chi_{(ij)},
    \end{equation*}
    is introduced. $\hfill \square$
\end{definition}

\begin{remark}
    \justifying
    \begin{itemize}
        \item For all $i,j \in \{1,\ldots,n\}: \, i<j$, $\epsilon > 0$, $V_{\epsilon}^{(ij)} = A_{\epsilon}^{(ij) \, \ast} A_{\epsilon}^{(ij)}$ and
        \begin{equation}\label{Approximating Factorized Hamiltonian}
            H_{\epsilon} = H_0 - g \sum_{1 \leq i < j \leq n} \, V_{\epsilon}^{(ij)} \equiv H_0 - g \sum_{1 \leq i < j \leq n} \, A_{\epsilon}^{(ij) \, \ast} A_{\epsilon}^{(ij)}.
        \end{equation}
        \item From now on, interacting pairs $(ij)$ will be denoted by a greek index $\sigma, \nu, \ldots$, varying in $\mathcal{I}$; clearly, $\abs{\mathcal{I}} = \left( \begin{matrix}
            n\\
            2
        \end{matrix} \right) = $ number of interacting pairs.
    \end{itemize}
    $\hfill \square$
\end{remark}

\begin{definition}
    \justifying
    Let the Hilbert space $\chi = \underset{\sigma}{\bigoplus} \, \chi_{\sigma}$ be. Given $\epsilon > 0$, the bounded operator
    \begin{equation*}
        A_\epsilon: L^2\left( \mathbb{R}^n, \, dx_1 \cdots dx_n \right) \longrightarrow \chi
    \end{equation*}
    is defined, where
    \begin{equation*}
        A_{\epsilon} \psi \doteq \left( A_{\epsilon}^{\sigma_1} \psi, \ldots, A_{\epsilon}^{\sigma_{\abs{ \mathcal{I} }}} \psi \right)
    \end{equation*}
    for all $\psi \in L^2\left( \mathbb{R}^n, \, dx_1 \cdots dx_n \right)$. $\hfill \square$
\end{definition}

\begin{remark}
    \justifying
    For all $\epsilon > 0$, the foregoing definition allows for $g \sum_{\sigma} \, V_{\epsilon}^{\sigma}$ to be equal to $g A_{\epsilon}^{\ast} A_{\epsilon}$, hence $H_{\epsilon} = H_0 - gA_{\epsilon}^{\ast} A_{\epsilon}$. Therefore (see App. 1)
    \begin{equation}\label{epsilon Konno-Kuroda}
        \left( H_{\epsilon} - z \mathds{1} \right)^{-1} \equiv R_{H_{\epsilon}}\left( z \right) = R_{H_0}\left( z \right) + g \sum_{\sigma, \nu} \, \left[ A_{\epsilon}^{\sigma}R_{H_0}\left( \overline{z} \right) \right]^{\ast} \left[ \Lambda_{\epsilon}\left( z \right)^{-1} \right]_{\sigma \nu} \left[ A_{\epsilon}^{\nu}R_{H_0}\left( z \right) \right],
    \end{equation}
    for all $z \in \rho\left( H_{\epsilon} \right) \cap \rho \left( H_0 \right)$, where $\left[ \Lambda_{\epsilon}\left( z \right) \right]_{\sigma \nu} = \delta_{\sigma \nu} - g A_{\epsilon}^{\sigma} R_{H_0} \left( z \right) A_{\epsilon}^{\nu \, \ast} \in \mathfrak{B}\left(\chi_{\nu}, \, \chi_{\sigma} \right)$, $\sigma, \nu \in \mathcal{I}$. The entire analysis is based on the $\epsilon \downarrow 0$ behaviour of such a formula. $\hfill \square$
\end{remark}


\section{The Limit of $A_{\epsilon}^{\sigma} R_{H_0}\left( z \right), \, z<0 $}
\justifying

\begin{remark}
    \justifying
    Given $U_\sigma$ as in (\ref{Coordinate transformation operator}), $U_{\sigma} H_0 = H_0^{\sigma} U_{\sigma}$ on $\mathcal{D}_{H_0}$, where
    \begin{equation*}
        H_0^{\sigma} = -\frac{1}{2\mu_{\sigma}} \frac{\partial^2}{\partial r_{\sigma}^2} - \frac{1}{2M_{\sigma}} \frac{\partial^2}{\partial R_{\sigma}^2} - \sum_{\underset{k\neq i,j}{k=1}}^n \, \frac{1}{2m_k} \frac{\partial^2}{\partial x_k^2},
    \end{equation*}
    \noindent hence
    \begin{equation*}
        U_{\sigma} \left( H_0 - z\mathds{1} \right)^{-1} = \left( H_0^{\sigma} - z\mathds{1} \right)^{-1} U_{\sigma}, \quad z \in \rho\left( H_0 \right) \equiv \rho \left( H_0^{\sigma} \right)
    \end{equation*}
    \noindent and
    \begin{align*}
        A_{\epsilon}^{\sigma} \left(H_0 - z\mathds{1} \right)^{-1} & = \left( v \otimes \mathds{1} \right) \frac{U_{\epsilon}^{\sigma}}{\sqrt{\epsilon}} U_{\sigma} \left( H_0 - z\mathds{1} \right)^{-1} = \left( v \otimes \mathds{1} \right) \frac{U_{\epsilon}^{\sigma}}{\sqrt{\epsilon}} \left( H_0^{\sigma} - z\mathds{1} \right)^{-1} U_{\sigma} = \\
        & = T_{\epsilon}^{\sigma} \left( z \right) U_{\sigma}
    \end{align*}
    \noindent collecting both the $z$ and $\epsilon$ dependence. Moreover, since $A_{\epsilon}^{\sigma} R_{H_0}\left( z \right) \in \mathfrak{B}\left(L^2\left( \mathbb{R}^n\right), \, \chi_{\sigma} \right)$, $T_{\epsilon}^{\sigma} \left( z \right) \in \mathfrak{B}\left( \chi_{\sigma} \right)$ for all $z \in \rho\left( H_0 \right)$.  $\hfill \square$
\end{remark}

\begin{definition}
        Denoted by $\mathfrak{F}_{\underline{Y}_{\sigma}}$ the Fourier operator on $\chi_{\sigma}^{(red)}$, the bounded operator
        \begin{align*}
            \left( \mathds{1} \otimes \mathfrak{F}_{\underline{Y}_{\sigma}} \right) T_{\epsilon}^{\sigma} \left( z \right) \left( \mathds{1} \otimes \mathfrak{F}_{\underline{Y}_{\sigma}}^{-1} \right) \doteq T_{\epsilon, \, \underline{P}_{\sigma}}^{\sigma} \left(z \right): L^2\left( \mathbb{R}, dr_{\sigma} \right) \otimes \Tilde{\chi}_{\sigma}^{(red)} \rightarrow L^2\left( \mathbb{R}, dr_{\sigma} \right) \otimes \Tilde{\chi}_{\sigma}^{(red)}
        \end{align*}
    is introduced, where $\Tilde{\chi}_{\sigma}^{(red)} = \mathfrak{F}_{\underline{Y}_{\sigma}} \, \chi_{\sigma}^{(red)}$. $\hfill \square$
\end{definition}

\begin{remark}
    \justifying
    By definition,
        \begin{align*}
            T_{\epsilon, \, \underline{P}_{\sigma}}^{\sigma} \left(z \right) & \equiv \left( \frac{vu_{\epsilon}^{\sigma}}{\sqrt{\epsilon}} \otimes \mathds{1} \right) \left( \mathds{1} \otimes \mathfrak{F}_{\underline{Y}_{\sigma}} \right) \left( H_0^{\sigma} - z\mathds{1} \right)^{-1} \left( \mathds{1} \otimes \mathfrak{F}_{\underline{Y}_{\sigma}}^{-1} \right) \equiv \\ & \equiv \left( 2\mu_{\sigma} \right) \left(\frac{vu_{\epsilon}^{\sigma}}{\sqrt{\epsilon}} \otimes \mathds{1} \right)  \left[ - \frac{\partial^2}{\partial r_{\sigma}^2} - \left( 2\mu_{\sigma} \right) \left( z - Q_{\sigma} \right) \mathds{1} \right]^{-1},
    \end{align*}
    \noindent where $Q_{\sigma} = \frac{P_\sigma^2}{2M_{\sigma}} + \sum_{\underset{k\neq i,j}{k=1}}^{n} \frac{p_k^2}{2m_k}$. In particular, for all $\psi \in L^2\left( \mathbb{R}, dr_{\sigma} \right) \otimes \tilde{\chi}_{\sigma}^{(red)}$,
    \begin{equation*}
        \left[ T_{\epsilon, \, \underline{P}_{\sigma}}^{\sigma} \left(z \right) \psi \right] \left( r_{\sigma}, \underline{P}_{\sigma} \right) = (2\mu_{\sigma})v(r_{\sigma}) \int_{\mathbb{R}} \, G^{(1)}_{\left( 2\mu_{\sigma} \right) \left( z - Q_{\sigma} \right)} \left( \epsilon r_{\sigma} - r^{\prime}_{\sigma}  \right) \psi\left(r_{\sigma}^{\prime}, \underline{P}_{\sigma} \right) dr_{\sigma}^{\prime},
    \end{equation*}
    \noindent i.e. on $\Tilde{\chi}_{\sigma}^{(red)}$, it behaves as a multiplication operator, while, on $L^2\left( \mathbb{R}, dr_{\sigma} \right)$, as an integral operator with kernel
    \begin{equation}
        \left( 2\mu_{\sigma} \right) v\left( r_{\sigma} \right) G^{(1)}_{ \left( 2\mu_{\sigma} \right) \left( z - Q_{\sigma} \right)} \left( \epsilon r_{\sigma} - r_{\sigma}^{\prime} \right).
    \end{equation}
    $\hfill \square$
\end{remark}

\begin{definition}
    Given $\sigma \in \mathcal{I}$, $z<0$, let $T_{0, \underline{P}_{\sigma}}^{\sigma} \left( z \right): \, \psi \in L^2\left( \mathbb{R}, dr_{\sigma} \right) \otimes \Tilde{\chi}_{\sigma}^{(red)} \longmapsto T_{0, \underline{P}_{\sigma}}^{\sigma} \left( z \right) \psi \in L^2\left( \mathbb{R}, dr_{\sigma} \right) \otimes \Tilde{\chi}_{\sigma}^{(red)}$ be such that
    \begin{equation*}
        \left[ T_{0, \, \underline{P}_{\sigma}}^{\sigma} \left(z \right) \psi \right] \left( r_{\sigma}, \underline{P}_{\sigma} \right) \doteq (2\mu_{\sigma})v(r_{\sigma}) \int_{\mathbb{R}} \, G^{(1)}_{ \left( 2\mu_{\sigma} \right) \left( z - Q_{\sigma} \right)} \left( - r^{\prime}_{\sigma}  \right) \psi\left(r_{\sigma}^{\prime}, \underline{P}_{\sigma} \right) dr_{\sigma}^{\prime},
    \end{equation*}
    with $Q_{\sigma}$ as above. Correspondingly
    \begin{equation}
        \left( \mathds{1} \otimes \mathfrak{F}_{\underline{Y}_{\sigma}}^{-1} \right) T_{0,\underline{P}_{\sigma}}^{\sigma} \left( z \right) \left( \mathds{1} \otimes \mathfrak{F}_{\underline{Y}_{\sigma}} \right) \doteq T_{0}^{\sigma}\left( z \right): \, L^2\left( \mathbb{R}, dr_{\sigma} \right) \otimes \chi_{\sigma}^{(red)} \longrightarrow L^2\left( \mathbb{R}, dr_{\sigma} \right) \otimes \chi_{\sigma}^{(red)}
    \end{equation}
    is introduced. $\hfill \square$
\end{definition}

\begin{lemma}\label{Limit of AR}
    Let $\sigma \in \mathcal{I}, \, z<0$ be. Then
    \begin{equation*}
        \underset{\epsilon \downarrow 0}{\lim} \, \norm{T_{\epsilon}^{\sigma} \left( z \right) - T_{0}^{\sigma} \left( z \right)}_{\mathfrak{B}\left( \chi_{\sigma} \right)} = 0.
    \end{equation*}
\end{lemma}

\begin{proof}
    Directly from \cite{05_GHL}, Lemma 3.1 and Proposition 3.2.
\end{proof}

\begin{remark}
    \justifying
    Direct consequence of Lemma \ref{Limit of AR}, for all $\sigma \in \mathcal{I}, \, z<0$, is
    \begin{equation*}
        \underset{\epsilon \downarrow 0}{\lim} \; A_{\epsilon}^{\sigma} \left( H_0 - z\mathds{1} \right)^{-1} = \underset{\epsilon \downarrow 0}{\lim} \; T_{\epsilon}^{\sigma} \left( z \right) U_{\sigma} = T_{0}^{\sigma}\left( z \right) U_{\sigma} \doteq S^{\sigma}\left( z \right),
    \end{equation*}
    with $S^{\sigma}\left( z \right) \in \mathfrak{B}\left(L^2\left( \mathbb{R}^n, dx_1 \cdots dx_n \right), \chi_{\sigma} \right)$. $\hfill \square$
\end{remark}

\section{$\Lambda_{\epsilon}\left( z \right)$-related analysis}
\justifying

\begin{remark}
    \justifying
    Given $z \in \rho\left( H_0 \right), \, \sigma, \nu \in \mathcal{I}$, 
    \begin{align*}
        \left[ \Lambda_{\epsilon}\left( z \right) \right]_{\sigma \nu} & = \delta_{\sigma \nu} - g A^{\sigma}_{\epsilon} R_{H_0} \left( z \right) A_{\epsilon}^{\nu \ast} 
        \\ & = \delta_{\sigma \nu} - g A^{\sigma}_{\epsilon} R_{H_0} \left( z \right) A_{\epsilon}^{\nu \ast} + \delta_{\sigma \nu} g A_{\epsilon}^{\sigma} R_{H_0} \left( z \right) A_{\epsilon}^{\nu \ast} - \delta_{\sigma \nu} g A_{\epsilon}^{\sigma} R_{H_0} \left( z \right) A_{\epsilon}^{\nu \ast} \\
        & = \left[\mathds{1} - g A_{\epsilon}^{\sigma} R_{H_0}\left( z \right) A_{\epsilon}^{\sigma \ast} \right] \delta_{\sigma \nu} + \left( \delta_{\sigma \nu} - 1 \right) g A_{\epsilon}^{\sigma} R_{H_0} \left( z \right) A_{\epsilon}^{\nu \ast} \\
        & \equiv \left[ \Lambda_{\epsilon}\left( z \right)_{\text{diag}} \right]_{\sigma \nu} + \left[ \Lambda_{\epsilon}\left( z \right)_{\text{off}} \right]_{\sigma \nu}.
    \end{align*}
    \noindent Where it all makes sense,
    \begin{equation}\label{Inversion formula Lambda-Matrix}
        \left[ \Lambda_{\epsilon}\left( z \right) \right]^{-1} = \left\{ \mathds{1} + \left[ \Lambda_{\epsilon}\left( z \right)_{\text{diag}} \right]^{-1} \left[ \Lambda_{\epsilon}\left( z \right)_{\text{off}} \right] \right\}^{-1} \left[ \Lambda_{\epsilon} \left( z \right)_{\text{diag}} \right]^{-1},
    \end{equation}
    \noindent hence, aim of the section is finding a range of values for $z \in \rho\left(H_0\right)$ such that (\ref{Inversion formula Lambda-Matrix}) holds. $\hfill \square$
\end{remark}

\subsection{$\Lambda_{\epsilon}\left( z \right)_{\text{diag}}$}
\justifying

\begin{remark}
    \justifying
    Let $\sigma \in \mathcal{I}, \, \epsilon >0$ be. Set $\phi_{\epsilon}^{\sigma} \left( z \right) = g A_{\epsilon}^{\sigma} R_{H_0} \left( z \right) A_{\epsilon}^{\sigma \ast} \in \mathfrak{B}\left(\chi_{\sigma}\right)$, $z \in \rho\left( H_{0} \right)$,
    \begin{align*}
        \phi_{\epsilon}^{\sigma}\left( z \right) & = g \left( v \otimes \mathds{1} \right) \frac{U_\epsilon^{\sigma}}{\sqrt{\epsilon}} U_{\sigma} \left( H_0 - z \mathds{1} \right)^{-1} \left[ \left( v \otimes \mathds{1} \right) \frac{U_\epsilon^{\sigma}}{\sqrt{\epsilon}} U_{\sigma} \right]^{\ast} \\
        & = g \left( v \otimes \mathds{1} \right) \frac{U_\epsilon^{\sigma}}{\sqrt{\epsilon}} \left( H_0^{\sigma} - z \mathds{1} \right)^{-1} \frac{U_{\epsilon}^{\sigma \, \ast}}{\sqrt{\epsilon}} \left( v \otimes \mathds{1} \right).
    \end{align*}
    \noindent Moreover
    \begin{equation*}
        \left[ -\frac{1}{2\mu_{\sigma}} \frac{\partial^2}{\partial r_{\sigma}^2} - \left( z - Q_{\sigma} \right) \mathds{1} \right]^{-1} = \left( \mathds{1} \otimes \mathfrak{F}_{\underline{Y}_{\sigma}} \right) \left( H_{0}^{\sigma} - z \mathds{1} \right)^{-1} \left( \mathds{1} \otimes \mathfrak{F}_{\underline{Y}_{\sigma}}^{-1} \right)
    \end{equation*}
    \noindent allows for
    \begin{align*}
        \phi_{\epsilon}^{\sigma}\left( z \right) & = \frac{g}{\epsilon} \left\{ \left( vu_{\epsilon}^{\sigma} \otimes \mathds{1} \right) \left( \mathds{1} \otimes \mathfrak{F}_{\underline{Y}_{\sigma}}^{-1} \right) \left[ -\frac{1}{2\mu_{\sigma}} \frac{\partial^2}{\partial r_{\sigma}^2} - \left( z - Q_{\sigma} \right) \mathds{1} \right]^{-1} \left( \mathds{1} \otimes \mathfrak{F}_{\underline{Y}_{\sigma}} \right) \left( u_{\epsilon}^{\sigma \ast}v \otimes \mathds{1} \right) \right\} \\
        & = \left( \mathds{1} \otimes \mathfrak{F}_{\underline{Y}_{\sigma}}^{-1} \right) \left\{ g \left( v \otimes \mathds{1} \right) \frac{U_{\epsilon}^{\sigma}}{\sqrt{\epsilon}} \left[ -\frac{1}{2\mu_{\sigma}} \frac{\partial^2}{\partial r_{\sigma}^2} - \left( z - Q_{\sigma} \right) \mathds{1} \right]^{-1} \frac{U_{\epsilon}^{\sigma \ast}}{\sqrt{\epsilon}} \left( v \otimes \mathds{1} \right) \right\} \left( \mathds{1} \otimes \mathfrak{F}_{\underline{Y}_{\sigma}} \right).
    \end{align*}
    $\hfill \square$
\end{remark}

\begin{definition}
    \justifying
    Fixed $\sigma \in \mathcal{I}, \, \epsilon > 0, \, z \in \rho\left(H_0\right)$, the linear bounded operator $ \phi^{\sigma}_{\epsilon, \, \underline{P}_{\sigma}}\left(z\right)$ on $L^2\left( \mathbb{R}, dr_{\sigma} \right) \otimes \Tilde{\chi}_{\sigma}^{(\text{red})}$, where
    \begin{equation*}
        \phi^{\sigma}_{\epsilon, \, \underline{P}_{\sigma}}\left(z\right) \doteq g \left( v \otimes \mathds{1} \right) \frac{U_{\epsilon}^{\sigma}}{\sqrt{\epsilon}} \left[ -\frac{1}{2\mu_{\sigma}} \frac{\partial^2}{\partial r_{\sigma}^2} - \left( z - Q_{\sigma} \right) \mathds{1} \right]^{-1} \frac{U_{\epsilon}^{\sigma \ast}}{\sqrt{\epsilon}} \left( v \otimes \mathds{1} \right)
    \end{equation*}
    and $Q_{\sigma}$ as above, is introduced. $\hfill \square$
\end{definition}

\begin{remark}
    \justifying
    By observing that, for all $\psi, \varphi \in L^2\left( \mathbb{R}, dr_{\sigma} \right), \, \epsilon >0$,
    \begin{align*}
        \langle \varphi, u_{\epsilon}^{\sigma}  \psi \rangle & = \int_{\mathbb{R}} \, \overline{\varphi}\left( x \right) \left( u_{\epsilon}^{\sigma} \psi \right) \left( x \right) dx = \int_{\mathbb{R}} \, \overline{\varphi}\left( x \right) \sqrt{\epsilon} \psi \left( \epsilon x \right) dx = \int_{\mathbb{R}} \, \overline{\left[ \frac{1}{\sqrt{\epsilon}} \varphi\left( \frac{x^\prime}{\epsilon} \right) \right]} \psi \left( x^\prime \right) dx^{\prime} \\
        & = \langle u_{\epsilon}^{\sigma \ast} \varphi, \psi \rangle,
    \end{align*}
    
    \noindent for all $z<0$, $\varphi \in L^2\left( \mathbb{R}, dr_{\sigma} \right) \otimes \Tilde{\chi}_{\sigma}^{(\text{red})}$,
    
    \begin{align*}
        \left\{\left[ - \frac{\partial^2}{\partial r_{\sigma}^2} - \left( 2\mu_{\sigma} \right) \left( z - Q_{\sigma} \right) \mathds{1} \right]^{-1} \frac{u_{\epsilon}^{\sigma \ast}}{\sqrt{\epsilon}} \varphi\right\}(r_{\sigma}, \underline{P}_{\sigma}) & = \int_{\mathbb{R}} \, G^{(1)}_{\left[(2\mu_{\sigma})\left(z - Q_{\sigma}\right) \right]} \left( r_{\sigma}, \, r^{\prime}_{\sigma} \right) \frac{1}{\epsilon} \varphi \left( \frac{r_{\sigma}^{\prime}}{\epsilon},\underline{P}_{\sigma} \right) dr^{\prime}_{\sigma} = \\
        & = \int_{\mathbb{R}} \, G^{(1)}_{\left[(2\mu_{\sigma})\left(z - Q_{\sigma}\right) \right]} \left( r_{\sigma}, \, \epsilon r^{\prime}_{\sigma} \right) \varphi \left( r_{\sigma}^{\prime},\underline{P}_{\sigma} \right) dr_{\sigma}^{\prime}.
    \end{align*}
    
    \noindent Consequently
    
    \begin{align*}
        & \left[ \left\{ \frac{U_\epsilon}{\sqrt{\epsilon}} \left[ -\frac{\partial^2}{\partial r_{\sigma}^2} - \left( 2\mu_{\sigma} \right) \left( z - Q_{\sigma} \right) \mathds{1} \right]^{-1} \frac{U_{\epsilon}^{\ast}}{\sqrt{\epsilon}} \right\} \varphi \right] \left( r_{\sigma}, \underline{P}_{\sigma} \right) = \int_{\mathbb{R}} \, G^{(1)}_{\left[(2\mu_{\sigma})\left(z - Q_{\sigma}\right) \right]} \left( \epsilon r_{\sigma}, \, \epsilon r^{\prime}_{\sigma} \right) \varphi \left( r_{\sigma}^{\prime}, \underline{P}_{\sigma} \right) dr_{\sigma}^{\prime} = \\
        & = \int_{\mathbb{R}} \, \left[ \int_{0}^{\infty} \, \frac{e^{-\frac{\abs{\epsilon r_{\sigma} - \epsilon r_{\sigma}^{\prime}}^{2}}{4t} + \left(2\mu_{\sigma}\right) \left( z - Q_{\sigma} \right) t }}{ \sqrt{4 \pi t}} dt \right] \varphi \left( r_{\sigma}^{\prime}, \underline{P}_{\sigma}\right) dr^{\prime}_{\sigma} = \epsilon \int_{\mathbb{R}} \, G^{(1)}_{ \left[ \epsilon^2 \left( 2\mu_{\sigma} \right) \left( z - Q_{\sigma} \right) \right]} \left( r_{\sigma}, r^{\prime}_{\sigma} \right) \varphi \left( r^{\prime}_{\sigma}, \underline{P}_{\sigma} \right) dr^{\prime}_{\sigma} = \\
        & = \epsilon \left\{ \left[ -\frac{\partial^2}{\partial r_{\sigma}^2} - \epsilon^2 \left( 2\mu_{\sigma} \right) \left( z - Q_{\sigma} \right) \mathds{1} \right]^{-1} \varphi \right\} \left( r_{\sigma}, \underline{P}_{\sigma} \right),
    \end{align*}
    
    \noindent i.e.
    
    \begin{equation*}
        \frac{U_{\epsilon}^{\sigma}}{\sqrt{\epsilon}} \left[ -\frac{\partial^2}{\partial r_{\sigma}^2} - \left( 2\mu_{\sigma} \right) \left( z - Q_{\sigma} \right) \mathds{1} \right]^{-1} \frac{U_{\epsilon}^{\sigma \ast}}{\sqrt{\epsilon}} = \epsilon \left[ - \frac{\partial^2}{\partial r_{\sigma}^2} -\epsilon^2 \left( 2\mu_{\sigma} \right) \left( z - Q_{\sigma} \right) \mathds{1} \right]^{-1},
    \end{equation*}
    
    \noindent allowing to state that
    
    \begin{align*}
        \phi_{\epsilon, \underline{P}_{\sigma}}^{\sigma} \left( z \right) & = g \left( v\otimes \mathds{1} \right) \frac{U_{\epsilon}^{\sigma}}{\sqrt{\epsilon}} \left( 2\mu_{\sigma} \right) \left[ -\frac{\partial^2}{\partial r_{\sigma}^2} - \left( 2\mu_{\sigma}\right) \left( z - Q_{\sigma} \right) \mathds{1} \right]^{-1} \frac{U_{\epsilon}^{\sigma \ast}}{\sqrt{\epsilon}} \left( v \otimes \mathds{1} \right) \equiv \\
        & \equiv  \left( v\otimes \mathds{1} \right) \left\{\left( 2\mu_{\sigma} \right) g \epsilon \left[ -\frac{\partial^2}{\partial r_{\sigma}^2} - \epsilon^2 \left( 2 \mu_{\sigma} \right) \left( z - Q_{\sigma} \right) \mathds{1} \right]^{-1} \right\} \left( v\otimes \mathds{1} \right).
    \end{align*}
    $\hfill \square$
\end{remark}

\vfill 
\newpage
\begin{theorem}
    \justifying
    For all $\sigma \in \mathcal{I}, \, \epsilon > 0, \, z<0$, $\norm{\phi_{\epsilon, \, \underline{P}_{\sigma}}^{\sigma} \left( z \right)}_{\mathfrak{B}\left( \chi_{\sigma} \right)} \leq \frac{\mathfrak{C} \abs{g}}{ \sqrt{ \abs{z} }}$, where $\mathfrak{C} = \sqrt{\left( \underset{\sigma}{\max} \, \frac{\mu_{\sigma}}{2} \right)}$.
\end{theorem}

\begin{proof}
    \justifying
    Given $\eta \in L^2\left( \mathbb{R}, dr_{\sigma}\right), \, \xi \in \Tilde{\chi}_{\sigma}^{(\text{red})}$ arbitrary.
    \small{\begin{align*}
        & \norm{\phi_{\epsilon, \, \underline{P}_{\sigma}}^{\sigma} \left( z \right)\eta \otimes \xi}_{2}^2 = \int_{\mathbb{R}^n} \, \abs{ \left( g \sqrt{ \frac{\mu_{\sigma}}{2}} \right) \frac{ \xi\left( \underline{P}_{\sigma} \right) }{ \sqrt{ \abs{z - Q_{\sigma}} } } v\left( r_{\sigma} \right) \int_{\mathbb{R}} \, e^{ - \epsilon \sqrt{ \left( 2\mu_{\sigma} \right) \abs{ z - Q_{\sigma}} } \abs{r_{\sigma} - r^{\prime}_{\sigma} } } v\left( r^{\prime}_{\sigma} \right) \eta\left( r^{\prime}_{\sigma} \right) dr^{\prime}_{\sigma} }^2 dr_{\sigma} d\underline{P}_{\sigma} \leq \\
        & \leq \frac{g^2 \mu_{\sigma}}{2} \, \int_{\mathbb{R}^n} \, \frac{ \abs{\xi \left( \underline{P}_{\sigma} \right)}^2 }{ \abs{z - Q_{\sigma}} } v\left( r_{\sigma} \right)^2 \left[ \int_{\mathbb{R}} \, e^{ - \epsilon \sqrt{ \left( 2\mu_{\sigma} \right) \abs{ z - Q_{\sigma}} } \abs{r_{\sigma} - r^{\prime}_{\sigma} } } v \left(r^{\prime}_{\sigma} \right) \abs{\eta \left( r^{\prime}_{\sigma} \right) } dr^{\prime}_{\sigma} \right]^{2} dr_{\sigma} d\underline{P}_{\sigma} \leq \\
        & \leq \frac{g^2 \mu_{\sigma}}{2} \, \int_{\mathbb{R}^n} \, \frac{ \abs{\xi \left( \underline{P}_{\sigma} \right)}^2 }{ \abs{z - Q_{\sigma}} } v\left( r_{\sigma} \right)^2 \left[ \int_{\mathbb{R}} \, e^{ - 2\epsilon \sqrt{\left(2 \mu_{\sigma} \right) \abs{z - Q_{\sigma}}} \abs{r_{\sigma} - r^{\prime}_{\sigma}}} v\left( r^{\prime}_{\sigma} \right)^2 dr^{\prime}_{\sigma} \right] \left[ \int_{\mathbb{R}} \, \abs{\eta\left( r^{\prime}_{\sigma} \right)}^2 dr^{\prime}_{\sigma} \right] dr_{\sigma} d\underline{P}_{\sigma} \leq \\
        & \leq \frac{g^2 \mu_{\sigma}}{2} \, \left[ \int_{\mathbb{R}^{n-1}} \, \frac{ \abs{\xi \left( \underline{P}_{\sigma} \right)}^2 }{ \abs{z - Q_{\sigma}} } d\underline{P}_{\sigma} \right] \left[ \int_{\mathbb{R}} \, \abs{\eta\left( r^{\prime}_{\sigma} \right)}^2 dr^{\prime}_{\sigma} \right] \left[ \int_{\mathbb{R}^2} \, v\left(r_{\sigma}\right)^2 v\left( r^{\prime}_{\sigma} \right)^2 e^{-2\epsilon \sqrt{\left(2\mu_{\sigma}\right) \abs{z}} \abs{r_{\sigma} - r^{\prime}_{\sigma}} } dr_{\sigma} dr^{\prime}_{\sigma} \right] \leq  \\ 
        & \leq \frac{g^2 \mu_{\sigma}}{2 \abs{z} } \norm{\eta \otimes \xi}_2^2 \leq \left( \underset{\sigma}{\max} \, \mu_{\sigma} \right) \frac{g^2}{2\abs{z}} \, \norm{\eta \otimes \xi}_2^2,
    \end{align*}}
    \normalsize by using H{\"o}lder inequality in passing from the second to the third line. Therefore
    \begin{equation*}
        \norm{\phi_{\epsilon, \, \underline{P}_{\sigma}}^{\sigma} \left( z \right)}_{\mathfrak{B}\left( \chi_{\sigma} \right)} \leq \sqrt{\left( \underset{\sigma}{\max} \, \frac{\mu_{\sigma}}{2} \right) \frac{g^2}{\abs{z}}} \equiv \frac{\mathfrak{C} \abs{g}}{ \sqrt{ \abs{z} } }.
    \end{equation*}
\end{proof}

\begin{corollary}
    \justifying
    For all $\sigma \in \mathcal{I}, \, \epsilon > 0$, if $z \in \mathbb{R}^{-}: \, z < - \mathfrak{C}^2g^2$, $\Lambda_{\epsilon}\left(z\right)_{\text{diag}}$ is invertible on $\chi_{\sigma}$.
\end{corollary}

\begin{proof}
    \justifying
    By definition
    \begin{equation*}
        \left[ \Lambda_{\epsilon} \left( z \right)_{\text{diag}}\right]_{\sigma \sigma} = \mathds{1} - \phi_{\epsilon}^{\sigma} \left( z \right),
    \end{equation*}
    hence $\norm{\phi_{\epsilon}^{\sigma}\left( z \right)}_{\mathfrak{B}\left( \chi_{\sigma} \right)} < 1$ guarantees the invertibility of $\mathds{1} - \phi_{\epsilon}^{\sigma}\left( z \right)$ on $\chi_{\sigma}$.
    \begin{equation*}
        z < - \mathfrak{C}^2g^2 \implies \frac{\mathfrak{C}\abs{g}}{\sqrt{\abs{z}}} < 1 \implies \norm{\phi_{\epsilon}^{\sigma}\left( z \right)}_{\mathfrak{B}\left( \chi_{\sigma} \right)} < 1. 
    \end{equation*}
\end{proof}

\begin{remark}
    Given $\sigma \in \mathcal{I}$ and $z < - \mathfrak{C}^2g^2$, what above allows to state that
    \begin{equation*}
        \left[ \left( \Lambda_{\epsilon} \left(z\right)_{\text{diag}} \right)^{-1} \right]_{\sigma \sigma} = \left[ \mathds{1} -\phi_{\epsilon}^{\, \sigma} \left(z\right) \right]^{-1} = \sum_{n \in \mathbb{N}_0} \, \left[\phi_{\epsilon}^{\, \sigma} \left(z\right)\right]^n,
    \end{equation*}
    \noindent hence 
    \begin{equation*}
        \norm{ \left[ \left( \Lambda_{\epsilon}\left(z\right)_{\text{diag}} \right)^{-1} \right]_{\sigma \sigma} }_{\mathfrak{B}\left( \chi_{\sigma} \right)} \leq  \sum_{n \in \mathbb{N}_0} \, \norm{\phi_{\epsilon}^{\sigma} \left(z\right)}_{\mathfrak{B}\left( \chi_{\sigma} \right)}^{n} = \frac{1}{1 - \norm{\phi_{\epsilon}^{\sigma} \left(z\right)}_{\mathfrak{B}\left( \chi_{\sigma} \right)}} \leq \left[ 1 - \frac{\mathfrak{C}\abs{g}}{\sqrt{\abs{z}}} \right]^{-1}.
    \end{equation*}
    $\hfill \square$
\end{remark}

\subsection{Investigating $\Lambda_{\epsilon}\left( z \right)_{\text{diag}}$ as $\epsilon \rightarrow 0^+$}

\begin{definition}
    Given $\sigma \in \mathcal{I}$ and $z<0$, the linear operator
    \begin{equation*}
        \phi_{\, 0, \, \underline{P}_{\sigma}}^{\sigma} \left( z \right): \, \varphi \in \mathcal{D} \subseteq  L^2\left( \mathbb{R}, dr_{\sigma} \right) \otimes \Tilde{\chi}_{\sigma}^{\left( \text{red} \right)} \longmapsto \phi_{\, 0, \, \underline{P}_{\sigma}}^{\sigma} \left( z \right) \varphi \in L^2\left( \mathbb{R}, dr_{\sigma} \right) \otimes \Tilde{\chi}_{\sigma}^{\left( \text{red} \right)},
    \end{equation*}
    where
    \begin{equation*}
        \left[\phi_{\, 0, \, \underline{P}_{\sigma}}^{\sigma} \left( z \right) \varphi\right] \left(r_{\sigma}, \underline{P}_{\sigma} \right) \doteq \left( g \sqrt{\frac{\mu_{\sigma}}{2}} \right)  v\left(r_{\sigma}\right) \int_{\mathbb{R}} \, \frac{\varphi\left(r_{\sigma}^{\prime}, \underline{P}_{\sigma} \right)}{\sqrt{ \abs{z - Q_{\sigma}}}}  v\left(r^{\prime}_{\sigma}\right) dr_{\sigma}^{\prime}, \quad \varphi \in \mathcal{D},
    \end{equation*}
    is introduced. $\hfill \square$
\end{definition}

\begin{lemma}
    \justifying
    For arbitrary $\sigma \in \mathcal{I}, \, z<0$, $\mathcal{D} \equiv L^2\left( \mathbb{R}, dr_{\sigma} \right) \otimes \Tilde{\chi}_{\sigma}^{\left( \text{red} \right)}$, $\phi_{\, 0, \, \underline{P}_{\sigma}}^{\sigma} \left(z\right)$ is bounded and
    \begin{equation*}
        \norm{ \phi_{\epsilon, \, \underline{P}_{\sigma}}^{\sigma} \left(z\right) - \phi_{\, 0, \, \underline{P}_{\sigma}}^{\sigma} \left(z\right) }_{\mathfrak{B}\left( \chi_{\sigma} \right)} \underset{\epsilon \downarrow 0}{ \longrightarrow} 0.
    \end{equation*}
\end{lemma}

\begin{proof}
    Let whatever $\eta \in L^2\left( \mathbb{R}, dr_{\sigma} \right), \, \xi \in \Tilde{\chi}_{\sigma}^{ \left(\text{red} \right) }$ be.
    \small{\begin{align*}
        & \norm{ \left[ \phi_{\,0, \, \underline{P}_{\sigma} }^{\sigma} \left(z\right) - \phi_{\epsilon, \, \underline{P}_{\sigma} }^{\sigma}\left(z\right) \right] \eta \otimes \xi }_{2}^{2} \equiv \int_{\mathbb{R}^n} \, \abs{ \left\{ \left[ \phi_{\,0, \, \underline{P}_{\sigma} }^{\sigma}\left(z\right) - \phi_{\epsilon, \, \underline{P}_{\sigma} }^{\sigma} \left(z\right) \right] \eta \otimes \xi \right\} \left(r_{\sigma}, \underline{P}_{\sigma} \right) }^2 dr_{\sigma}d\underline{P}_{\sigma} \\ 
        & = \int_{\mathbb{R}^n} \, \abs{ \left(g\sqrt{\frac{\mu_{\sigma}}{2}}\right) \frac{\xi \left( \underline{P}_{\sigma}\right)}{\sqrt{\abs{z - Q_{\sigma}}}} \int_{\mathbb{R}} \, v\left(r_{\sigma} \right) \left[ e^{-\epsilon \sqrt{\left( 2\mu_{\sigma} \right) \abs{z - Q_{\sigma}} } \abs{r_{\sigma} - r^{\prime}_{\sigma}} } - 1 \right] v\left(r^{\prime}_{\sigma} \right) \eta \left(r^{\prime}_{\sigma} \right) dr^{\prime}_{\sigma} }^2 dr_{\sigma} d\underline{P}_{\sigma} \\
        & \leq g^2 \left( \frac{\mu_{\sigma}}{2} \right) \int_{\mathbb{R}^n} \, \frac{\abs{\xi \left( \underline{P}_{\sigma} \right)}^2}{\abs{ z - Q_{\sigma}}} \left[ \int_{\mathbb{R}} \, \abs{e^{-\epsilon \sqrt{\left(2\mu_{\sigma}\right) \abs{z - Q_{\sigma}} } \abs{r_{\sigma} - r^{\prime}_{\sigma}} } - 1} v\left(r_{\sigma}\right) v\left(r^{\prime}_{\sigma} \right) \abs{\eta\left(r^{\prime}_{\sigma} \right)} dr^{\prime}_{\sigma} \right]^2 dr_{\sigma} d\underline{P}_{\sigma} \\
        & \leq \left( \frac{g^2\mu_{\sigma}}{2} \right) \int_{\mathbb{R}^n} \, \frac{\abs{\xi\left( \underline{P}_{\sigma} \right)}^2}{\abs{z-Q_{\sigma}}} \left\{ \left[ \int_{\mathbb{R}} \, \abs{e^{-\epsilon \sqrt{\left(2\mu_{\sigma}\right) \abs{z - Q_{\sigma}} } \abs{r_{\sigma} - r^{\prime}_{\sigma}} } - 1}^2 v\left(r_{\sigma}\right)^2 v\left(r^{\prime}_{\sigma} \right)^2 dr^{\prime}_{\sigma} \right] \left[ \int_{\mathbb{R}} \, \abs{\eta \left( r^{\prime}_{\sigma} \right)}^2 dr^{\prime}_{\sigma} \right]  \right\} dr_{\sigma} d\underline{P}_{\sigma} \\
        & \leq \left( \frac{g^2 \mu_{\sigma}}{2} \right) \left( \int_{\mathbb{R}} \, \abs{\eta \left( r^{\prime}_{\sigma} \right)}^2 dr^{\prime}_{\sigma} \right)\int_{\mathbb{R}^{n-1}} \, \left\{ \frac{\abs{\xi\left( \underline{P}_{\sigma}\right)}^2}{ \abs{z-Q_{\sigma}} } \left[\int_{\mathbb{R}^2} \, \abs{\epsilon \sqrt{\left(2\mu_{\sigma}\right) \abs{ z - Q_{\sigma}} } \abs{r_{\sigma} - r^{\prime}_{\sigma}}}^2 v\left(r_{\sigma}\right)^2 v\left(r^{\prime}_{\sigma} \right)^2 dr_{\sigma} dr^{\prime}_{\sigma} \right] d\underline{P}_{\sigma} \right\} \\
        & \leq \epsilon^2 \left(g^2 \mu_{\sigma}^{2}\right) \left[2 \int_{\mathbb{R}^2} \, v\left(r_{\sigma}\right)^2 v\left(r^{\prime}_{\sigma} \right)^2 \left(r_{\sigma}^2 + r^{\prime \, 2}_{\sigma} \right) dr_{\sigma} dr^{\prime}_{\sigma} \right] \norm{\eta \otimes \xi}_2^2,
    \end{align*}}
    
    \noindent \normalsize by using H{\"o}lder inequality in passing from the third to the fourth line. Since
    \begin{equation*}
        2 \int_{\mathbb{R}^2} \, v\left(r_{\sigma}\right)^2 v\left(r^{\prime}_{\sigma} \right)^2 \left(r_{\sigma}^2 + r^{\prime \, 2}_{\sigma} \right) dr_{\sigma} dr^{\prime}_{\sigma} < 4 \norm{V}_1^2 \left( \underset{\text{supp} \, v}{\sup} \, r_{\sigma}^2 \right) < \infty,
    \end{equation*}
    \noindent it results
    \begin{equation*}
        \norm{ \phi_{\,0, \, \underline{P}_{\sigma} }^{\sigma}\left( z \right) - \phi_{\epsilon, \, \underline{P}_{\sigma} }^{\sigma}\left( z \right) }_{\mathfrak{B}\left( \chi_{\sigma} \right)} \leq K \epsilon,
    \end{equation*}
    for some constant $K>0$.
\end{proof}

\vfill
\newpage

\begin{lemma}
    Given $z<0$ and $\phi_{\,0}^{\, \sigma} \left(z\right) \doteq \left( \mathds{1} \otimes \mathfrak{F}_{\underline{Y}_{\sigma}}^{-1} \right) \phi_{\,0, \, \underline{P}_{\sigma}}^{\, \sigma} \left(z\right) \left( \mathds{1} \otimes \mathfrak{F}_{\underline{Y}_{\sigma}} \right)$,  if $z < - \mathfrak{C}^2 g^2$, $\mathds{1} - \phi_{\,0}^{\sigma} \left( z \right)$ is invertible and
    \begin{equation*}
        \underset{\epsilon \downarrow 0}{ \lim} \, \left[ \mathds{1} - \phi_{\epsilon}^{\sigma}\left(z\right) \right]^{-1} = \left[ \mathds{1} -\phi_{\,0}^{\sigma}\left( z \right) \right]^{-1}.
    \end{equation*}
\end{lemma}

\begin{proof}
    Let $\eta \in L^2\left( \mathbb{R}, dr_{\sigma}\right), \, \xi \in \Tilde{\chi}_{\sigma}^{(\text{red})}$ be arbitrary.
    \begin{align*}
        & \norm{\phi_{\,0, \, \underline{P}_{\sigma}}^{\sigma} \eta \otimes \xi}_2^2 = \int_{\mathbb{R}^n} \, \abs{\left(g \sqrt{\frac{\mu_{\sigma}}{2}}\right) \frac{v\left(r_{\sigma}\right)}{\sqrt{\abs{z - Q_{\sigma}}}} \xi\left( \underline{P}_{\sigma}\right) \int_{\mathbb{R}} \, v\left(r^{\prime}_{\sigma}\right) \eta\left(r^{\prime}_{\sigma}\right) dr^{\prime}_{\sigma}}^2 dr_{\sigma} d\underline{P}_{\sigma} \\
        & \leq \left(g^2 \frac{\mu_{\sigma}}{2}\right) \int_{\mathbb{R}^n} \, \frac{\abs{\xi \left( \underline{P}_{\sigma} \right)}^2}{\abs{z - Q_{\sigma}}} v^2\left(r_{\sigma}\right) \left[ \int_{\mathbb{R}} \, \abs{\eta\left(r^{\prime}_{\sigma}\right)}  v\left(r^{\prime}_{\sigma}\right) dr^{\prime}_{\sigma} \right]^2 dr_{\sigma}d\underline{P}_{\sigma} \leq \left( g^2 \frac{\mu_{\sigma}}{2 \abs{z}} \right) \norm{\eta \otimes \xi}_2^2.
    \end{align*}
    As a consequence
    \begin{equation*}
        \norm{\phi_{\,0}^{\sigma}\left(z\right)}_{\mathfrak{B}\left( \chi_{\sigma} \right)} \equiv \norm{\phi_{\,0, \, \underline{P}_{\sigma}}^{\sigma} \left( z \right)}_{\mathfrak{B}\left( \chi_{\sigma} \right)} \leq \frac{\mathfrak{C}\abs{g}}{\sqrt{\abs{z}}}.
    \end{equation*}
\end{proof}

\begin{remark}
    Direct consequence of what above is that, as $z < - \mathfrak{C}^2 g^2$,
    \begin{align}
        \underset{\epsilon \downarrow 0}{\lim} \, \left[ \left(\Lambda_{\epsilon} \left(z\right)_{\text{diag}} \right)_{\sigma \sigma} \right]^{-1} & = \underset{\epsilon \downarrow 0}{\lim} \,  \left[ \mathds{1} - \phi_{\epsilon}^{\sigma} \left( z \right) \right]^{-1} = \left[ \mathds{1} - \phi_{0}^{\sigma} \left(z\right) \right]^{-1} = \\
        & = \left[ \left( \Lambda_{0}\left(z\right)_{\text{diag}} \right)_{\sigma \sigma} \right]^{-1} = \left[\left( \Lambda_{0} \left(z\right)_{\text{diag}} \right)^{-1}\right]_{\sigma \sigma}.
    \end{align}
    $\hfill \square$
\end{remark}

\subsection{$\left\{ \mathds{1} + \left[ \Lambda_{\epsilon}\left(z\right)_{\text{diag}}  \right]^{-1} \Lambda_{\epsilon}\left(z\right)_{\text{off}} \right\}^{-1}$-related investigations}

\begin{theorem}
    Let $\epsilon > 0, \, z < 0$ be arbitrary. (\ref{Norm estimate 1 Off-Diagonal Lambda}) and (\ref{Norm estimate 2 Off-Diagonal Lambda}) hold.
\end{theorem}

\begin{proof}
    First of all, let $\sigma, \nu \in \mathcal{I}: \, \sigma \neq \nu$ be. Then
    \begin{align*}
        \left[ \Lambda_{\epsilon}\left(z\right)_{\text{off}} \right]_{\sigma \nu} & = - g A_{\epsilon}^{\sigma} R_{H_0}\left(z\right) A_{\epsilon}^{\nu \, \ast} = - g \left( v \otimes \mathds{1} \right) \frac{U_{\epsilon}^{\sigma}}{\sqrt{\epsilon}} U_{\sigma} R_{H_0}\left(z\right) \left[ \left( v \otimes \mathds{1} \right) \frac{U_{\epsilon}^{\nu}}{\sqrt{\epsilon}} U_{\nu} \right]^{\ast} = \\
        & = - g \left( v \otimes \mathds{1} \right) \frac{U_{\epsilon}^{\sigma}}{\sqrt{\epsilon}} U_{\sigma} R_{H_0} \left(z\right) U_{\nu}^{\ast} \frac{U_{\epsilon}^{\nu \, \ast}}{\sqrt{\epsilon}} \left( v \otimes \mathds{1} \right): \chi_{\nu} \longrightarrow \chi_{\sigma}.
    \end{align*}
    \noindent To simplify the notation, without harming generality, $\sigma = (12)$ is assumed; for all $\psi \in \chi_{\nu}$, $\left[\Lambda_{\epsilon}\left(z\right)_{\text{off}}\right]_{\sigma \nu} \psi \in \chi_{\sigma}$, hence
    \begin{align*}
        & \left( \left[ \Lambda_{\epsilon}\left(z\right)_{\text{off}} \right]_{\sigma \nu} \psi \right) \left( r_{\sigma}, R_{\sigma}, x_3, \ldots, x_n \right) = \\
        & = g \left\{ - \left[ \left( v \otimes \mathds{1} \right) \frac{U_{\epsilon}^{\sigma}}{\sqrt{\epsilon}} U_{\sigma} R_{H_0} \left(z\right) U_{\nu}^{\ast} \frac{U_{\epsilon}^{\nu \, \ast}}{\sqrt{\epsilon}} \left( v \otimes \mathds{1} \right) \right] \psi \right\} \left( r_{\sigma}, R_{\sigma}, x_3, \ldots, x_n \right) \\
        & = - g v\left(r_{\sigma}\right) \left\{ \left[ U_{\sigma} R_{H_0} \left(z\right) U_{\nu}^{\ast} \frac{U_{\epsilon}^{\nu \, \ast}}{\sqrt{\epsilon}} \left( v \otimes \mathds{1} \right) \right] \psi \right\} \left(\epsilon r_{\sigma}, R_{\sigma}, x_3, \ldots, x_n \right).
    \end{align*}
    \noindent It is recalled that
    \begin{equation*}
        U_{(12)}: \, \Psi \in L^2\left( \mathbb{R}^n, dx_1 \ldots dx_n \right) \longmapsto U_{(12)}\Psi \in L^2\left( \mathbb{R}^n, dr_{(12)} dR_{(12)} dx_3 \ldots dx_n \right)
    \end{equation*}
    \noindent with
    \begin{equation*}
        \left(U_{(12)}\Psi\right) \left( r_{(12)}, R_{(12)}, x_3, \ldots, x_n \right) \equiv \Psi \left( R_{(12)} - \frac{m_2}{m_1+m_2}r_{(12)}, R_{(12)} + \frac{m_1}{m_1+m_2}r_{(12)}, \, x_3, \ldots, x_n \right),
    \end{equation*}
    
    \noindent therefore
    
    \small{\begin{align*}
        & \left( \left[ \Lambda_{\epsilon}\left(z\right)_{\text{off}} \right]_{\sigma \nu} \psi \right) \left( r_{\sigma}, R_{\sigma}, x_3, \ldots, x_n \right) = \\
        & = - g v\left( r_{\sigma} \right) \left[ R_{H_0}\left(z\right) U_{\nu}^{\ast} \frac{U_{\epsilon}^{\nu \, \ast}}{\sqrt{\epsilon}} \left(v \otimes \mathds{1} \right) \psi \right] \left( R_{\sigma} - \frac{m_2}{m_1+m_2}\epsilon r_\sigma, \, R_\sigma + \frac{m_1}{m_1+m_2} \epsilon r_\sigma, \, x_3, \ldots, x_n \right) \\
        \begin{split}
            & = - g v\left( r_{\sigma} \right) \int_{\mathbb{R}^n} \, \left\{ \left[ R_{H_0} \left(z\right) \right] \left( R_{\sigma} - \frac{\epsilon m_2}{m_1 + m_2}r_{\sigma}, R_{\sigma} + \frac{\epsilon m_1}{m_1 + m_2}r_{\sigma}, \, x_3, \ldots, x_n, x_1^{\prime}, \ldots, x_n^{\prime} \right) \cdot \right. \\ 
            &\phantom{ = {- v\left( r_{\sigma} \right) \int_{\mathbb{R}^n} \, \left\{ \right.}} \left. \cdot \left[ U_{\nu}^{\ast} \frac{U_{\epsilon}^{\nu \, \ast}}{\sqrt{\epsilon}} \left(v \otimes \mathds{1} \right) \psi \right] \left(x_1^{\prime}, \ldots, x_n^{\prime}\right) \right\} dx_1^{\prime} \cdots dx_n^{\prime}
        \end{split} \\
        \begin{split}
            & = - g v\left( r_{\sigma} \right) \int_{\mathbb{R}^n} \, \left\{ \left[ R_{H_0} \left(z\right) \right] \left( R_{\sigma} - \frac{\epsilon m_2}{m_1 + m_2}r_{\sigma}, R_{\sigma} + \frac{\epsilon m_1}{m_1 + m_2}r_{\sigma}, \, x_3, \ldots, x_n, x_1^{\prime}, \ldots, x_n^{\prime} \right) \cdot \right. \\ 
            &\phantom{ = {- v\left( r_{\sigma} \right) \int_{\mathbb{R}^n} \, \left\{ \right.}} \left. \cdot \left[ \frac{U_{\epsilon}^{\nu \, \ast}}{\sqrt{\epsilon}} \left( v \otimes \mathds{1} \right) \psi \right] \left( x_1^{\prime}, \ldots, x_{\nu_1}^{\prime} - x_{\nu_2}^{\prime}, \ldots, \frac{m_{\nu_1} x_{\nu_1}^{\prime} + m_{\nu_2} x_{\nu_2}^{\prime} }{m_{\nu_1} + m_{\nu_2}}, \ldots, x_n^{\prime} \right) \right\} dx_1^{\prime} \cdots dx_n^{\prime}
        \end{split} \\
        \begin{split}
            & = - g v\left( r_{\sigma} \right) \int_{\mathbb{R}^n} \, \left\{ \left[ R_{H_0} \left(z\right) \right] \left( R_{\sigma} - \frac{\epsilon m_2}{m_1 + m_2}r_{\sigma}, R_{\sigma} + \frac{\epsilon m_1}{m_1 + m_2}r_{\sigma}, \, x_3, \ldots, x_n, x_1^{\prime}, \ldots, x_n^{\prime} \right) \frac{1}{\epsilon} v\left( \frac{x_{\nu_1}^{\prime} - x_{\nu_2}^{\prime}}{\epsilon} \right) \cdot \right. \\ 
            &\phantom{ = {- v\left( r_{\sigma} \right) \int_{\mathbb{R}^n} \, \left\{ \right.}} \left. \cdot \, \psi \left( \frac{x_{\nu_1}^{\prime} - x_{\nu_2}^{\prime}}{\epsilon}, \, \frac{m_{\nu_1} x_{\nu_1}^{\prime} + m_{\nu_2}x_{\nu_2}^{\prime}}{m_{\nu_1} + m_{\nu_2}}, \, x_1^{\prime}, \ldots, \hat{x}_{\nu_1}^{\prime}, \ldots, \hat{x}_{\nu_2}^{\prime}, \ldots, x_{n}^{\prime} \right) \right\} dx_1^{\prime} \cdots dx_n^{\prime}
        \end{split}\\
        \begin{split}
            & = - g v\left( r_{\sigma} \right) \int_{\mathbb{R}^n} \, \left\{ \left[ R_{H_0} \left(z\right) \right] \left( R_{\sigma} - \frac{\epsilon m_2}{m_1 + m_2}r_{\sigma}, R_{\sigma} + \frac{\epsilon m_1}{m_1 + m_2}r_{\sigma}, \, x_3, \ldots, x_n, x_1^{\prime}, \ldots, x_n^{\prime} \right) \frac{1}{\epsilon} v\left( \frac{x_{\nu_1}^{\prime} - x_{\nu_2}^{\prime}}{\epsilon} \right) \cdot \right. \\ 
            &\phantom{ = {- v\left( r_{\sigma} \right) \int_{\mathbb{R}^n} \, \left\{ \right.}} \left. \cdot \, \psi \left( \frac{x_{\nu_1}^{\prime} - x_{\nu_2}^{\prime}}{\epsilon}, \, x_{\nu_1}^{\prime} - \frac{\epsilon m_{\nu_2}}{m_{\nu_1} + m_{\nu_2}} \frac{x_{\nu_1}^{\prime} - x_{\nu_2}^{\prime}}{\epsilon} , \, x_1^{\prime}, \ldots, \hat{x}_{\nu_1}^{\prime}, \ldots, \hat{x}_{\nu_2}^{\prime}, \ldots, x_{n}^{\prime} \right) \right\} dx_1^{\prime} \cdots dx_n^{\prime}
        \end{split}\\
        \begin{split}
            & = - g v\left( r_{\sigma} \right) \int_{\mathbb{R}^{n+1}} \, \left\{ \left[ R_{H_0} \left(z\right) \right] \left( R_{\sigma} - \frac{\epsilon m_2}{m_1 + m_2}r_{\sigma}, R_{\sigma} + \frac{\epsilon m_1}{m_1 + m_2}r_{\sigma}, \, x_3, \ldots, x_n, x_1^{\prime}, \ldots, x_n^{\prime} \right) \frac{1}{\epsilon} v\left( r_{\nu}^{\prime} \right) \cdot \right. \\ 
            &\phantom{ = {- v\left( r_{\sigma} \right) \int_{\mathbb{R}^n} \, \left\{ \right.}}  \cdot \, \psi \left( r_{\nu}^{\prime}, \, x_{\nu_1}^{\prime} - \frac{\epsilon m_{\nu_2}}{m_{\nu_1} + m_{\nu_2}} r_{\nu}^{\prime} , \, x_1^{\prime}, \ldots, \hat{x}_{\nu_1}^{\prime}, \ldots, \hat{x}_{\nu_2}^{\prime}, \ldots, x_{n}^{\prime} \right) \cdot \\ &\phantom{ = {- v\left( r_{\sigma} \right) \int_{\mathbb{R}^n} \, \left\{ \right.}}  \cdot \, \left. \delta\left( r_{\nu}^{\prime} - \frac{x_{\nu_1}^{\prime} - x_{\nu_2}^{\prime}}{\epsilon} \right) \right\} dx_1^{\prime} \cdots dx_n^{\prime}dr_{\nu}^{\prime}
        \end{split}\\
    \end{align*}
    
    \begin{align*}
        \begin{split}
             & \equiv \left( \frac{1}{\epsilon} \, \delta\left( r_{\nu}^{\prime} - \frac{x_{\nu_1}^{\prime} - x_{\nu_2}^{\prime}}{\epsilon} \right) \equiv \delta \left(\epsilon r_{\nu}^{\prime} - \left( x_{\nu_1}^{\prime} - x_{\nu_2}^{\prime} \right) \right) \equiv \delta\left(x_{\nu_1}^{\prime} - x_{\nu_2}^{\prime} - \epsilon r_{\nu}^{\prime} \right) \right. \equiv \\
            & \phantom{ = {\equiv \left(\right.}} \equiv \left. \delta\left( \left( x_{\nu_1}^{\prime} - \frac{\epsilon m_{\nu_2}}{m_{\nu_1} + m_{\nu_2}} r_{\nu}^{\prime} \right) - \left( x_{\nu_2}^{\prime} + \frac{\epsilon m_{\nu_1}}{m_{\nu_1} + m_{\nu_2}}r_{\nu}^{\prime} \right) \right) \right) \equiv
        \end{split}\\
        \begin{split}
            & \equiv \int_{\mathbb{R}^{n+2}} \, \left\{ - g v\left(r_{\sigma}\right)  \left[ R_{H_0} \left(z\right) \right]\left( R_{\sigma} - \frac{\epsilon m_2}{m_1+m_2}r_{\sigma}, R_{\sigma} + \frac{\epsilon m_1}{m_1+m_2}r_{\sigma}, x_3, \ldots, x_n, x_1^{\prime}, \ldots, x_{n}^{\prime} \right) v\left( r_{\nu}^{\prime} \right) \right. \\
            & \phantom{ = { \equiv \int_{\mathbb{R}^n} \, \left\{ \right. } } \left. \delta\left( R_{\nu}^{\prime} - x_{\nu_1}^{\prime} + \frac{\epsilon m_{\nu_2}}{m_{\nu_1} + m_{\nu_2}} r_{\nu}^{\prime} \right) \delta\left( R_{\nu}^{\prime} - x_{\nu_2}^{\prime} - \frac{\epsilon m_{\nu_1}}{m_{\nu_1} + m_{\nu_2}}r_{\nu}^{\prime} \right) \right\} \cdot  \\
            & \phantom{ = { \equiv \int_{\mathbb{R}^n} \, \left\{ \right. } } \cdot \psi \left( r_{\nu}^{\prime}, R_{\nu}^{\prime}, x_1^{\prime}, \ldots, \hat{x}_{\nu_1}^{\prime}, \ldots, \hat{x}_{\nu_2}^{\prime}, \ldots, x_{n}^{\prime} \right) dr_{\nu}^{\prime} dR_{\nu}^{\prime} dx_{1}^{\prime} \cdots dx_{n}^{\prime}.
        \end{split}
     \end{align*}}

    \normalsize
    \noindent It is further observed that $\sigma \neq \nu$ may nonetheless imply $\nu_i \in \left\{1,2\right\}, \, i=1$ or $2$, hence the following cases are discussed.
    
    \vspace{5mm}
    \noindent $\boxed{\sigma = (12), \, \nu = (1\nu_2), \; \nu_2 \geq 3}$
    
    \small{\begin{align*}
        & \left[ \left(\Lambda_{\epsilon}\left(z\right)_{\text{off}}\right)_{\sigma \nu} \psi \right] \left( r_{\sigma},R_{\sigma},x_3,\ldots,x_n \right) = \\
        \begin{split}
            & \equiv \int_{\mathbb{R}^n} \, \left\{ - g v\left(r_{\sigma}\right)  \left[ R_{H_0} \left(z\right) \right]\left( R_{\sigma} - \frac{\epsilon m_2}{m_1+m_2}r_{\sigma}, R_{\sigma} + \frac{\epsilon m_1}{m_1+m_2}r_{\sigma}, x_3, \ldots, x_n, x_1^{\prime}, \ldots, x_{n}^{\prime} \right) v\left( r_{\nu}^{\prime} \right) \right. \\
            & \phantom{ = { \equiv \int_{\mathbb{R}^n} \, \left\{ \right. } } \left. \delta\left( R_{\nu}^{\prime} - x_1^{\prime} + \frac{\epsilon m_{\nu_2}}{m_1 + m_{\nu_2}} r_{\nu}^{\prime} \right) \delta\left( R_{\nu}^{\prime} - x_{\nu_2}^{\prime} - \frac{\epsilon m_1}{m_1 + m_{\nu_2}}r_{\nu}^{\prime} \right) \right\} \cdot  \\
            & \phantom{ = { \equiv \int_{\mathbb{R}^n} \, \left\{ \right. } } \cdot \psi \left( r_{\nu}^{\prime}, R_{\nu}^{\prime}, \hat{x}_1^{\prime}, x_2^{\prime}, \ldots, \hat{x}_{\nu_2}^{\prime}, \ldots, x_{n}^{\prime} \right) dr_{\nu}^{\prime} dR_{\nu}^{\prime} dx_{1}^{\prime} \cdots dx_{n}^{\prime} \equiv
        \end{split}\\
        & \equiv \left( \text{by integrating with respect to } x_{1}^{\prime}, \, x_{\nu_2}^{\prime} \right) \equiv \\
        \begin{split}
            & \equiv \int_{\mathbb{R}^n} \, \left\{ - g v\left( r_{\sigma} \right) \left[ R_{H_0}\left( z \right) \right] \left( R_{\sigma} - \frac{\epsilon m_2}{m_1 + m_2}r_{\sigma} - R_{\nu}^{\prime} - \frac{\epsilon m_{\nu_2}}{m_1 + m_{\nu_2}}r_{\nu}^{\prime}, \, R_{\sigma} + \frac{\epsilon m_1}{m_1 + m_2}r_{\sigma} - x_2^{\prime}, \, x_3 - x_3^{\prime}, \, \ldots, \right. \right. \\
            & \phantom{ = {\equiv \int_{\mathbb{R}^n} \, \left\{\right.}} \left. \left. x_{\nu_2} - R_{\nu}^{\prime} + \frac{\epsilon m_1}{m_1+m_{\nu_2}}r_{\nu}, \ldots, x_n - x_n^{\prime} \right) v\left(r_{\nu}^{\prime} \right) \right\} \cdot \\
            & \phantom{ = {\equiv \int_{\mathbb{R}^n} \, \left\{\right.}} \cdot \psi \left( r_{\nu}^{\prime}, R_{\nu}^{\prime}, \hat{x}_1^{\prime}, x_2^{\prime}, \ldots, \hat{x}_{\nu_2}^{\prime}, \ldots, x_{n}^{\prime} \right) dr_{\nu}^{\prime} dR_{\nu}^{\prime} dx_{2}^{\prime} \cdots d\hat{x}_{\nu_2}^{\prime} \cdots dx_{n}^{\prime}.
        \end{split}
    \end{align*}}
    
    \normalsize
    \noindent It is clearly seen that, with respect to the tuple of variables $\underline{Y}_{\nu} \equiv \left( x_3, \ldots, \hat{x}_{\nu_2}, \ldots, x_n \right) \in \mathbb{R}^{n-3}$, $\left( \Lambda_{\epsilon}\left(z\right)_{\text{off}} \right)_{\sigma \nu}$ behaves as a convolution operator; set $\chi^{-}_{\nu} \equiv L^2\left( \mathbb{R}^{n-3}, dx_3 \cdots d\hat{x}_{\nu_2} \cdots dx_{n} \right)$ and denoted by $\mathfrak{F}_{\underline{Y}_{\nu}}$ the Fourier operator on $\chi^{-}_{\nu}$, the operator 
    \begin{equation}
        \left( \mathds{1} \otimes \mathfrak{F}_{\underline{Y}_{\nu}} \right) \left( \Lambda_{\epsilon}\left(z\right)_{\text{off}} \right)_{\sigma \nu} \left( \mathds{1} \otimes \mathfrak{F}_{\underline{Y}_{\nu}}^{-1} \right) \doteq \left[ \Lambda_{\epsilon, \, \underline{P}_{\nu}}\left(z\right)_{\text{off}} \right]_{\sigma \nu}
    \end{equation}
    \noindent will be multiplicative with respect to the conjugate tuple $\underline{P}_{(1\nu_2)} \equiv \left( p_3,\ldots, \hat{p}_{\nu_2},\ldots,p_n \right)$; concerning the remaining variables, it behaves as an integral operator whose kernel is
    \begin{equation*}
        - 2^{\frac{3}{2}} g \sqrt{m_1m_2m_{\nu_2}} \, v\left(r_{\sigma}\right) G_{z-Q_{\nu}}^{(3)} \left( X_{\sigma \nu, \epsilon} \right) v\left( r_{\nu}^{\prime} \right) \equiv C(g,m_1,m_2,m_{\nu_2}) \; v\left(r_{\sigma}\right) G_{z-Q_{\nu}}^{(3)} \left( X_{\sigma \nu, \epsilon} \right) v\left( r_{\nu}^{\prime} \right),
    \end{equation*}
    \noindent and
    \begin{equation*}
        X_{\sigma \nu, \, \epsilon} = \left(
        \begin{matrix}
            \sqrt{2m_1} \left[ R_{\sigma} - R_{\nu}^{\prime} - \epsilon \left( \frac{m_2}{m_1+m_2}r_{\sigma} - \frac{m_{\nu_2}}{m_1 + m_{\nu_2}} r_{\nu}^{\prime} \right) \right]\\
            \sqrt{2m_2} \left[ R_{\sigma} - x_2^{\prime} + \epsilon \left( \frac{m_1}{m_1 + m_2} \right) r_{\sigma} \right]\\
            \sqrt{2m_{\nu_2}} \left[ x_{\nu_2} - R_{\nu}^{\prime} + \epsilon\left( \frac{m_1}{m_1 + m_{\nu_2}} \right) r_{\nu}^{\prime} \right]
        \end{matrix} \right), \quad \quad Q_{\nu} = \sum_{\underset{k \neq \nu_2}{k=3}}^{n} \, \frac{p_k^2}{2m_k}.
    \end{equation*}
    \noindent Given $\eta \in L^2\left( \mathbb{R}^3, dr_{\nu} dR_{\nu} dx_2 \right), \, \xi \in \tilde{\chi}^{-}_{\nu} \equiv \mathfrak{F}_{\underline{Y}_{\nu}} \chi^{-}_{\nu}$ arbitrarily,
    
    \footnotesize{\begin{align*}
        & \norm{ \left[\Lambda_{\epsilon, \, \underline{P}_{\nu}}\left(z\right)_{\text{off}}\right]_{\sigma \nu} \eta \otimes \xi }_2^2 = C^2 \,  \int_{\mathbb{R}^n} \, d\underline{P}_{\nu} dr_{\sigma}dR_{\sigma}dx_{\nu_2} \, \abs{ \int_{\mathbb{R}^3} \, dr_{\nu}^{\prime}dR_{\nu}^{\prime} dx_{\nu_2}^{\prime} \, v\left(r_{\sigma}\right) v\left(r_{\nu}^{\prime}\right) G_{z-Q_{\nu}}^{(3)} \left( X_{\sigma \nu, \epsilon} \right) \eta\left(r_{\nu}^{\prime}, R_{\nu}^{\prime}, x_{2}^{\prime} \right) \xi\left( \underline{P}_{\nu} \right) }^2 \\
        & = C^2 \,  \int_{\mathbb{R}^n} \, d\underline{P}_{\nu} dr_{\sigma}dR_{\sigma}dx_{\nu_2} \abs{ \int_{\mathbb{R}} \, dr_{\nu}^{\prime} \left\{ v\left(r_{\sigma}\right) v\left(r_{\nu}^{\prime} \right) \int_{\mathbb{R}^2} dR_{\nu}^{\prime}dx_2^{\prime} \left[ G_{z-Q_{\nu}}^{(3)} \left( X_{\sigma \nu, \epsilon} \right) \eta\left(r_{\nu}^{\prime},R_{\nu}^{\prime},x_2^{\prime} \right) \xi\left(\underline{P}_{\nu}\right) \right] \right\} }^2 \\
        & \leq C^2 \, \int_{\mathbb{R}^n} \, d\underline{P}_{\nu}dr_{\sigma}dR_{\sigma}dx_{\nu_2} \left[ \int_{\mathbb{R}} \, dr_{\nu}^{\prime} \left\{ \, v\left(r_{\sigma}\right)v\left(r_{\nu}^{\prime}\right) \abs{\int_{\mathbb{R}^2} \, dR_{\nu}^{\prime} dx_2^{\prime} \; G_{z-Q_{\nu}}^{(3)}\left( X_{\sigma\nu,\epsilon} \right) \eta\left(r_{\nu}^{\prime},R_{\nu}^{\prime},x_2^{\prime} \right) \xi\left( \underline{P}_{\nu}\right) }  \right\} \right]^2 \\
        & \leq C^2 \, \int_{\mathbb{R}^n} \, d\underline{P}_{\nu}dr_{\sigma}dR_{\sigma}dx_{\nu_2} \left\{ \left[ \int_{\mathbb{R}} dr_{\nu}^{\prime} \, V\left(r_{\sigma}\right) V\left(r_{\nu}^{\prime}\right) \right] \int_{\mathbb{R}} dr_{\nu}^{\prime} \abs{\int_{\mathbb{R}^2} \, G_{z-Q_{\nu}}^{(3)} \left(X_{\sigma \nu,\epsilon} \right) \eta\left(r_{\nu}^{\prime},R_{\nu}^{\prime},x_2^{\prime} \right) \xi\left( \underline{P}_{\nu}\right) }^2 \right\} \\
        & \leq C^2 \, \left\{ \underset{r_{\sigma} \in \mathbb{R}} {\sup} \; \int_{\mathbb{R}^n} \, d\underline{P}_{\nu}dR_{\sigma}dx_{\nu_2}dr_{\nu}^{\prime} \, \left[ \int_{\mathbb{R}^2} \, dR_{\nu}^{\prime}dx_{2}^{\prime} \; G_{z-Q_{\nu}}^{(3)} \left(X_{\sigma\nu,\epsilon} \right) \abs{\eta\left(r_{\nu}^{\prime}, R_{\nu}^{\prime}, x_2^{\prime} \right)} \abs{\xi\left( \underline{P}_{\nu} \right)} \right]^2 \right\}.
    \end{align*}}

    \normalsize
    \noindent Let the following coordinate transformation be
    \begin{equation*}
        \begin{cases}
            \overline{R}_{\nu}^{\prime} = R_{\nu}^{\prime} + \epsilon \left( \frac{m_2}{M_{\sigma}}r_{\sigma} - \frac{m_{\nu_2}}{M_{\nu}}r_{\nu}^{\prime} \right) \\
            \overline{x}_2^{\prime} = x_2^{\prime} - \epsilon \left( \frac{m_1}{M_{\sigma}} \right)r_{\sigma} \\
            \overline{x}_{\nu_2} = x_{\nu_2} + \epsilon \left( \frac{m_2}{M_{\sigma}}r_{\sigma} - r_{\nu}^{\prime} \right)
        \end{cases}.
    \end{equation*}

    \noindent What follows holds.

    \small{\begin{align*}
        & \underset{r_{\sigma} \in \mathbb{R}} {\sup} \; \int_{\mathbb{R}^n} \, d\underline{P}_{\nu}dR_{\sigma}dx_{\nu_2}dr_{\nu}^{\prime} \, \left[ \int_{\mathbb{R}^2} \, dR_{\nu}^{\prime}dx_{2}^{\prime} \; G_{z-Q_{\nu}}^{(3)} \left(X_{\sigma\nu,\epsilon} \right) \abs{\eta\left(r_{\nu}^{\prime}, R_{\nu}^{\prime}, x_2^{\prime} \right)} \abs{\xi\left( \underline{P}_{\nu} \right)} \right]^2 \equiv \\
        \begin{split}
            & \equiv \underset{r_{\sigma} \in \mathbb{R}}{\sup} \; \int_{\mathbb{R}^n} \, d\underline{P}_{\nu} dR_{\sigma} d\overline{x}_{\nu_2} dr_{\nu}^{\prime} \, \left[ \int_{\mathbb{R}^2} \, d\overline{R}_{\nu}^{\prime} d\overline{x}_2^{\prime} \; G_{z-Q_{\nu}}^{(3)}\left( \sqrt{2m_1}\left(R_{\sigma} - \overline{R}_{\nu}^{\prime} \right), \sqrt{2m_2}\left( R_{\sigma} - \overline{x}_2^{\prime} \right), \sqrt{2m_{\nu_2}}\left( \overline{x}_{\nu_2} - \overline{R}_{\nu}^{\prime} \right)  \right) \right. \\
            &\phantom{ = {\equiv \underset{r_{\sigma} \in \mathbb{R}}{\sup} \; \int_{\mathbb{R}^n} \, d\underline{P}_{\nu} dR_{\sigma} d\overline{x}_{\nu_2} dr_{\nu}^{\prime} \, \left[\right. \int_{\mathbb{R}^2} \,} } \left. \abs{ \eta \left( r_{\nu}^{\prime}, \, \overline{R}_{\nu}^{\prime} - \epsilon \left( \frac{m_2}{M_{\sigma}}r_{\sigma} - \frac{m_{\nu_2}}{M_{\nu}} r_{\nu}^{\prime} \right), \, \overline{x}_2^{\prime} + \epsilon \left( \frac{m_1}{M_{\sigma}}\right) r_{\sigma} \right) } \abs{\xi \left( \underline{P}_{\nu} \right)} \right]^2
        \end{split}\\
        \begin{split}
            & \leq \underset{r_{\sigma} \in \mathbb{R}}{\sup} \; \int_{\mathbb{R}^n} \, d\underline{P}_{\nu} dR_{\sigma} d\overline{x}_{\nu_2} dr_{\nu}^{\prime} \, \left[ \int_{\mathbb{R}^2} \, d\overline{R}_{\nu}^{\prime} d\overline{x}_2^{\prime} \; G_{z}^{(3)}\left( \sqrt{2m_1}\left(R_{\sigma} - \overline{R}_{\nu}^{\prime} \right), \sqrt{2m_2}\left( R_{\sigma} - \overline{x}_2^{\prime} \right), \sqrt{2m_{\nu_2}}\left( \overline{x}_{\nu_2} - \overline{R}_{\nu}^{\prime} \right)  \right) \right. \\
            &\phantom{ = {\equiv \underset{r_{\sigma} \in \mathbb{R}}{\sup} \; \int_{\mathbb{R}^n} \, d\underline{P}_{\nu} dR_{\sigma} d\overline{x}_{\nu_2} dr_{\nu}^{\prime} \, \left[\right. \int_{\mathbb{R}^2} \,} } \left. \abs{ \eta \left( r_{\nu}^{\prime}, \, \overline{R}_{\nu}^{\prime} - \epsilon \left( \frac{m_2}{M_{\sigma}}r_{\sigma} - \frac{m_{\nu_2}}{M_{\nu}} r_{\nu}^{\prime} \right), \, \overline{x}_2^{\prime} + \epsilon \left( \frac{m_1}{M_{\sigma}}\right) r_{\sigma} \right) } \abs{\xi \left( \underline{P}_{\nu} \right)} \right]^2
        \end{split}\\
    \end{align*}
    \begin{align*}
        \begin{split}
            & \leq \left\{ \underset{r_{\sigma} \in \mathbb{R}}{\sup} \; \int_{\mathbb{R}^{3}} \, dR_{\sigma} d\overline{x}_{\nu_2} dr_{\nu}^{\prime} \, \left[ \int_{\mathbb{R}^2} \, d\overline{R}_{\nu}^{\prime} d\overline{x}_2^{\prime} \; G_{z}^{(3)}\left( \sqrt{2m_1}\left(R_{\sigma} - \overline{R}_{\nu}^{\prime} \right), \sqrt{2m_2}\left( R_{\sigma} - \overline{x}_2^{\prime} \right), \sqrt{2m_{\nu_2}}\left( \overline{x}_{\nu_2} - \overline{R}_{\nu}^{\prime} \right)  \right) \right. \right. \\
            &\phantom{ = {\equiv \underset{r_{\sigma} \in \mathbb{R}}{\sup} \; \int_{\mathbb{R}^n} \, d\underline{P}_{\nu} dR_{\sigma} } }  \left. \left. \abs{ \eta \left( r_{\nu}^{\prime}, \, \overline{R}_{\nu}^{\prime} - \epsilon \left( \frac{m_2}{M_{\sigma}}r_{\sigma} - \frac{m_{\nu_2}}{M_{\nu}} r_{\nu}^{\prime} \right), \, \overline{x}_2^{\prime} + \epsilon \left( \frac{m_1}{M_{\sigma}}\right) r_{\sigma} \right) } \right]^2 \right\} \left( \int_{\mathbb{R}^{n-3}} \, d\underline{P}_{\nu} \abs{\xi \left( \underline{P}_{\nu} \right)}^2 \right) \\
            & \leq \norm{F}^2 \norm{\eta \otimes \xi}^2.
        \end{split}
    \end{align*}}
    
    \normalsize
    \noindent Eventually\footnote{Appendix 4 enters the argument.}, 
    \begin{equation}\label{Norm estimate 1 Off-Diagonal Lambda}
        \norm{\left[ \Lambda_{\epsilon}\left( z \right)_{\text{off}} \right]_{\sigma \nu}}_{\mathfrak{B}\left( \chi_{\sigma}, \chi_{\nu} \right)} \leq \abs{C} \norm{F} \leq \frac{\abs{C}}{2\sqrt{2\abs{z}}} \leq \left( \underset{i}{\max} \; m_i^{\frac{3}{2}} \right) \frac{\abs{g}}{\sqrt{\abs{z}}},
    \end{equation}
    \noindent with $\sigma = (12), \, \nu = (1\nu_2), \, \nu_2 \geq 3$ and independently of $\epsilon > 0$. An analogous argument holds for $\sigma = (12), \, \nu = (2\nu_2), \; \nu_2 \geq 3$. 
    
    \vspace{5mm}
    \noindent $\boxed{\sigma = (12), \, \nu = (\nu_1\nu_2), \; 3 \leq \nu_1 < \nu_2 \leq n}$
    \small{\begin{align*}
        & \left\{ \left[ \Lambda_{\epsilon}\left( z \right)_{\text{off}} \right]_{\sigma \nu} \psi \right\} \left( r_{\sigma}, R_{\sigma}, x_3, \ldots, x_n \right) = \\
        \begin{split}
            & = \int_{\mathbb{R}^{n+2}} \, \left\{ - g v\left(r_{\sigma}\right)  \left[ R_{H_0} \left(z\right) \right]\left( R_{\sigma} - \frac{\epsilon m_2}{m_1+m_2}r_{\sigma}, R_{\sigma} + \frac{\epsilon m_1}{m_1+m_2}r_{\sigma}, x_3, \ldots, x_n, x_1^{\prime}, \ldots, x_{n}^{\prime} \right) v\left( r_{\nu}^{\prime} \right) \right. \\
            & \phantom{ = { \equiv \int_{\mathbb{R}^n} \, \left\{ \right. } } \left. \delta\left( R_{\nu}^{\prime} - x_{\nu_1}^{\prime} + \frac{\epsilon m_{\nu_2}}{m_{\nu_1} + m_{\nu_2}} r_{\nu}^{\prime} \right) \delta\left( R_{\nu}^{\prime} - x_{\nu_2}^{\prime} - \frac{\epsilon m_{\nu_1}}{m_{\nu_1} + m_{\nu_2}}r_{\nu}^{\prime} \right) \right\} \cdot  \\
            & \phantom{ = { \equiv \int_{\mathbb{R}^n} \, \left\{ \right. } } \cdot \psi \left( r_{\nu}^{\prime}, R_{\nu}^{\prime}, x_1^{\prime}, \ldots, \hat{x}_{\nu_1}^{\prime}, \ldots, \hat{x}_{\nu_2}^{\prime}, \ldots, x_{n}^{\prime} \right) dr_{\nu}^{\prime} dR_{\nu}^{\prime} dx_{1}^{\prime} \cdots dx_{n}^{\prime}
        \end{split}\\
        \begin{split}
            & = \int_{\mathbb{R}^n} \, \left\{ - g v\left(r_{\sigma} \right) \cdot \right.
        \end{split}\\
        \begin{split}
            & \phantom{ = {\int_{\mathbb{R}^n} \, \left\{\right.} } \cdot \left[ R_{H_0} \left(z\right) \right] \left( R_{\sigma} - \frac{\epsilon m_2}{m_1+m_2}r_{\sigma} - x_1^{\prime}, \, R_{\sigma} + \frac{\epsilon m_1}{m_1 + m_2}r_{\sigma} - x_2^{\prime}, \, x_3 - x_3^{\prime}, \ldots, \, x_{\nu_1} - R_{\nu}^{\prime} - \frac{\epsilon m_{\nu_2}}{m_{\nu_1} + m_{\nu_2}} r_{\nu}^{\prime}, \right.
        \end{split}\\
        \begin{split}
            & \phantom{ = {\int_{\mathbb{R}^n} \, \left\{ \left[ R_{H_0} \left(z\right) \right] \left( \right. \right.} } \left. \ldots, \, x_{\nu_2} - R_{\nu}^{\prime} + \frac{\epsilon m_{\nu_1}}{m_{\nu_1} + m_{\nu_2}}r_{\nu}^{\prime}, \ldots, x_n - x_n^{\prime} \right) \cdot
        \end{split}\\
        \begin{split}
            & \phantom{ = {\int_{\mathbb{R}^n} \, \left\{ \right.} } \cdot \left. \psi \left( r_{\nu}^{\prime}, R_{\nu}^{\prime}, x_1^{\prime}, \ldots, \hat{x}_{\nu_1}^{\prime}, \ldots, \hat{x}_{\nu_2}^{\prime}, \ldots, x_{n}^{\prime} \right) \right\}  dr_{\nu}^{\prime} dR_{\nu}^{\prime} dx_{1}^{\prime} \cdots d\hat{x}_{\nu_1}^{\prime} \cdots d\hat{x}_{\nu_2}^{\prime} \cdots dx_n.
        \end{split}
    \end{align*}}
    
    \normalsize
    \noindent $\left[ \Lambda_{\epsilon}\left( z \right)_{\text{off}} \right]_{\sigma \nu}$ behaves as a convolution operator with respect to $\underline{Y}_{\nu} = \left( x_3,\ldots, \hat{x}_{\nu_1}, \ldots, \hat{x}_{\nu_2}, \ldots, x_n \right) \in \mathbb{R}^{n-4}$, therefore, by introducing
    \begin{equation*}
        \chi^{-}_{\nu} \doteq L^2\left( \mathbb{R}^{n-4}, \, dx_3 \ldots d\hat{x}_{\nu_1} \ldots d\hat{x}_{\nu_2} \ldots dx_n  \right)
    \end{equation*}
    \noindent and the corresponding Fourier operator $\mathfrak{F}_{\underline{Y}_{\nu}}$ on it, 
    \begin{equation*}
        \left[ \Lambda_{\epsilon, \, \underline{P}_{\nu}} \left(z\right)_{\text{off}} \right]_{\sigma \nu} \doteq \left( \mathds{1} \otimes \mathfrak{F}_{\underline{Y}_{\nu}} \right) \left[ \Lambda_{\epsilon}\left(z\right)_{\text{off}} \right]_{\sigma \nu} \left( \mathds{1} \otimes \mathfrak{F}_{\underline{Y}_{\nu}}^{-1} \right)
    \end{equation*}
    \noindent would be multiplicative in $\underline{P}_{\nu} = \left(p_3, \ldots, \hat{p}_{\nu_1}, \ldots, \hat{p}_{\nu_2}, \ldots, p_n \right)$; on the other hand, $\left[ \Lambda_{\epsilon, \, \underline{P}_{\nu}} \left(z\right)_{\text{off}} \right]_{\sigma \nu}$ is a integral operator on $L^2\left( \mathbb{R}^4, dr_{\nu} dR_{\nu} dx_1 dx_2 \right)$, with kernel
    \begin{equation*}
        - 4 g \sqrt{m_1m_2m_{\nu_1}m_{\nu_2}} \, v\left( r_{\sigma} \right) G_{z-Q_{\nu}}^{(4)} \left( X_{\sigma \nu, \epsilon} \right) v\left( r_{\nu}^{\prime} \right) \equiv C\left(g,m_1,m_2,m_{\nu_1},m_{\nu_2}\right) \, v\left( r_{\sigma} \right) G_{z-Q_{\nu}}^{(4)} \left( X_{\sigma \nu, \epsilon} \right) v\left( r_{\nu}^{\prime} \right),
    \end{equation*}
    \noindent where
    \begin{equation*}
        X_{\sigma \nu, \epsilon} = 
        \left(\begin{matrix}
            & \sqrt{2m_1} \, \left( R_{\sigma} - \frac{\epsilon m_2}{m_1 + m_2}r_{\sigma} - x_1^{\prime} \right) \\
            & \sqrt{2m_2} \, \left( R_{\sigma} + \frac{\epsilon m_1}{m_1 + m_2}r_{\sigma} - x_2^{\prime} \right) \\
            & \sqrt{2m_{\nu_1}} \, \left( x_{\nu_1} - R_{\nu}^{\prime} - \frac{\epsilon m_{\nu_2}}{m_{\nu_1} + m_{\nu_2}} r_{\nu}^{\prime} \right) \\
            & \sqrt{2m_{\nu_2}} \, \left( x_{\nu_2} - R_{\nu}^{\prime} + \frac{\epsilon m_{\nu_1}}{m_{\nu_1} + m_{\nu_2}} r_{\nu}^{\prime} \right)
        \end{matrix}\right), \quad \quad Q_{\nu} = \sum_{\underset{k \neq \nu_1, \, \nu_2 }{k=3}}^n \, \frac{p^2_k}{2m_k}.
    \end{equation*}
    
    \noindent Therefore, arbitrarily given $\eta \in L^2\left( \mathbb{R}^4, dr_{\nu} dR_{\nu} dx_1 dx_2 \right), \, \xi \in \tilde{\chi}^{-}_{\nu}$,
    
    \small{\begin{align*}
        & \norm{ \left[ \Lambda_{\epsilon, \underline{P}_{\nu}} \left(z\right)_{\text{off}} \right]_{\sigma \nu} \eta \otimes \xi }_2^{2} = \\
        & = C^2 \int_{\mathbb{R}^n} \, d\underline{P}_{\nu} dr_{\sigma} dR_{\sigma} dx_{\nu_1} dx_{\nu_2} \, \abs{ \int_{\mathbb{R}^4} \, dr_{\nu}^{\prime}dR_{\nu}^{\prime}dx_1^{\prime}dx_2^{\prime} \; v\left( r_{\sigma} \right) G_{z-Q_{\nu}}^{(4)} \left( X_{\sigma \nu, \epsilon} \right) v\left( r_{\nu}^{\prime} \right) \eta\left(r_{\nu}^{\prime}, R_{\nu}^{\prime}, x_1^{\prime}, x_2^{\prime} \right) \xi\left( \underline{P}_{\nu} \right) }^2 \\
        & \leq C^2 \int_{\mathbb{R}^n} \, d\underline{P}_{\nu} dr_{\sigma} dR_{\sigma} dx_{\nu_1} dx_{\nu_2} \, \left\{ \int_{\mathbb{R}} \, dr_{\nu}^{\prime} v\left(r_{\sigma}\right) v\left( r_{\nu}^{\prime} \right) \abs{ \int_{\mathbb{R}^3} \, dR_{\nu}^{\prime}dx_1^{\prime} dx_{2}^{\prime} \, G_{z-Q_{\nu} }^{(4)} \left( X_{\sigma \nu, \epsilon} \right) \eta\left(r_{\nu}^{\prime}, R_{\nu}^{\prime}, x_1^{\prime}, x_2^{\prime} \right) \xi \left( \underline{P}_{\nu} \right) } \right\}^2 \\
        & \leq \left( \text{by using H{\"o}lder inequality together with the fact that} \, \int_{\mathbb{R}} \, V = 1 \right) \\
        & \leq C^2 \underset{r_{\sigma} \in \mathbb{R}}{\sup} \, \int_{\mathbb{R}^{n-1}} d\underline{P}_{\nu} dR_{\sigma} dx_{\nu_1} dx_{\nu_2}  \, \int_{\mathbb{R}} \, dr_{\nu}^{\prime} \, \abs{ \int_{\mathbb{R}^3} \, dR_{\nu}^{\prime}dx_1^{\prime}dx_2^{\prime} \, G_{z-Q_{\nu}}^{(4)} \left( X_{\sigma \nu, \epsilon} \right) \eta\left( r_{\nu}^{\prime}, R_{\nu}^{\prime}, x_1^{\prime}, x_2^{\prime} \right) \xi\left(\underline{P}_{\nu} \right) }^2.
    \end{align*}}
    
    \normalsize
    \noindent The coordinate transformation
    \begin{equation*}
        \begin{cases}
            \overline{R}_{\nu}^{\prime} & = R_{\nu}^{\prime} + \epsilon \left( \frac{m_1}{m_1+m_2} r_{\sigma} - \frac{m_{\nu_1}}{m_{\nu_1} + m_{\nu_2}} r_{\nu}^{\prime} \right) \\
            \overline{x}_{\nu_1} & = x_{\nu_1} + \epsilon \left( \frac{m_1}{m_1 + m_2} r_{\sigma} - r_{\nu}^{\prime} \right) \\
            \overline{x}_{\nu_2} & = x_{\nu_2} + \epsilon \frac{m_1}{m_1+m_2} r_{\sigma} \\
            \overline{x}_1^{\prime} & = x_1^{\prime} + \epsilon \frac{m_2}{m_1 + m_2} r_{\sigma} \\
            \overline{x}_2^{\prime} & = x_2^{\prime} - \epsilon \frac{m_1}{m_1+m_2} r_{\sigma}
        \end{cases}
    \end{equation*}
    
    \noindent allows for
    
    \small{\begin{align*}
        & \underset{r_{\sigma} \in \mathbb{R}}{\sup} \, \int_{\mathbb{R}^{n-1}} d\underline{P}_{\nu} dR_{\sigma} dx_{\nu_1} dx_{\nu_2}  \, \int_{\mathbb{R}} \, dr_{\nu}^{\prime} \, \abs{ \int_{\mathbb{R}^3} \, dR_{\nu}^{\prime}dx_1^{\prime}dx_2^{\prime} \, G_{z-Q_{\nu}}^{(4)} \left( X_{\sigma \nu, \epsilon} \right) \eta\left( r_{\nu}^{\prime}, R_{\nu}^{\prime}, x_1^{\prime}, x_2^{\prime} \right) \xi\left(\underline{P}_{\nu} \right) }^2 \equiv \\
        \begin{split}
            & \equiv \underset{r_{\sigma} \in \mathbb{R}}{\sup} \, \int_{\mathbb{R}^n} \, d\underline{P}_{\nu} dR_{\sigma} d\overline{x}_{\nu_1} d\overline{x}_{\nu_2} dr_{\nu}^{\prime} \, \bigg\lvert  \int_{\mathbb{R}^3} \, d\overline{R}_{\nu}^{\prime} d\overline{x}_1^{\prime} d\overline{x}_2^{\prime} \; \left[\right. \tilde{\eta}\left( r_{\nu}^{\prime}, \overline{R}_{\nu}^{\prime}, \overline{x}_1^{\prime}, \overline{x}_2^{\prime} \right) \xi\left( \underline{P}_{\nu} \right)  \\
            &\phantom{ = {} } \left. G_{z-Q_{\nu}}^{(4)} \left( \sqrt{2m_1} \left( R_{\sigma} - \overline{x}_1^{\prime} \right), \sqrt{2m_2} \left( R_{\sigma} - \overline{x}_2^{\prime} \right), \sqrt{2m_{\nu_1}} \left( \overline{x}_{\nu_1} - \overline{R}_{\nu}^{\prime} \right), \sqrt{2m_{\nu_2}} \left( \overline{x}_{\nu_2} - \overline{R}_{\nu}^{\prime} \right) \right) \right] \bigg\rvert ^{2}
        \end{split}\\
        \begin{split}
            & \leq \underset{r_{\sigma} \in \mathbb{R}}{\sup} \, \int_{\mathbb{R}^n} \, d\underline{P}_{\nu} dR_{\sigma} d\overline{x}_{\nu_1} d\overline{x}_{\nu_2} dr_{\nu}^{\prime} \, \left[  \int_{\mathbb{R}^3} \, d\overline{R}_{\nu}^{\prime} d\overline{x}_1^{\prime} d\overline{x}_2^{\prime} \;  \abs{\tilde{\eta}\left( r_{\nu}^{\prime}, \overline{R}_{\nu}^{\prime}, \overline{x}_1^{\prime}, \overline{x}_2^{\prime} \right) \xi\left( \underline{P}_{\nu} \right)} \right.  \\
            &\phantom{ = {} } \left. G_z^{(4)} \left( \sqrt{2m_1} \left( R_{\sigma} - \overline{x}_1^{\prime} \right), \sqrt{2m_2} \left( R_{\sigma} - \overline{x}_2^{\prime} \right), \sqrt{2m_{\nu_1}} \left( \overline{x}_{\nu_1} - \overline{R}_{\nu}^{\prime} \right), \sqrt{2m_{\nu_2}} \left( \overline{x}_{\nu_2} - \overline{R}_{\nu}^{\prime} \right) \right) \right]^2
        \end{split} \\
        & \leq \left( \text{by using H{\"o}lder inequality} \right)\\
        \begin{split}
            & \leq \left( \int_{\mathbb{R}^{n-4}} \,  d\underline{P}_{\nu} \, \abs{\xi\left( \underline{P}_{\nu} \right) }^2 \right) \cdot  \underset{r_{\sigma} \in \mathbb{R}}{\sup} \, \int_{\mathbb{R}} dr_{\nu}^{\prime} \int_{\mathbb{R}^3} \, dR_{\sigma} d\overline{x}_{\nu_1} d\overline{x}_{\nu_2} \left[ \int_{\mathbb{R}^3} \, d\overline{R}_{\nu}^{\prime} d\overline{x}_1^{\prime} d\overline{x}_2^{\prime} \, \abs{\tilde{\eta}\left( r_{\nu}^{\prime}, \overline{R}_{\nu}^{\prime}, \overline{x}_1^{\prime}, \overline{x}_2^{\prime} \right)} \cdot \right.
        \end{split}\\
        \begin{split}
            &\phantom{ = {  } } \left. \cdot \, G_z^{(4)} \left( \sqrt{2m_1} \left( R_{\sigma} - \overline{x}_1^{\prime} \right), \sqrt{2m_2} \left( R_{\sigma} - \overline{x}_2^{\prime} \right), \sqrt{2m_{\nu_1}} \left( \overline{x}_{\nu_1} - \overline{R}_{\nu}^{\prime} \right), \sqrt{2m_{\nu_2}} \left( \overline{x}_{\nu_2} - \overline{R}_{\nu}^{\prime} \right) \right) \right]^2,
        \end{split}
    \end{align*}}
    
    \normalsize
    \noindent where
    
    \small{\begin{equation*}
        \tilde{\eta}\left( r_{\nu}^{\prime}, \overline{R}_{\nu}^{\prime}, \overline{x}_1^{\prime}, \overline{x}_2^{\prime} \right) = \eta\left( r_{\nu}^{\prime}, \, \overline{R}_{\nu}^{\prime} - \epsilon\left( \frac{m_1}{m_1+m_2}r_{\sigma} - \frac{m_{\nu_1}}{m_{\nu_1} + m_{\nu_2}} r_{\nu}^{\prime} \right), \, \overline{x}_1^{\prime} - \epsilon \frac{m_2}{m_1+m_2}r_{\sigma}, \, \overline{x}_2^{\prime} + \epsilon \frac{m_1}{m_1+m_2} r_{\sigma} \right).
    \end{equation*}}
    
    \normalsize
    \noindent Consequently, by Appendix 4, independently of $\epsilon >0$,
    
    \begin{equation}\label{Norm estimate 2 Off-Diagonal Lambda}
        \norm{ \left[\Lambda_{\epsilon}\left( z \right)_{\text{off}} \right]_{\sigma \nu} }_{\mathfrak{B}\left( \chi_{\sigma}, \chi_{\nu} \right)} \leq \left( \underset{i}{\max} \, m_i^2 \right) \frac{\abs{g}}{\sqrt{\abs{z}}}.
    \end{equation}
\end{proof}

\begin{remark}
    \justifying
    Summarizing, set $K \doteq \max \left[ \left( \underset{i}{\max} \, m_i^{\frac{3}{2}}, \, \underset{i}{\max} \, m_i^2 \right) \right]$,
    \begin{equation*}
        \norm{ \left[ \Lambda_{\epsilon} \left( z \right)_{\text{off}} \right]_{\sigma \nu} }_{\mathfrak{B}\left( \chi_{\sigma}, \chi_{\nu} \right)} \leq K \frac{\abs{g}}{\sqrt{\abs{z}}},
    \end{equation*}
    \noindent for all $\epsilon > 0, \, z < 0, \, \sigma \neq \nu$, hence
    \begin{equation*}
        \underset{\sigma,\, \nu}{\max} \, \norm{ \left[ \Lambda_{\epsilon} \left( z \right)_{\text{off}} \right]_{\sigma \nu} }_{\mathfrak{B}\left( \chi_{\sigma}, \chi_{\nu} \right)} \leq K \frac{\abs{g}}{\sqrt{\abs{z}}}.
    \end{equation*}
    By also taking into account the diagonal contribution and denoting by $\norm{\cdot}_{\oplus}$ the operator norm of $\mathfrak{B}\left( \chi \right)$,
    \begin{align*}
        \norm{ \left[ \Lambda_{\epsilon}\left( z \right)_{\text{diag}} \right]^{-1} \left[ \Lambda_{\epsilon}\left(z\right)_{\text{off}} \right] }_{\oplus} & \leq \norm{ \left[ \Lambda_{\epsilon}\left( z \right)_{\text{diag}} \right]^{-1} }_{\oplus} \, \norm{ \left[ \Lambda_{\epsilon}\left(z\right)_{\text{off}} \right] }_{\oplus} \leq \\
        & \leq \frac{n \left( n-1 \right)}{2} \cdot \left[1 - \frac{ \mathfrak{C}\abs{g} }{ \sqrt{\abs{z}} } \right]^{-1} \cdot K \cdot \frac{\abs{g}}{ \sqrt{ \abs{z} } },
    \end{align*}
    \noindent resulting in $\norm{ \left[ \Lambda_{\epsilon}\left( z \right)_{\text{diag}} \right]^{-1} \left[ \Lambda_{\epsilon}\left(z\right)_{\text{off}} \right] }_{\oplus} < 1$ as long as $ z < - g^2 \left[ \frac{n\left(n-1\right)}{2} K + \mathfrak{C} \right]^2 \doteq z_0 $. Consequently, $ \left\{ \mathds{1} + \left[ \Lambda_{\epsilon}\left( z \right)_{\text{diag}} \right]^{-1} \left[ \Lambda_{\epsilon}\left(z\right)_{\text{off}} \right] \right\} $ is invertible in $\mathfrak{B}\left( \chi \right)$, for all $\epsilon > 0$. $\hfill \square$
\end{remark}

\subsection{Computing $\Lambda_{\epsilon}\left(z\right)_{\text{off}}$ as $\epsilon \downarrow 0$}

\vspace{2mm}
$\boxed{ \sigma=(12), \, \nu = (1\nu_2), \, 3 \leq \nu_2 \leq n }$

\begin{theorem}
    For all $z<0$, given
    \begin{equation*}
        \left[ \Lambda_{0, \, \underline{P}_{\nu}}\left( z \right)_{\text{off}} \right]_{\sigma \nu}: \, L^2\left( \mathbb{R}^3,dr_{\nu} dR_{\nu} dx_2 \right) \otimes \Tilde{\chi}^{-}_{\nu} \longrightarrow L^2\left( \mathbb{R}^3, dr_{\sigma} dR_{\sigma} dx_{\nu_2} \right) \otimes \Tilde{\chi}^{-}_{\nu}
    \end{equation*}
    \noindent defined by
    \begin{equation*}
        \left(\left[ \Lambda_{0, \, \underline{P}_{\nu}}\left( z \right)_{\text{off}} \right]_{\sigma \nu} \psi\right) \left( r_{\sigma}, R_{\sigma}, x_{\nu_2}, \underline{P}_{\nu} \right) \doteq C v\left(r_{\sigma}\right) \int_{\mathbb{R}^3} \, dr_{\nu}^{\prime} dR_{\nu}^{\prime} dx_2^{\prime} \;    G_{z-Q_{\nu}}^{(3)} \left( X_{\sigma \nu, \, 0} \right) v\left(r_{\nu}^{\prime}\right) \psi\left( r_{\nu}^{\prime}, R_{\nu}^{\prime}, x_{2}^{\prime}, \underline{P}_{\nu} \right),
    \end{equation*}
    \noindent for all $\psi \in L^2\left( \mathbb{R}^3,dr_{\nu} dR_{\nu} dx_2 \right) \otimes \Tilde{\chi}^{-}_{\nu}$, where
    \begin{equation*}
        X_{\sigma \nu, \, 0} = 
        \left(\begin{matrix}
            & \sqrt{2m_1} \left(R_{\sigma} - R_{\nu}^{\prime}\right)\\
            & \sqrt{2m_2} \left( R_{\sigma} - x_2^{\prime} \right)\\
            & \sqrt{2m_{\nu_2}} \left( x_{\nu_2} - R_{\nu}^{\prime} \right)
        \end{matrix} \right), \quad Q_{\nu} = \sum_{\underset{k \neq \nu_2}{k=3}}^n \, \frac{p_k^2}{2m_k}, \quad C = - \left( 2 \right)^{\frac{3}{2}} g \sqrt{m_1m_2m_{\nu_2}},
    \end{equation*}
    
    \begin{equation*}
        \underset{\epsilon \downarrow 0}{\lim} \; \norm{ \left[ \Lambda_{\epsilon, \underline{P}_{\nu}} \left( z \right)_{\text{off}} \right]_{\sigma \nu} - \left[ \Lambda_{0, \, \underline{P}_{\nu}} \left( z \right)_{\text{off}} \right]_{\sigma \nu}}_{\mathfrak{B}\left( \chi_{\sigma}, \chi_{\nu} \right)} = 0.
    \end{equation*}
\end{theorem}

\begin{proof}
    Let $\eta \in L^2\left( \mathbb{R}^3,dr_{\nu}dR_{\nu}x_2 \right), \, \xi \in \Tilde{\chi}^{-}_{\nu}$ be arbitrary.
    \small{\begin{align*}
        & \frac{1}{C^2} \, \norm{ \left\{ \left[ \Lambda_{\epsilon, \underline{P}_{\nu}} \left( z \right)_{\text{off}} \right]_{\sigma \nu} - \left[ \Lambda_{0, \, \underline{P}_{\nu}} \left( z \right)_{\text{off}} \right]_{\sigma \nu} \right\} \eta \otimes \xi }_2^2 = \\
        & = \int_{\mathbb{R}^n} \, d\underline{P}_{\nu}dr_{\sigma}dR_{\sigma}dx_{\nu_2} \, \abs{ \int_{\mathbb{R}^3} \, dr_{\nu}^{\prime}dR_{\nu}^{\prime}dx_2^{\prime} \, v\left(r_{\sigma}\right) \left[ G_{z-Q_{\nu}}^{(3)}\left( X_{\sigma \nu, \epsilon} \right) - G_{z-Q_{\nu}}^{(3)}\left( X_{\sigma \nu, \, 0} \right) \right] v\left(r_{\nu}^{\prime}\right) \eta\left( r_{\nu}^{\prime}, R_{\nu}^{\prime},x_2^{\prime} \right) \xi\left( \underline{P}_{\nu} \right)}^2 \\
        \begin{split}
            & = \int_{\mathbb{R}^n} \, d\underline{P}_{\nu}dr_{\sigma}dR_{\sigma}dx_{\nu_2} \, V\left( r_{\sigma} \right) \bigg\lvert \int_{\mathbb{R}} \,dr_{\nu}^{\prime} \, v\left(r_{\nu}^{\prime} \right) \cdot \\
            &\phantom{ = {= \int_{\mathbb{R}^n} \, d\underline{P}_{\nu}dr_{\sigma}dR_{\sigma}dx_2 \, V\left( r_{\sigma} \right) \bigg\lvert} } \cdot \int_{\mathbb{R}^2} \, dR_{\nu}^{\prime}dx_2^{\prime} \left[ G_{z-Q_{\nu}}^{(3)}\left( X_{\sigma \nu, \epsilon} \right) - G_{z-Q_{\nu}}^{(3)}\left( X_{\sigma \nu, \, 0} \right) \right] \eta\left( r_{\nu}^{\prime}, R_{\nu}^{\prime},x_2^{\prime} \right) \xi\left( \underline{P}_{\nu} \right) \bigg\rvert ^2
        \end{split}\\
        \begin{split}
            & \leq \int_{\mathbb{R}^n} \, d\underline{P}_{\nu}dr_{\sigma}dR_{\sigma}dx_{\nu_2} \, V\left( r_{\sigma} \right) \left\{ \int_{\mathbb{R}} \,dr_{\nu}^{\prime} \, v\left(r_{\nu}^{\prime} \right) \frac{1 + \abs{r_{\nu}^{\prime}}^\frac{1}{2}}{1 + \abs{r_{\nu}^{\prime}}^\frac{1}{2}} \cdot \right. \\
            &\phantom{ = { = \int_{\mathbb{R}^n} \, d\underline{P}_{\nu}dr_{\sigma}dR_{\sigma}dx_2 \, V\left( r_{\sigma} \right) \left\{ \right.} } \left. \cdot \abs{ \int_{\mathbb{R}^2} \, dR_{\nu}^{\prime}dx_2^{\prime} \left[ G_{z-Q_{\nu}}^{(3)}\left( X_{\sigma \nu, \epsilon} \right) - G_{z-Q_{\nu}}^{(3)}\left( X_{\sigma \nu, \, 0} \right) \right] \eta\left( r_{\nu}^{\prime}, R_{\nu}^{\prime},x_2^{\prime} \right) \xi\left( \underline{P}_{\nu} \right) } \right\}^2
        \end{split}\\
        \begin{split}
            & \leq 2 \int_{\mathbb{R}^n} \, d\underline{P}_{\nu}dr_{\sigma}dR_{\sigma}dx_{\nu_2} \, V\left( r_{\sigma} \right) \left[ \int_{\mathbb{R}} \, dr_{\nu}^{\prime} \left(1 + \abs{r_{\nu}^{\prime}}\right) V\left( r_{\nu}^{\prime} \right) \right] \cdot \\
            &\phantom{ = {\leq \int_{\mathbb{R}^n} \, d\underline{P}_{\nu}dr_{\sigma} \left[\right.} } \cdot \int_{\mathbb{R}} dr_{\nu}^{\prime} \frac{1}{1 + \abs{r_{\nu}^{\prime}}} \, \abs{ \int_{\mathbb{R}^2} \, dR_{\nu}^{\prime}dx_2^{\prime} \left[ G_{z-Q_{\nu}}^{(3)}\left( X_{\sigma \nu, \epsilon} \right) - G_{z-Q_{\nu}}^{(3)}\left( X_{\sigma \nu, \, 0} \right) \right] \eta\left( r_{\nu}^{\prime}, R_{\nu}^{\prime},x_2^{\prime} \right) \xi\left( \underline{P}_{\nu} \right) }^2
        \end{split}\\
        \begin{split}
            & \equiv I\left(V,\frac{1}{2}\right) \int_{\mathbb{R}^n} \, d\underline{P}_{\nu} dR_{\sigma} dx_{\nu_2} \, \int_{\mathbb{R}^2} dr_{\sigma} dr_{\nu}^{\prime} \, \left[\frac{V\left( r_{\sigma} \right)}{1 + \abs{r_{\nu}^{\prime}}}\right] \cdot \\
            &\phantom{ = { \equiv \int_{\mathbb{R}^n} \, d\underline{P}_{\nu} dR_{\sigma} dx_2 \, \int_{\mathbb{R}^2} dr_{\sigma} dr_{\nu}^{\prime} } } \, \cdot \left[ \int_{\mathbb{R}^2} dR_{\nu}^{\prime}dx_2^{\prime}  \abs{G_{z-Q_{\nu}}^{(3)}\left( X_{\sigma \nu,\epsilon} \right) - G_{z-Q_{\nu}}^{(3)}\left( X_{\sigma, \nu, \, 0} \right)} \abs{\eta\left( r_{\nu}^{\prime}, R_{\nu}^{\prime}, x_2^{\prime} \right) \xi\left( \underline{P}_{\nu} \right)} \right]^2
        \end{split}\\
        \begin{split}
            & \leq I\left(V,\frac{1}{2}\right) \int_{\mathbb{R}^n} \, d\underline{P}_{\nu} dR_{\sigma} dx_{\nu_2} \, \int_{\mathbb{R}^2} dr_{\sigma} dr_{\nu}^{\prime} \, \left[\frac{V\left( r_{\sigma} \right)}{1 + \abs{r_{\nu}^{\prime}}} \right] \cdot \\
            &\phantom{ = { \equiv \int_{\mathbb{R}^n} \, d\underline{P}_{\nu} dR_{\sigma} dx_2 \, \int_{\mathbb{R}^2} dr_{\sigma} dr_{\nu}^{\prime} } } \, \cdot \left[ \int_{\mathbb{R}^2} dR_{\nu}^{\prime}dx_2^{\prime}  \abs{G_{z}^{(3)}\left( X_{\sigma \nu,\epsilon} \right) - G_{z}^{(3)}\left( X_{\sigma, \nu, \, 0} \right)} \abs{\eta\left( r_{\nu}^{\prime}, R_{\nu}^{\prime}, x_2^{\prime} \right) \xi\left( \underline{P}_{\nu} \right)} \right]^2
        \end{split}\\
        \begin{split}
            & \equiv I\left(V,\frac{1}{2}\right) \left[ \int_{\mathbb{R}^{n-3}} d\underline{P}_{\nu} \, \abs{\xi \left( \underline{P}_{\nu} \right)}^2 \right] \left\{\int_{\mathbb{R}^2} \, dr_{\sigma}dr_{\nu}^{\prime} \, \left[ \frac{V\left(r_{\sigma}\right)}{1 + \abs{r_{\nu}^{\prime}}} \right] \cdot \right.\\
            &\phantom{ = { \equiv I\left(V,s\right) \left[ \int_{\mathbb{R}^{n-3}} d\underline{P}_{\nu} \, \abs{\xi \left( \underline{P}_{\nu} \right)}^2 \right] } } \left. \cdot \int_{\mathbb{R}^2} \, dR_{\sigma} dx_{\nu_2} \left[ \int_{\mathbb{R}^2} dR_{\nu}^{\prime}dx_2^{\prime} \, \abs{ G_{z}^{(3)}\left(X_{\sigma \nu, \epsilon}\right) - G_{z}^{(3)}\left(X_{\sigma \nu, \, 0}\right) } \abs{\eta\left( r_{\nu}^{\prime}, R_{\nu}^{\prime},x_2^{\prime} \right)} \right]^2 \right\}
        \end{split}
    \end{align*}}
    \normalsize
    \noindent where
    \begin{equation*}
        I\left( V, \frac{1}{2} \right) = 2 \int_{\mathbb{R}} dr \, \left( 1 + \abs{r} \right) V\left(r\right) < \infty.
    \end{equation*}
    \noindent Let then $F_{r_{\sigma}, r_{\nu}^{\prime}, \epsilon}: L^2\left(\mathbb{R}^2,dR_{\nu}dx_2\right) \longrightarrow L^2\left(\mathbb{R}^2,dR_{\sigma}dx_{\nu_2}\right)$ be defined by
    \begin{equation*}
        \left[F_{r_{\sigma}, r_{\nu}^{\prime}, \epsilon}\varphi\right] \left( R_{\sigma},x_{\nu_2} \right) \doteq \int_{\mathbb{R}^2} dR_{\nu}^{\prime}dx_2^{\prime} \, \left[ G_{z}^{(3)}\left(X_{\sigma \nu, \, \epsilon}\right) - G_{z}^{(3)}\left(X_{\sigma \nu, \, 0}\right) \right]\varphi\left( R_{\nu}^{\prime},x_{2}^{\prime} \right).
    \end{equation*}
    \noindent The Schur test is going to be used to ascertain whether it is bounded or not. By introducing
    \begin{equation*}
        \Tilde{X}_{\sigma \nu,\epsilon} = 
        \left( 
        \begin{matrix}
            \sqrt{2m_1} \left[ R_{\sigma} - R_{\nu}^{\prime} - \epsilon \left( \frac{m_2}{m_1+m_2} r_{\sigma} - \frac{m_{\nu_2}}{m_{\nu_1} + m_{\nu2}} r_{\nu}^{\prime} \right) \right] \\
            \sqrt{2m_2} \left( R_{\sigma} - x_2^{\prime} \right) \\
            \sqrt{2m_{\nu_2}} \left( x_{\nu_2} - R_{\nu}^{\prime} \right)
        \end{matrix}
        \right),
    \end{equation*}
    \noindent the following procedure is adopted.
    \begin{align*}
        & \int_{\mathbb{R}^2} \, dR_{\nu}^{\prime}dx_{2}^{\prime} \, \abs{G_{z}^{(3)}\left( X_{\sigma \nu, \epsilon} \right) - G_{z}^{(3)}\left( X_{\sigma \nu, \, 0} \right) } \leq \int_{\mathbb{R}^2} \, dR_{\nu}^{\prime}dx_{2}^{\prime} \, \abs{G_{z}^{(3)}\left( X_{\sigma \nu, \epsilon} \right) - G_{z}^{(3)}\left( \Tilde{X}_{\sigma \nu, \epsilon} \right) } + \\
        & + \int_{\mathbb{R}^2} \, dR_{\nu}^{\prime}dx_{2}^{\prime} \abs{G_{z}^{(3)} \left( \Tilde{X}_{\sigma \nu, \epsilon} \right) - G_{z}^{(3)} \left( X_{\sigma \nu, \, 0} \right)} \equiv \boxed{A} + \boxed{B} \, .
    \end{align*}
    \noindent By observing that
    
    \small{\begin{align*}
        & \abs{ G_{z}^{(3)}\left(X_{\sigma \nu, \epsilon}\right) - G_{z}^{(3)}\left(\tilde{X}_{\sigma \nu,\epsilon} \right) } = \\
        & \bigg\lvert \int_0^{\infty} \, \left\{ e^{ - \frac{ \left\{ \sqrt{2m_1} \left[ R_{\sigma} - R_{\nu}^{\prime} - \epsilon \left( \frac{m_2}{m_1+m_2}r_{\sigma} - \frac{m_{\nu_2}}{m_{\nu_1} + m_{\nu_2}}r_{\nu}^{\prime} \right) \right] \right\}^2 + \left\{ \sqrt{2m_2} \left( R_{\sigma} - x_2^{\prime} + \epsilon \frac{m_1}{m_1+m_2}r_{\sigma} \right) \right\}^2 + \left\{ \sqrt{2m_{\nu_2}} \left[ x_{\nu_2} - R_{\nu}^{\prime} + \epsilon \left( \frac{m_{\nu_1}}{m_{\nu_1} + m_{\nu_2}} \right)r_{\nu}^{\prime} \right] \right\}^2 }{4t} + zt } \right. \\
        & \left. - e^{ - \frac{ \left\{ \sqrt{2m_1} \left[ R_{\sigma} - R_{\nu}^{\prime} - \epsilon \left( \frac{m_2}{m_1+m_2} r_{\sigma} - \frac{m_{\nu_2}}{m_{\nu_1} + m_{\nu_2}} r_{\nu}^{\prime} \right) \right] \right\}^2 + \left\{ \sqrt{2m_2} \left( R_{\sigma} - x_2^{\prime} \right) \right\}^2 + \left\{ \sqrt{2m_{\nu_2}} \left( x_{\nu_2} - R_{\nu}^{\prime} \right) \right\}^2 }{4t} + zt } \right\} \frac{dt}{ \left( 4\pi t \right)^{\frac{3}{2}}} \bigg\rvert \leq \\
        \begin{split}
            & \leq \bigg\lvert \int_0^{\infty} \, \left[ e^{ - \frac{ \left[ \sqrt{2m_2} \left( R_{\sigma} - x_2^{\prime} + \epsilon \frac{m_1}{m_1+m_2}r_{\sigma} \right) \right]^2 + \left\{ \sqrt{2m_{\nu_2}} \left[ x_{\nu_2} - R_{\nu}^{\prime} + \epsilon \left( \frac{m_{\nu_1}}{m_{\nu_1} + m_{\nu_2}} \right) r_{\nu}^{\prime} \right] \right\}^2 }{4t} + zt } + \right. \\
            &\phantom{ = { \bigg\lvert \int_0^{\infty} \, \left[\right. } } \left. - e^{ - \frac{ \left[ \sqrt{2m_2} \left( R_{\sigma} - x_2^{\prime} \right) \right]^2 + \left[ \sqrt{2m_{\nu_2}} \left( x_{\nu_2} - R_{\nu}^{\prime} \right) \right]^2 }{4t} + zt } \right] \frac{dt}{\left( 4\pi t \right)^{\frac{3}{2}}} \bigg\rvert \leq 
        \end{split}\\
        \begin{split}
            & \leq \bigg\lvert G_{z}^{(3)} \left( 0, \, \sqrt{2m_{2}} \left[ R_{\sigma} - x_2^{\prime} + \epsilon \left( \frac{m_1}{m_1+m_2} \right) r_{\sigma} \right], \sqrt{2m_{\nu_2}} \left[ x_{\nu_2} - R_{\nu}^{\prime} + \epsilon \left( \frac{m_{\nu_1}}{m_{\nu_1} + m_{\nu_2}} \right) r_{\nu}^{\prime} \right] \right) + \\
            &\phantom{ = { \leq \bigg\lvert G_{z}^{(3)} \left( 0, \, \sqrt{2m_{2}} \left[ R_{\sigma} - x_2^{\prime} + \epsilon \left( \frac{m_1}{m_1+m_2} \right) r_{\sigma} \right], \right. } } - G_{z}^{(3)} \left( 0, \, \sqrt{2m_{2}} \left( R_{\sigma} - x_2^{\prime} \right), \sqrt{2m_{\nu_2}} \left( x_{\nu_2} - R_{\nu}^{\prime} \right) \right) \bigg\rvert,
        \end{split}
    \end{align*}}
    
    \normalsize
    \noindent concerning $\boxed{A} \,$, what follows holds.
    
    \small{\begin{align*}
        & \int_{\mathbb{R}^2} dR_{\nu}^{\prime}dx_2^{\prime} \, \abs{ G_{z}^{(3)} \left( X_{\sigma \nu, \epsilon} \right) - G_{z}^{(3)} \left( \tilde{X}_{\sigma \nu, \epsilon} \right) } \leq \\
        \begin{split}
            & \leq \int_{\mathbb{R}^2} dR_{\nu}^{\prime}dx_2^{\prime} \, \bigg\lvert G_{z}^{(3)} \left( 0, \, \sqrt{2m_{2}} \left[ R_{\sigma} - x_2^{\prime} + \epsilon \left( \frac{m_1}{m_1+m_2} \right) r_{\sigma} \right], \sqrt{2m_{\nu_2}} \left[ x_{\nu_2} - R_{\nu}^{\prime} + \epsilon \left( \frac{m_{\nu_1}}{m_{\nu_1} + m_{\nu_2}} \right) r_{\nu}^{\prime} \right] \right) + \\
            &\phantom{ = { \leq \bigg\lvert G_{z}^{(3)} \left( 0, \, \sqrt{2m_{2}} \left[ R_{\sigma} - x_2^{\prime} + \epsilon \left( \frac{m_1}{m_1+m_2} \right) r_{\sigma} \right], \right. } } - G_{z}^{(3)} \left( 0, \, \sqrt{2m_{2}} \left( R_{\sigma} - x_2^{\prime} \right), \sqrt{2m_{\nu_2}} \left( x_{\nu_2} - R_{\nu}^{\prime} \right) \right) \bigg\rvert
        \end{split}\\
        \begin{split}
            & = \int_{\mathbb{R}^2} dR_{\nu}^{\prime}dx_2^{\prime} \bigg\lvert \int_0^{\infty} \, \left[ e^{ - \frac{ \left[ \sqrt{2m_2} \left( R_{\sigma} - x_2^{\prime} + \epsilon \frac{m_1}{m_1+m_2}r_{\sigma} \right) \right]^2 + \left\{ \sqrt{2m_{\nu_2}} \left[ x_{\nu_2} - R_{\nu}^{\prime} + \epsilon \left( \frac{m_{\nu_1}}{m_{\nu_1} + m_{\nu_2}} \right) r_{\nu}^{\prime} \right] \right\}^2 }{4t} + zt } + \right. \\
            &\phantom{ = { \bigg\lvert \int_0^{\infty} \, \left[\right. } } \left. - e^{ - \frac{ \left[ \sqrt{2m_2} \left( R_{\sigma} - x_2^{\prime} \right) \right]^2 + \left[ \sqrt{2m_{\nu_2}} \left( x_{\nu_2} - R_{\nu}^{\prime} \right) \right]^2 }{4t} + zt } \right] \frac{dt}{\left( 4\pi t \right)^{\frac{3}{2}}} \bigg\rvert. 
        \end{split}
    \end{align*}}
    
    \normalsize
    \noindent By using
    \begin{equation*}
        \begin{cases}
            \overline{x}_2^{\prime} & = \sqrt{2m_2} \, x_2^{\prime} \\
            \overline{R}_{\nu}^{\prime} & = \sqrt{2m_{\nu_2}} \,  R_{\nu}^{\prime}
        \end{cases}
    \end{equation*}
    \noindent the following estimate holds
    \begin{align*}
        \begin{split}
            & \boxed{A} \leq \int_{\mathbb{R}^2} \frac{d\overline{x}_2^{\prime} d\overline{R}_{\nu}^{\prime}}{\sqrt{2m_2}\sqrt{2m_{\nu_2}}} \, \bigg\lvert \int_0^{\infty} \, \left\{ e^{ - \frac{ \left[ \overline{x}_2^{\prime} - \sqrt{2m_2} \left( R_{\sigma} + \epsilon \frac{m_1}{m_1+m_2}r_{\sigma} \right) \right]^2 + \left[ \overline{R}_{\nu}^{\prime} - \sqrt{2m_{\nu_2}} \left( x_{\nu_2} + \epsilon \frac{m_{\nu_1}}{m_{\nu_1} + m_{\nu_2}} r_{\nu}^{\prime} \right) \right]^2 }{4t} + zt } + \right. \\
            &\phantom{={\boxed{A} \leq \int_{\mathbb{R}^2} \frac{d\overline{x}_2^{\prime} d\overline{R}_{\nu}^{\prime}}{\sqrt{2m_2}\sqrt{2m_{\nu_2}}} \, \bigg\lvert \int_0^{\infty} \, \left\{ \right.}} \left. - e^{ - \frac{ \left( \overline{x}_2^{\prime} - \sqrt{2m_2} R_{\sigma} \right)^2 + \left( \overline{R}_{\nu}^{\prime} - \sqrt{2m_{\nu_2}} x_{\nu_2} \right)^2 }{4t} + zt } \right\} \frac{dt}{\left( 4\pi t \right)^{\frac{3}{2}}} \bigg\rvert \equiv
        \end{split}\\
        & \equiv \left( 
        \begin{cases}
            x_2^{\prime} & = \overline{x}_2^{\prime} - \sqrt{2m_2}R_{\sigma}\\
            R_{\nu}^{\prime} & = \overline{R}_{\nu}^{\prime} - \sqrt{2m_{\nu_2}} x_{\nu_2}
        \end{cases}
        \right)\\
        \begin{split}
            & \equiv \int_{\mathbb{R}^2}  \frac{dx_2^{\prime}dR_{\nu}^{\prime}}{2 \sqrt{m_2m_{\nu_2}}} \, \bigg\lvert \, G_{z}^{(3)}\left( 0, x_2^{\prime} - \sqrt{2m_2} \left( \frac{\epsilon m_1}{m_1+m_2}\right)r_{\sigma}, \, R_{\nu}^{\prime} - \sqrt{2m_{\nu_2}} \left( \frac{\epsilon m_{\nu_1}}{m_{\nu_1}+m_{\nu_2}} \right) r_{\nu}^{\prime} \right) + \\
            & \phantom{ = { \equiv \int_{\mathbb{R}^2}  \frac{dx_2^{\prime}dR_{\nu}^{\prime}}{2 \sqrt{m_2m_{\nu_2}}} \, \bigg\lvert } } - G_{z}^{(3)} \left(0, x_2^{\prime}, R_{\nu}^{\prime} \right) \bigg\rvert.
        \end{split}
    \end{align*}
    
    \noindent What obtained mimics the structure of what reported in \cite{05_GHL}, Proposition 4.5; analogous arguments hold true for $\boxed{B}$ all the same, hence $F_{r_{\sigma},r_{\nu}^{\prime},\epsilon}$ is a bounded operator and 
    \begin{equation*}
        \underset{\epsilon \downarrow 0}{\lim} \; \norm{ \left[ \Lambda_{\epsilon, \underline{P}_{\nu}} \left( z \right)_{\text{off}} \right]_{\sigma \nu} - \left[ \Lambda_{0, \, \underline{P}_{\nu}} \left( z \right)_{\text{off}} \right]_{\sigma \nu}}_{\mathfrak{B}\left( \chi_{\sigma}, \chi_{\nu} \right)} = 0.
    \end{equation*}
\end{proof}

\begin{remark}
    Similarly proven results hold for $\sigma = (12), \, \nu = (2\nu_2), \, \nu_2 \geq 3$. $\hfill \square$ 
\end{remark}

\newpage
\noindent $\boxed{ \sigma = ( 12 ), \nu = (\nu_1 \nu_2), \, 3 \leq \nu_1 < \nu_2 \leq n }$

\begin{theorem}
    For all $z<0$, given
    \begin{equation*}
        \left[ \Lambda_{0, \underline{P}_{\nu}}\left(z\right)_{\text{off}} \right]_{\sigma \nu}: L^2\left( \mathbb{R}^4, dx_1dx_2dR_{\nu}dr_{\nu} \right) \otimes \Tilde{\chi}^{-}_{\nu} \longrightarrow L^2\left( \mathbb{R}^4,dr_{\sigma}dR_{\sigma}dx_{\nu_1}dx_{\nu_2} \right) \otimes \Tilde{\chi}^{-}_{\nu}
    \end{equation*}
    \noindent defined by
    \small{\begin{equation*}
        \left( \left[\Lambda_{0,\underline{P}_{\nu}}\left( z \right)_{\text{off}} \right]_{\sigma \nu} \psi \right)\left( r_{\sigma}, R_{\sigma}, x_{\nu_1}, x_{\nu_2} \right) \doteq C v\left( r_{\sigma} \right) \int_{\mathbb{R}^4} dr_{\nu}^{\prime}dR_{\nu}^{\prime}dx_{1}^{\prime}dx_2^{\prime} \, G_{z-Q_{\nu}}^{(4)}\left( X_{\sigma \nu, 0} \right) v\left( r_{\nu}^{\prime} \right) \psi\left( r_{\nu}^{\prime}, R_{\nu}^{\prime}, x_1^{\prime}, x_2^{\prime} \right),
    \end{equation*}}
    \normalsize
    
    \noindent with $\psi \in L^2\left( \mathbb{R}^4, dx_1dx_2dR_{\nu}dr_{\nu} \right) \otimes \Tilde{\chi}^{-}_{\nu}$ and
    \begin{equation*}
        X_{\sigma \nu,0} = \left(
        \begin{matrix}
            \sqrt{2m_1} \left( R_{\sigma} - x_1^{\prime} \right)\\
            \sqrt{2m_2} \left( R_{\sigma} - x_2^{\prime} \right)\\
            \sqrt{2m_{\nu_1}} \left( x_{\nu_1} - R_{\nu}^{\prime} \right)\\
            \sqrt{2m_{\nu_2}} \left( x_{\nu_2} - R_{\nu}^{\prime} \right)
        \end{matrix}
        \right), \quad Q_{\nu} = \sum_{ \underset{k \neq \nu_1, \nu_2}{k = 3} }^{n} \, \frac{p_k^2}{2m_k}, \quad C = - 4 g \sqrt{m_1m_2m_{\nu_1}m_{\nu_2}},
    \end{equation*}
    
    \begin{equation*}
        \underset{\epsilon \downarrow 0}{\lim} \; \norm{ \left[ \Lambda_{\epsilon, \underline{P}_{\nu}} \left( z \right)_{\text{off}} \right]_{\sigma \nu} - \left[ \Lambda_{0, \, \underline{P}_{\nu}} \left( z \right)_{\text{off}} \right]_{\sigma \nu}}_{\mathfrak{B}\left( \chi_{\sigma}, \chi_{\nu} \right)} = 0.
    \end{equation*}
\end{theorem}

\begin{proof}
    Let $\eta \in L^2\left( \mathbb{R}^4, dx_1dx_2dr_{\nu}dR_{\nu} \right), \xi \in \Tilde{\chi}^{-}_{\nu}$ be arbitrary.
    
    \small{\begin{align*}
        & \norm{ \left\{ \left[ \Lambda_{\epsilon, \underline{P}_{\nu}} \left( z \right)_{\text{off}} \right]_{\sigma \nu} - \left[ \Lambda_{0, \, \underline{P}_{\nu}} \left( z \right)_{\text{off}} \right]_{\sigma \nu} \right\} \eta \otimes \xi}_2^2 = \\
        \begin{split}
            & = \int_{\mathbb{R}^n} d\underline{P}_{\nu}dr_{\sigma}dR_{\sigma}dx_{\nu_1}dx_{\nu_2} \, \bigg\lvert C \int_{\mathbb{R}^4} dx_1^{\prime}dx_2^{\prime}dr_{\nu}^{\prime}dR_{\nu}^{\prime} \left\{ v\left(r_{\sigma}\right) v\left(r_{\nu}^{\prime}\right) \left[ G_{z-Q_{\nu}}^{(4)} \left(X_{\sigma \nu, \epsilon}\right) - G_{z-Q_{\nu}}^{(4)} \left(X_{\sigma \nu, \, 0}\right) \right] \right. \\
            &\phantom{ = {= \int_{\mathbb{R}^n} d\underline{P}_{\nu}dr_{\sigma}dR_{\sigma}dx_{\nu_1}dx_{\nu_2} \, \bigg\lvert C \int_{\mathbb{R}^4} dx_1^{\prime}dx_2^{\prime}dr_{\nu}^{\prime}dR_{\nu}^{\prime} \left\{ \right. } } \left. \eta\left( x_1^{\prime},x_2^{\prime},r_{\nu}^{\prime},R_{\nu}^{\prime} \right) \xi\left( \underline{P}_{\nu} \right) \right\} \bigg\rvert^2
        \end{split}\\
        \begin{split}
            & = C^2 \, \int_{\mathbb{R}^n} d\underline{P}_{\nu}dr_{\sigma}dR_{\sigma}dx_{\nu_1}dx_{\nu_2} \, V\left( r_{\sigma} \right) \bigg\lvert \int_{\mathbb{R}^4} dx_1^{\prime} dx_2^{\prime} dr_{\nu}^{\prime} dR_{\nu}^{\prime} \, \left\{ v\left( r_{\nu}^{\prime} \right) \left[ G_{z-Q_{\nu}}^{(4)} \left(X_{\sigma \nu, \epsilon}\right) - G_{z-Q_{\nu}}^{(4)} \left(X_{\sigma \nu, \, 0}\right) \right] \right. \\
            &\phantom{ = {= C^2 \, \int_{\mathbb{R}^n} d\underline{P}_{\nu}dr_{\sigma}dR_{\sigma}dx_{\nu_1}dx_{\nu_2} \, V\left( r_{\sigma} \right) \bigg\lvert \int_{\mathbb{R}^4} dx_1^{\prime} dx_2^{\prime} dr_{\nu}^{\prime} dR_{\nu}^{\prime} \, \left\{ \right. } } \left. \eta\left(x_1^{\prime}, x_2^{\prime}, r_{\nu}^{\prime}, R_{\nu}^{\prime} \right) \xi\left( \underline{P}_{\nu} \right) \right\} \bigg\rvert^2
        \end{split}\\
        \begin{split}
            & \leq C^2 \int_{\mathbb{R}^n} d\underline{P}_{\nu}dr_{\sigma}dR_{\sigma}dx_{\nu_1}dx_{\nu_2} \, V\left(r_{\sigma}\right) \left\{ \int_{\mathbb{R}} dr_{\nu}^{\prime} v\left(r_{\nu}^{\prime}\right) \frac{1 + \abs{r_{\nu}^{\prime}}^\frac{1}{2}}{1 + \abs{r_{\nu}^{\prime}}^\frac{1}{2}} \right. \\ 
            & \phantom{ = { \leq C^2 \int_{\mathbb{R}^n} d\underline{P}_{\nu} dr_{\sigma} dR_{\sigma} dx_{\nu_1}} } \left. \bigg\lvert \int_{\mathbb{R}^3} dR_{\nu}^{\prime} dx_2^{\prime} dx_1^{\prime} \left[ G_{z-Q_{\nu}}^{(4)} \left(X_{\sigma \nu, \epsilon}\right) - G_{z-Q_{\nu}}^{(4)} \left(X_{\sigma \nu, \, 0}\right) \right] \eta\left(x_1^{\prime}, x_2^{\prime}, r_{\nu}^{\prime}, R_{\nu}^{\prime} \right) \xi\left( \underline{P}_{\nu} \right) \bigg\rvert \right\}^2
        \end{split}\\
        \begin{split}
            & \leq C^2 I\left(V, \frac{1}{2}\right) \int_{\mathbb{R}^{n-1}} d\underline{P}_{\nu} dR_{\sigma} dx_{\nu_1} dx_{\nu_2} \, \left\{ \int_{\mathbb{R}^2} dr_{\nu}^{\prime} dr_{\sigma} \frac{V\left( r_{\sigma} \right)}{1 + \abs{r_{\nu}^{\prime}}} \right.  \cdot \\
            &\phantom{ = {\leq C^2} } \left. \cdot \bigg\lvert \int_{\mathbb{R}^3} \, dR_{\nu}^{\prime}dx_1^{\prime}dx_2^{\prime} \left[ G_{z-Q_{\nu}}^{(4)} \left(X_{\sigma \nu, \epsilon}\right) - G_{z-Q_{\nu}}^{(4)} \left(X_{\sigma \nu, \, 0}\right) \right] \eta\left(x_1^{\prime}, x_2^{\prime}, r_{\nu}^{\prime}, R_{\nu}^{\prime} \right) \xi\left( \underline{P}_{\nu} \right) \bigg\rvert^2 \right\} 
        \end{split}\\
        \begin{split}
            & \leq C^2 I\left(V, \frac{1}{2}\right) \left( \int_{\mathbb{R}^{n-4}} d\underline{P}_{\nu} \abs{\xi\left(\underline{P}_{\nu} \right)}^2 \right) \, \left\{ \int_{\mathbb{R}^2} dr_{\nu}^{\prime} dr_{\sigma} \frac{V\left( r_{\sigma} \right)}{1 + \abs{r_{\nu}^{\prime}}} \right.  \cdot \\
            &\phantom{ = {\leq C^2} } \left. \cdot \int_{\mathbb{R}^3} dR_{\sigma} dx_{\nu_1} dx_{\nu_2} \bigg\lvert \int_{\mathbb{R}^3} \, dR_{\nu}^{\prime}dx_1^{\prime}dx_2^{\prime} \left[ G_{z}^{(4)} \left(X_{\sigma \nu, \epsilon}\right) - G_{z}^{(4)} \left(X_{\sigma \nu, \, 0}\right) \right] \eta\left(x_1^{\prime}, x_2^{\prime}, r_{\nu}^{\prime}, R_{\nu}^{\prime} \right)  \bigg\rvert^2 \right\}
        \end{split}
    \end{align*}}
    
    \normalsize
    \noindent It is then considered the linear map $K: \, \psi \in L^2\left( \mathbb{R}^3,dx_1dx_2dR_{\nu} \right) \longmapsto K\psi \in L^2\left( \mathbb{R}^3, dx_{\nu_1}dx_{\nu_2}dR_{\sigma} \right)$ defined by
    \begin{equation*}
        \left(K\psi\right) \left( x_{\nu_1}, x_{\nu_2}, R_{\sigma} \right) = \int_{\mathbb{R}^3} \, dR_{\nu}^{\prime}dx_1^{\prime}dx_2^{\prime} \left[ G_{z-Q_{\nu}}^{(4)} \left(X_{\sigma \nu, \epsilon}\right) - G_{z-Q_{\nu}}^{(4)} \left(X_{\sigma \nu, \, 0}\right) \right] \psi\left(x_1^{\prime}, x_2^{\prime},  R_{\nu}^{\prime} \right).
    \end{equation*}
    \noindent By introducing the point
    \begin{equation*}
        \tilde{X}_{\sigma \nu, \epsilon} = \left( 
        \begin{matrix}
            \sqrt{2m_1} \left( x_1^{\prime} - R_{\sigma} \right)\\
            \sqrt{2m_2} \left( x_2^{\prime} - R_{\sigma} \right)\\
            \sqrt{2m_{\nu_1}} \left( R_{\nu}^{\prime} - x_{\nu_1} \right)\\
            \sqrt{2m_{\nu_2}} \left[ R_{\nu}^{\prime} - x_{\nu_2} - \epsilon \left( \frac{m_{\nu_1}}{m_{\nu_1} + m_{\nu_2}} \right) r_{\nu}^{\prime} \right]
        \end{matrix} \right),
    \end{equation*}
    \noindent to check whether $K$ is bounded or not, the Schur test is referred to again.
    
    \begin{align*}
        & \int_{\mathbb{R}^3} dR_{\nu}^{\prime}dx_1^{\prime}dx_2^{\prime} \; \abs{G_{z-Q_{\nu}}^{(4)} \left(X_{\sigma \nu, \epsilon}\right) - G_{z-Q_{\nu}}^{(4)} \left(X_{\sigma \nu, \, 0}\right)} \leq \int_{\mathbb{R}^3} dR_{\nu}^{\prime}dx_1^{\prime}dx_2^{\prime} \; \abs{G_{z}^{(4)} \left(X_{\sigma \nu, \epsilon}\right) - G_{z}^{(4)} \left(X_{\sigma \nu, \, 0}\right)} \equiv \\ & = \int_{\mathbb{R}^3} dR_{\nu}^{\prime}dx_{1}^{\prime} dx_{2}^{\prime} \; \abs{ G_{z}^{(4)}\left( X_{\sigma \nu,\epsilon}  \right) - G_{z}^{(4)}\left( \tilde{X}_{\sigma \nu,\epsilon} \right) + G_{z}^{(4)}\left( \tilde{X}_{\sigma \nu,\epsilon} \right) - G_{z}^{(4)} \left( X_{\sigma \nu,\, 0} \right) } \leq \\
        & \leq \int_{\mathbb{R}^3} \, dR_{\nu}^{\prime}dx_{1}^{\prime} dx_{2}^{\prime} \; \abs{ G_{z}^{(4)}\left( X_{\sigma \nu,\epsilon}  \right) - G_{z}^{(4)}\left( \tilde{X}_{\sigma \nu,\epsilon} \right) } + \int_{\mathbb{R}^3} dR_{\nu}^{\prime}dx_{1}^{\prime} dx_{2}^{\prime} \; \abs{ G_{z}^{(4)}\left( \tilde{X}_{\sigma \nu,\epsilon} \right) - G_{z}^{(4)} \left( X_{\sigma \nu,\, 0} \right) } \equiv \\
        & \equiv \boxed{A} + \boxed{B} \, .
    \end{align*}
    \noindent Then
    
    \small{\begin{align*}
        & \boxed{A} = \int_{\mathbb{R}^3} \, dR_{\nu}^{\prime}dx_1^{\prime}dx_2^{\prime} \; \abs{ G_{z}^{(4)}\left( X_{\sigma \nu,\epsilon}  \right) - G_{z}^{(4)}\left( \tilde{X}_{\sigma \nu,\epsilon} \right) } \leq \\
        \begin{split}
            & \leq \int_{\mathbb{R}^3} \frac{ dx_1^{\prime} dx_2^{\prime} dR_{\nu}^{\prime} }{\sqrt{6 m_1 m_2 m_{\nu_1}}} \; \bigg\lvert G_{z}^{(4)} \left( x_1^{\prime} + \epsilon \left( \frac{\sqrt{2m_1}m_2}{m_1 + m_2} \right) r_{\sigma}, \, x_2^{\prime} - \epsilon \left( \frac{\sqrt{2m_2} m_1}{m_1 + m_2} \right) r_{\sigma}, R_{\nu}^{\prime} + \epsilon \left( \frac{\sqrt{2m_{\nu_1}} m_{\nu_2} }{m_{\nu_1} + m_{\nu_2}} \right) r_{\nu}^{\prime}, 0 \right) + \\
            &\phantom{ = { \leq \int_{\mathbb{R}^3} \frac{ dx_1^{\prime} dx_2^{\prime} dR_{\nu}^{\prime} }{\sqrt{6 m_1 m_2 m_{\nu_1}}} \; \bigg\lvert } } - G_{z}^{(4)} \left( x_1^{\prime}, x_2^{\prime}, R_{\nu}^{\prime}, 0 \right) \bigg\rvert,
        \end{split}
    \end{align*}}

    \normalsize
    \noindent by having respectively used the coordinate transformations
    \begin{equation*}
        \begin{cases}
            \overline{x}_1^{\prime} = \sqrt{2m_1} x_1^{\prime}\\
            \overline{x}_2^{\prime} = \sqrt{2m_2} x_2^{\prime}\\
            \overline{R}_{\nu}^{\prime} = \sqrt{2m_{\nu_1}} R_{\nu}^{\prime}
        \end{cases} \quad \quad \text{and} \quad \quad
        \begin{cases}
            x_1^{\prime} = \overline{x}_1^{\prime} - \sqrt{2m_1}R_{\sigma}\\
            x_2^{\prime} = \overline{x}_2^{\prime} - \sqrt{2m_2}R_{\sigma}\\
            R_{\nu}^{\prime} = \overline{R}_{\nu}^{\prime} - \sqrt{2m_{\nu_1}} x_{\nu_1}
        \end{cases}.
    \end{equation*}
    \noindent From this point on, it is possible to proceed as in \cite{05_GHL}, Proposition 4.8, to eventually state that $K$ is a bounded operator. Analogous arguments apply to $\boxed{B}$, therefore 
    \begin{equation*}
        \underset{\epsilon \downarrow 0}{\lim} \; \norm{ \left[ \Lambda_{\epsilon, \underline{P}_{\nu}} \left( z \right)_{\text{off}} \right]_{\sigma \nu} - \left[ \Lambda_{0, \, \underline{P}_{\nu}} \left( z \right)_{\text{off}} \right]_{\sigma \nu}}_{\mathfrak{B}\left( \chi_{\sigma}, \chi_{\nu} \right)} = 0.
    \end{equation*}
\end{proof}

\begin{corollary}
    For all $z < z_0 \doteq - g^2 \left[ \frac{n(n-1)}{2} K + \mathfrak{C} \right]^2$, set $\Lambda_0 \left(z\right) \doteq \Lambda_0 \left( z \right)_{\text{diag}} + \Lambda_0\left( z \right)_{\text{off}}$,
    \begin{equation*}
        \underset{\epsilon \downarrow 0}{\lim} \; \Lambda_{\epsilon}\left( z \right)^{-1} = \Lambda_0\left(z\right)^{-1} \equiv \left\{ \mathds{1} + \left[\Lambda_0\left(z\right)_{\text{diag}} \right]^{-1} \Lambda_0\left(z\right)_{\text{off}} \right\}^{-1} \left[ \Lambda_0 \left(z\right)_{\text{diag}} \right]^{-1}
    \end{equation*}
    \noindent in $\mathfrak{B}\left( \chi \right)$. Consequently
    \begin{equation*}
        \underset{\epsilon \downarrow 0}{\lim} \; \left( H_{\epsilon} - z \mathds{1} \right)^{-1} = R_{H_{0}}\left(z\right) + g \sum_{\sigma, \nu \in \mathcal{I}} \, \left[S^{(\sigma)}\left( z \right)\right]^{\ast} \left[ \Lambda_0\left(z\right)^{-1} \right]_{\sigma \nu} \left[S^{(\nu)}\left(z\right)\right] \doteq R(z). 
    \end{equation*}
    $\hfill \blacksquare$
\end{corollary}

\begin{remark}
    \justifying
    By recalling the self-adjoint operator $\left( H, \mathcal{D}_{H} \right)$ introduced in Appendix 3, \cite{05_GHL}, Appendix C allows to state that $H_{\epsilon}$ converges to $H$ in the strong resolvent sense, as $\epsilon \downarrow 0$. Consequently, as long as $z < z_0$, $R_{H}\left( z \right) = \left( H - z \mathds{1} \right)^{-1} = R(z)$, i.e. if $z<z_0$,
    \begin{equation*}
        \norm{ R_{H}\left( z \right) - R_{H_{\epsilon}}\left( z \right) } \underset{\epsilon \downarrow 0}{\longrightarrow} 0.
    \end{equation*}
    $\hfill \square$
\end{remark}

\begin{theorem}\label{Norm Resolvent Convergence}
    $H_{\epsilon} \underset{\epsilon \downarrow 0}{\longrightarrow} H$ in the norm resolvent sense.
\end{theorem}

\begin{proof}
    Let $z \in \left( - \infty, z_0 \right)$ be arbitrary and $\delta>0$ such that $\delta < \abs{z}$. Let then $\omega_{\pm} \doteq z \pm i\delta$ be; $H_{\epsilon},H$ are self-adjoint operators, for all $\epsilon > 0$, hence $\omega_{\pm} \in \rho\left(H_{\epsilon}\right) \cap \rho\left(H\right)$ for all $\epsilon > 0$ and $R_{H_{\epsilon}}\left( \omega_{\pm} \right) - R_{H}\left( \omega_{\pm} \right)$ makes sense. Eventually, the Neumann series expansion allows for
    \begin{align*}
        \norm{ R_{H_{\epsilon}}\left( \omega_{\pm} \right) - R_{H}\left( \omega_{\pm} \right) } & \leq \norm{ \sum_{n \in \mathbb{N}_0} \left( \omega_{\pm} - z \right)^{n} R_{H_{\epsilon}}\left( z \right)^{n+1} - \sum_{n \in \mathbb{N}_0} \left( \omega_{\pm} - z \right)^{n} R_{H}\left( z \right)^{n+1} } \leq \\
        & \leq \sum_{n \in \mathbb{N}_0} \delta^{n} \norm{ R_{H_{\epsilon}}\left( z \right)^{n+1} - R_{H}\left( z \right)^{n+1} } \underset{\epsilon \downarrow 0}{\longrightarrow} 0,
    \end{align*}
    because of the $n^{\text{th}}$-power function continuity. Repeating the process, the result holds for all $z \in \mathbb{C} \setminus \mathbb{R}$.
\end{proof}

\begin{corollary}\label{N-body result}
    The self-adjoint operator $\left(H, \mathcal{D}_{H}\right)$ is affiliated to $\mathcal{R}\left( \mathbb{R}^{2n},\sigma \right)$.
\end{corollary}

\begin{proof}
    By Proposition \ref{Norm Resolvent Convergence}, $\norm{R_{H_{\epsilon}}\left(z\right) - R_{H}\left(z\right)} \underset{\epsilon \downarrow 0}{\longrightarrow} 0$ for all $z \in i \mathbb{R} \setminus \left\{0\right\}$, i.e., because of \cite{04_BG} prop. 4.1, $\left(H - i \lambda \mathds{1}\right)^{-1} \in \pi_S\left[ \mathcal{R}\left(\mathbb{R}^{2n},\sigma\right) \right]$ for all $\lambda \in \mathbb{R}\setminus\{0\}$.
\end{proof}

\begin{theorem}\label{C*-dynamical system}
    Let $\mathfrak{K}_0$ be the $\text{C}^{\ast}-$subalgebra of $\pi_S\left[ \mathcal{R}\left( \mathbb{R}^{2n},\sigma \right) \right]$ generated by $\mathcal{B}_{\infty} \left( L^2\left( \mathbb{R}^n \right) \right)$ and the identity operator. $\left( \mathcal{K}_0 \equiv \pi_S^{-1}\left(\mathfrak{K}_0\right), \, \mathbb{R}, \, \beta \right)$, where
        \begin{equation*}
            \beta: \, t \in \mathbb{R} \longmapsto \beta_{t} \in \text{Aut}\left( \mathcal{K}_0 \right),
        \end{equation*}
    \noindent with
        \begin{equation*}
            \beta_{t}: \, a \in \mathcal{K}_0 \longmapsto \beta_{t}\left( a \right) \doteq \pi_S^{-1}\left[U(t)^{\ast} \, \pi_S\left(a\right) \, U(t)\right] \in \mathcal{K}_0,
        \end{equation*}
    \noindent and $U(t) \equiv \exp{ \left( -itH \right)}$, is a $\text{C}^{\ast}-$dynamical system.
\end{theorem}

\begin{proof}
    See \cite{06_M}, prop. 3.6.
\end{proof}

\begin{theorem}
    \justifying
    Let $f_0\left( H \right)$ be the (commutative) $\text{C}^{\ast}-$subalgebra of $\pi_S\left[ \mathcal{R}\left( \mathbb{R}^{2n},\sigma \right) \right]$ generated by $R_{H}\left( z \right)$, $z \in i\mathbb{R} \setminus \left\{ 0 \right\}$. Denoted by $\mathfrak{S}_0$ the $\text{C}^{\ast}-$subalgebra of $\pi_S\left[ \mathcal{R}\left( \mathbb{R}^{2n},\sigma \right) \right]$ generated by $f_0\left( H \right)$ and $\mathfrak{K}_0$, 
    \begin{equation*}
        e^{ itH } \, a \, e^{-itH} \in \mathfrak{S}_0, \quad \forall a \in \mathfrak{S}_0.
    \end{equation*}
\end{theorem}

\begin{proof}
    The result follows from \cite{05_GHL}, remark 1 and proposition \ref{C*-dynamical system}.
\end{proof}


\section*{Conclusions}
\justifying

This paper shows that the Hamiltonian of $n \in \mathbb{N}: n \geq 2$ distinguishable spinless, non-relativistic particles interacting via a two-body delta potential moving in one spatial dimension is affiliated to the resolvent algebra $\mathcal{R}\left( \mathbb{R}^{2n}, \sigma \right)$. Moreover, a $C^{\ast}-$dynamical system is singled out, together with a subalgebra $\pi_S^{-1}\left( \mathfrak{S}_0 \right)$ of $\mathcal{R}\left( \mathbb{R}^{2n}, \sigma \right)$, stable under Heisenberg time evolution. Nevertheless, the time evolution stability of the whole algebra is still an open problem.


\newpage
\section*{Appendix 1 - The Konno-Kuroda Formula}
\justifying

\begin{theorem}\label{Konno-Kuroda Theorem}
    \justifying
    Let $\mathcal{H}, \mathcal{K}$ be complex Hilbert spaces. Let $\left( H_0,\mathcal{D}_{H_0} \right)$ be a self-adjoint operator on $\mathcal{H}$ and let $A: \, \mathcal{H} \longrightarrow \mathcal{K}$ be a bounded operator. Given the self-adjoint operator $\left( H_g \doteq H_0 - gA^{\ast}A, \mathcal{D}_{H_0} \right)$ on $\mathcal{H}$, $g \in \mathbb{R}\setminus\{0\}$, for all $z \in \rho\left(H_0\right) \cap \rho\left(H_g\right)$,
    \begin{enumerate}
        \item $\left[ \mathds{1}_{\mathcal{K}} - \phi\left(z\right) \right]^{-1} = \mathds{1}_{\mathcal{K}} + M\left(z\right)$, with $\phi\left(z\right) = gAR_{H_0}\left(z\right)A^{\ast}$ and $M\left(z\right) = gAR_{H_g}\left(z\right)A^{\ast}$,
        \item $ R_{H_g}\left(z\right) = R_{H_0}\left(z\right) + g R_{H_0}\left(z\right)A^{\ast}\left[\mathds{1}_{\mathcal{K}} - \phi\left(z\right)\right]^{-1}AR_{H_0}\left(z\right)$.
    \end{enumerate}
\end{theorem}

\begin{proof}
    \justifying
    It is first observed that $A^{\ast}A$ is a bounded self-adjoint operator on $\mathcal{H}$, hence the same holds for $V \equiv g A^{\ast}A$. Particularly, for all $x \in \mathcal{D}_{H_0}$,
    \begin{equation*}
        \norm{Vx} \leq \norm{V} \norm{x} \leq \epsilon \norm{H_0x} + \norm{V} \norm{x}
    \end{equation*}
    \noindent for all $\epsilon \in \mathbb{R}^{+}_0$; consequently, $\left( H_g, \mathcal{D}_{H_0} \right)$ is self-adjoint by the Kato-Rellich theorem. Let then $z \in \rho\left(H_0\right) \cap \rho\left(H_g\right)$ be arbitrary.
    \begin{enumerate}
        \item By direct inspection, the second resolvent formula allows for
        \begin{equation*}
            \left[ \mathds{1}_{\mathcal{K}} - \phi\left(z\right) \right] \left[ \mathds{1}_{\mathcal{K}} + M\left(z\right) \right] = \mathds{1}_{\mathcal{K}} = \left[ \mathds{1}_{\mathcal{K}} + M\left(z\right) \right] \left[ \mathds{1}_{\mathcal{K}} - \phi\left(z\right) \right].
        \end{equation*}
        \item 
        \begin{align*}
            R_{H_g}\left(z\right) & = R_{H_0}\left(z\right) +  \left[R_{H_g}\left(z\right) - R_{H_0}\left(z\right)\right] = \left( \text{by the second resolvent formula} \right) \\
            & = R_{H_0}\left(z\right) + g R_{H_0}\left(z\right) A^{\ast} A R_{H_g}\left(z\right) + g R_{H_0}\left(z\right) A^{\ast} A R_{H_0}\left(z\right) - g R_{H_0}\left(z\right) A^{\ast} A R_{H_0}\left(z\right) = \\
            & = R_{H_0}\left(z\right) + g R_{H_0}\left(z\right) A^{\ast} A R_{H_0}\left(z\right) + g R_{H_0}\left(z\right) A^{\ast} A \left[R_{H_g}\left(z\right) - R_{H_0}\left(z\right)\right] = \\
            & = \left( \text{by the second resolvent formula again} \right) = \\
            & = R_{H_0}\left(z\right) + g R_{H_0}\left(z\right) A^{\ast} A R_{H_0}\left(z\right) + g^2 R_{H_0}\left(z\right) A^{\ast} A R_{H_g}\left(z\right) A^{\ast} A R_{H_0}\left(z\right) = \\
            & = R_{H_0}\left(z\right) + g R_{H_0}\left(z\right)A^{\ast} \left[ \mathds{1}_{\mathcal{K}} + M\left(z\right) \right]AR_{H_0}\left(z\right) \equiv \\ & = R_{H_0}\left(z\right) + g R_{H_0}\left(z\right)A^{\ast} \left[ \mathds{1}_{\mathcal{K}} - \phi\left(z\right) \right]^{-1} A R_{H_0}\left(z\right),
        \end{align*}
        \noindent allowing to express the resolvent of $\left(H_g,\mathcal{D}_{H_0}\right)$ at $z \in \rho\left(H_0\right) \cap \rho\left(H_g\right)$ in terms of $R_{H_0}\left(z\right)$, $A$ and $A^{\ast}$ only.
    \end{enumerate}
\end{proof}

\begin{corollary}
    \justifying
    Given $n \in \mathbb{N}$, let $\mathcal{H}, \mathcal{K}_{i}$, $i=1,\ldots,n$ be complex Hilbert spaces. Let $\left(H_0,\mathcal{D}_{H_0}\right)$ be a self-adjoint operator on $\mathcal{H}$ and let $A_i: \mathcal{H} \longrightarrow \mathcal{K}_i$, $i=1,\ldots,n$, be bounded operators. Given $g \in \mathbb{R} \setminus \{0\}$ and considered the self-adjoint operator $\left(H_g = H_0 - g\sum_{i=1}^{n} A_i^{\ast}A_i, \mathcal{D}_{H_0} \right)$ on $\mathcal{H}$, for all $z \in \rho\left(H_0\right) \cap \left(H_g\right)$,
    \begin{equation}\label{Konno-Kuroda Formula}
        R_{H_g}\left( z \right) = R_{H_0}\left( z \right) + g \sum_{i,j = 1}^n R_{H_0}\left(z\right) A_i^{\ast} \left[\Lambda\left(z\right)^{-1}\right]_{ij} A_j R_{H_0}\left(z\right)
    \end{equation}
    \noindent where $\Lambda\left(z\right)_{ij} \doteq \delta_{ij} - gA_iR_{H_0}\left(z\right)A_j^{\ast}: \mathcal{K}_j \longrightarrow \mathcal{K}_i$.
\end{corollary}

\begin{proof}
    \justifying
    By introducing the Hilbert space $\mathcal{K} \doteq \underset{i}{\bigoplus} \, \mathcal{K}_i$, let $A: \mathcal{H} \longrightarrow \mathcal{K}$ be the bounded operator such that $A\psi \doteq \left(A_i \psi\right)_{i=1}^n$, for all $\psi \in \mathcal{H}$. For all $\psi \in \mathcal{D}_{H_0}$,
    \begin{equation*}
        H_g\psi = H_0\psi - g \sum_{i=1}^n A_i^{\ast}A_i\psi \equiv H_0\psi -gA^{\ast}A\psi,
    \end{equation*}
    \noindent and Proposition \ref{Konno-Kuroda Theorem} can be applied. Straightforward computations allow to get formula (\ref{Konno-Kuroda Formula}).
\end{proof}


\section*{Appendix 2 - Trace Operator on Hyper-planes for Sobolev Functions and Related}
\justifying

\begin{lemma}
    \justifying
    Let $\psi \in C^{\infty}_{c} \left( \mathbb{R}^{n+1} \simeq \mathbb{R} \times \mathbb{R}^{n} \right)$ be a real function. For all $x \in \mathbb{R}$, $\psi_x \in C^{\infty}_{c} \left( \mathbb{R}^{n} \right)$, where
    \begin{equation*}
        \psi_{x}: \, \textbf{y} \in \mathbb{R}^{n} \longmapsto \psi_{x} \left( \textbf{y} \right) \doteq \psi \left(x, \textbf{y} \right) \in \mathbb{R}.
    \end{equation*}
\end{lemma}

\begin{proof}
    \justifying
    Let $K \equiv \text{supp}\, \psi$ be and let $\pi_{\mathbb{R}^n}\left(K\right)$ be the $\mathbb{R}^n$ projection of $K$; by definition of product topology, $\pi_{\mathbb{R}^n}\left(K\right)$ is compact in $\mathbb{R}^n$. If $\textbf{y} \in \mathbb{R}^n: \, \textbf{y} \notin \pi_{\mathbb{R}^n}\left( K \right)$, $\left( x, \textbf{y} \right) \notin K$, i.e. $\psi_{x}\left(\textbf{y}\right) = 0$. 
    
    \noindent Given whatever $\textbf{y} \in \mathbb{R}^n$, let $\left\{ \textbf{y}_m\right\}_m \subset \mathbb{R}^n$ be such that $\textbf{y}_m \underset{n}{\longrightarrow} \textbf{y}$; $\psi_x$ is continuous at $\textbf{y}$ if and only if $\underset{m}{\lim} \, \psi_x\left( \textbf{y}_m \right) = \psi_x\left( \textbf{y} \right) $. However, $\left(x,\textbf{y}_m\right) \underset{m}{\longrightarrow} \left(x,\textbf{y}\right)$, hence the continuity of $\psi$ implies $\psi_x\left( \textbf{y}_m \right) \equiv \psi\left( x, \textbf{y} \right) \underset{m}{\longrightarrow} \psi\left( x,\textbf{y} \right) \equiv \psi_x\left( \textbf{y} \right)$; in other words, $\psi_x$ is at least a continuous function of compact support. 

    \noindent Let then $\textbf{y} \in \mathbb{R}^n$ be arbitrary; for all $j = 1, \ldots, n$, what follows holds.
    \begin{equation*}
        \frac{ \psi_{x} \left( \textbf{y} + t \textbf{e}_{j} \right) - \psi_{x} \left( \textbf{y} \right) }{t} = \frac{\psi \left( x, \textbf{y} + t \textbf{e}_{j} \right) - \psi\left( x, \textbf{y} \right)}{t} \underset{t \rightarrow 0}{\longrightarrow} \left( \frac{\partial \psi}{\partial \textbf{f}_{j+1}} \right) \left(x, \textbf{y} \right)
    \end{equation*}
    i.e.
    \begin{equation*}
        \left(\frac{\partial \psi_{x}}{\partial \textbf{e}_{j}}\right) \left( \textbf{y} \right) = \left( \frac{\partial \psi}{ \partial \textbf{f}_{j+1} }\right) \left(x, \textbf{y} \right),
    \end{equation*}
    where
    \begin{equation*}
        \textbf{e}_{j} = \left( \underbrace{0, \ldots, \underbrace{1}_{j-th}, 0, \ldots, 0}_{n} \right), \quad \textbf{f}_{j+1} = \left(0, \textbf{e}_j \right). 
    \end{equation*}
    The continuity of $ \left( \partial_j \psi_x \right)$ is proved as above, hence induction gives $\psi_x \in C^{\infty}_{c} \left( \mathbb{R}^{n} \right)$.
\end{proof}

\begin{remark}
    Given $\psi_x$ as above, $\psi_x \in C^{\infty}_{c}\left(\mathbb{R}^n \right)$ implies $\psi_x \in L^2\left(\mathbb{R}^n \right)$, hence
    \begin{equation}\label{varphifrompsi}
        \varphi: \, x \in \mathbb{R} \longmapsto \varphi \left( x \right) \doteq \int_{\mathbb{R}^n} \, \abs{\psi_x \left(\textbf{y}\right)}^2 \, d\lambda^{(n)}\left( \textbf{y}\right) \equiv \int_{\mathbb{R}^n} \, \psi\left(x, \textbf{y}\right)^2 \, d\lambda^{(n)}\left(\textbf{y}\right) \in \mathbb{R}^+_{0}
    \end{equation}
    is well-defined. Particularly, set $K^{\prime} = \pi_{\mathbb{R}^n} \left(K\right)$,
    \begin{equation*}
        \varphi: \, x \in \mathbb{R} \longmapsto \varphi \left( x \right) \equiv \int_{K^\prime} \, \psi_{x}^2 \left( \textbf{y} \right) d\lambda^{(n)}\left( \textbf{y} \right).
    \end{equation*}
    $\hfill \square$
\end{remark}

\begin{lemma}
    \justifying
    Let $\psi \in C^{\infty}_{c} \left( \mathbb{R}^{n+1} \simeq \mathbb{R} \times \mathbb{R}^{n} \right) $ be a real function and $\varphi$ as in (\ref{varphifrompsi}), $\varphi \in C^{\infty}_{c} \left( \mathbb{R} \right)$.
\end{lemma}

\begin{proof}
    Let $\pi_{\mathbb{R}}\left(K\right)$ be the compact projection of $K \equiv \text{supp} \, \psi$ on the real line; if $x \notin \pi_{\mathbb{R}}\left( K \right)$, $\left(x, \textbf{y} \right) \notin K$ for all $\textbf{y} \in \mathbb{R}^n$, hence $\psi_x^2\left( \textbf{y} \right) = 0$ and, correspondingly, $\varphi\left(x\right) = 0$. 
    
    \noindent Given arbitrarily $x \in \mathbb{R}$, let $\left\{ x_m \right\}_m \subset \mathbb{R}$ be such that it converges to $x$. Since $\psi^2$ is continuous of compact support, there exists $C>0$ such that $\abs{ \psi^2_{x_m} \left( \textbf{y} \right)}  \leq C$ for all $\textbf{y} \in \mathbb{R}^n, \, m \in \mathbb{N}$. The dominated convergence theorem then implies
    \begin{align*}
        \underset{m}{\lim} \, \varphi\left( x_m \right) & = \underset{m}{\lim} \, \int_{K^{\prime}} \, \psi_{x_m}^2 \left( \textbf{y} \right) \, d\lambda^{(n)}\left( \textbf{y} \right) = \int_{K^{\prime}} \underset{m}{\lim} \left[\, \psi_{x_m}^2 \left( \textbf{y} \right)\right] \, d\lambda^{(n)}\left( \textbf{y} \right) = \\ 
        & = \int_{K^{\prime}} \, \psi_{x}^2 \left( \textbf{y} \right) \, d\lambda^{(n)}\left( \textbf{y} \right) = \varphi(x),
    \end{align*}
    i.e., the arbitrarity of $x \in \mathbb{R}$ gives that $\varphi$ is a continuous function of compact support. 

    \noindent $\left( \partial \psi_x^2 \right) \in C^{\infty}_{c}(\mathbb{R}^{1+n})$, hence uniformly bounded over $K$, giving
    \begin{equation*}
        \varphi^{\prime}\left( x \right) = \int_{K^{\prime}} \, \left( \frac{\partial \psi^2}{\partial x} \right)(x,\textbf{y}) \, d\lambda^{(n)}\left(\textbf{y}\right), \quad \forall x \in \mathbb{R}.
    \end{equation*}
    By repeating the process, $\varphi^{\prime}$ is a continuous function of compact support; inductively, $\varphi \in C_{0}^{\infty}\left( \mathbb{R} \right)$.
\end{proof}

\begin{lemma}
    Let $\psi \in C^{\infty}_{c}\left( \mathbb{R}^{1+n} \right)$ be a complex function.
    \begin{equation*}
        \underset{r \in \mathbb{R}}{\sup} \, \int_{\mathbb{R}^n} \, \abs{\psi\left( r, \textbf{x} \right)}^2 d\lambda^{(n)}\left( \textbf{x} \right) \leq \norm{\psi}^2_{H^1\left( \mathbb{R}^{n+1}\right)}
    \end{equation*}
    holds.
\end{lemma}

\begin{proof}
    Let $r \in \mathbb{R}$ be. Trivially,
    \begin{equation*}
        \int_{\mathbb{R}^n} \, \abs{\psi \left(r, \textbf{x} \right)}^2 \, d\lambda^{(n)}\left( \textbf{x} \right) = \int_{\mathbb{R}^n} \, {\psi}_{R}^2 \left(r, \textbf{x} \right) \, d\lambda^{(n)}\left( \textbf{x} \right) + \int_{\mathbb{R}^n} \, {\psi}_{I}^2 \left(r, \textbf{x} \right) \, d\lambda^{(n)}\left( \textbf{x} \right),
    \end{equation*}
    where ${\psi}_{R} \equiv \Re{\psi}$ and $\psi_{I} \equiv \Im{\psi}$. Then
    \begin{align*}
        & \int_{\mathbb{R}^n} \, \psi_i^2\left(r,\textbf{x}\right) \, d\lambda^{(n)}\left( \textbf{x} \right) \equiv \varphi\left(r\right) = \int_{-\infty}^{r} \, \varphi^{\prime}\left( s \right) d\lambda^{(1)}(s) \leq \int_{\mathbb{R}} \, \abs{\varphi^{\prime}(s)} d\lambda^{(1)}(s) \equiv \\
        & \equiv \int_{\mathbb{R}} \, \abs{ \left[ \frac{d}{dr} \int_{\mathbb{R}^n} \, \psi_i^2\left(r,\textbf{x}\right) \, d\lambda^{(n)}\left( \textbf{x} \right) \right](s)} \, d\lambda^{(1)}(s) \leq \\
        & \leq 2 \int_{\mathbb{R}^{n+1}} \, \abs{\psi_i\left( \textbf{y}\right) \left( \partial_r \psi_i \right)\left( \textbf{y} \right)} d\lambda^{(n+1)}(\textbf{y}) \leq \left( \text{by using H{\"o}lder inequality} \right) \leq 2 \norm{\psi_i}_2 \norm{\partial_r \psi_i}_2 \leq \\
        & \leq \norm{\psi_i}_2^2 + \norm{\left(\partial_r \psi\right)_i}_2^2
    \end{align*}
    with $i = R, I$, therefore
    \begin{align*}
        & \int_{\mathbb{R}^n} \, \abs{\psi\left( r, \textbf{x}\right)}^2 d\lambda^{(n)}(\textbf{x}) \leq \left( \norm{\psi_R}_2^2 + \norm{\psi_I}_2^2 \right) + \left( \norm{\left(\partial_r \psi\right)_R }_2^2  + \norm{\left( \partial_r \psi\right)_I}_2^2 \right) \equiv \\ & \equiv \norm{\psi}_2^2 + \norm{\partial_r \psi}_2^2 \leq \norm{\psi}_2^2 + \norm{ \, \norm{\nabla \psi}}_2^2 \equiv \norm{\psi}_{H^{1}\left( \mathbb{R}^{n+1} \right)}^2.
    \end{align*}
    The right hand side of the foregoing inequality is independent of $r \in \mathbb{R}$, hence
    \begin{equation*}
        \underset{r}{\sup} \, \int_{\mathbb{R}^n} \, \abs{\psi\left(r, \textbf{x} \right)}^2 d\lambda^{(n)}(\textbf{x}) \leq \norm{\psi}_{H^{1}\left( \mathbb{R}^{n+1} \right)}^2.  
    \end{equation*}
\end{proof}

\begin{theorem}
    \justifying
    The map $\Tilde{\tau}_0: \, \psi \in C^{\infty}_c\left( \mathbb{R}^{1+n} \right) \longmapsto \Tilde{\tau}_0 \psi \in L^2\left(\mathbb{R}^n\right)$, with $\left(\Tilde{\tau}_0 \psi \right)(\textbf{x}) = \psi(0,\textbf{x})$, for all $\textbf{x} \in \mathbb{R}^n$, results in a linear, densely defined bounded operator from $H^{1}\left( \mathbb{R}^{n+1} \right)$ to $L^2\left(\mathbb{R}^n\right)$.
\end{theorem}

\begin{proof}
    $\Tilde{\tau}_0$ is clearly well-defined, linear and densely defined. Concerning boundedness, one has
    \begin{align*}
        \norm{\Tilde{\tau}_0 \psi}_{L^2\left( \mathbb{R}^n \right)}^2 & = \int_{\mathbb{R}^n} \, \abs{\left( \Tilde{\tau}_0 \psi \right) \left(\textbf{x} \right)}^2 d\lambda^{(n)}\left( \textbf{x} \right) = \int_{\mathbb{R}^n} \, \abs{\psi\left(0,\textbf{x}\right)}^2 d\lambda^{(n)}(\textbf{x}) \leq \\ & \leq \underset{r \in \mathbb{R}}{\sup} \, \int_{\mathbb{R}^n} \, \abs{\psi\left(r, \textbf{x}\right)}^2 d\lambda^{(n)}(\textbf{x}) \leq \norm{\psi}_{H^1\left( \mathbb{R}^{n+1}\right)}^2.
    \end{align*}
\end{proof}

\begin{remark}
    \justifying
    The foregoing proposition allows for a bounded, norm-preserving extension $\tau_0$ of $\Tilde{\tau}_0$ to all $H^1\left( \mathbb{R}^{n+1} \right)$. $\hfill \square$
\end{remark}

\begin{definition}
    \justifying
    Given $n \in \mathbb{N}$ and $i,j \in \left\{ 1, \ldots , n \right\}: \; i < j$, by considering $U_{(ij)}$ as in (\ref{Coordinate transformation operator}), $\tau_{(ij)} \doteq \tau_0 U_{(ij)}$ denotes the corresponding \textbf{trace operator}. $\hfill \square$
\end{definition}

\begin{remark}
    \justifying
    Let $\varphi \in C^{\infty}_c\left( \mathbb{R}^{n+1} \right)$ be. It is observed that
    \begin{align*}
        \left( \tau_0\varphi \right) \left( \textbf{y} \right) = \varphi \left(0, \textbf{y} \right) = \frac{1}{\sqrt{2\pi}} \int_{\mathbb{R}} \, \left( \mathfrak{F}_{\mathbb{R}} \varphi \right) \left( p, \textbf{y} \right) e^{ip0} d\lambda^{(1)}\left(p\right) \equiv \frac{1}{\sqrt{2\pi}} \int_{\mathbb{R}} \, \left( \mathfrak{F}_{\mathbb{R}} \varphi \right) \left( p, \textbf{y} \right) d\lambda^{(1)}\left(p\right),
    \end{align*}
    \noindent where $\mathfrak{F}$ is the Fourier-Plancherel operator on $L^2\left( \mathbb{R},d\lambda^{(1)}\left(p\right)\right)$. As a consequence $\hfill \square$
\end{remark}

\begin{definition}
    \justifying
    Given $H^1\left(\mathbb{R},d\lambda^{(1)}\left(p\right)\right) = \mathfrak{F}H^1\left( \mathbb{R}, d\lambda^{(1)}\left(x\right) \right)$, the \textbf{Fourier trace operator}
    \begin{equation*}
        \hat{\tau}_0: \, H^1\left(\mathbb{R}, d\lambda^{(1)}\left(p\right) \right) \otimes H^1\left( \mathbb{R}^n, d\lambda^{(n)}\left( \textbf{y} \right) \right) \longrightarrow L^2\left( \mathbb{R}^{n}, d\lambda^{(n)}\left( \textbf{y} \right) \right)
    \end{equation*}
    is introduced, where
    \begin{equation*}
        \left(\hat{\tau}_0 \xi\right) \left( \textbf{y} \right) = \frac{1}{\sqrt{2\pi}} \int_{\mathbb{R}} \xi \left( p, \textbf{y} \right) d\lambda^{(1)}\left(p\right).
    \end{equation*}
    $\hfill \square$
\end{definition}

\begin{lemma}
    \justifying
    What follows holds.
    \begin{enumerate}
        \item $\mathfrak{F}_{\mathbb{R}^n} \tau_0 = \tau_0 \left( \mathds{1}_{\mathbb{R}} \otimes \mathfrak{F}_{\mathbb{R}^n} \right)$.
        \item $\mathfrak{F}_{\mathbb{R}^n} \tau_0 = \hat{\tau}_0 \mathfrak{F}_{\mathbb{R}^{n+1}}$.
        \item $\tau_0 = \hat{\tau}_0 \left( \mathfrak{F}_{\mathbb{R}} \otimes \mathds{1}_{\mathbb{R}^{n}} \right)$.
    \end{enumerate}
\end{lemma}

\begin{proof}
    \justifying
    Let $\varphi \in C^{\infty}_c\left( \mathbb{R}^{n+1} \right)$ be arbitrary.
    \begin{enumerate}
        \item For all $\textbf{P} \in \mathbb{R}^n$,
        \begin{align*}
            \left[ \mathfrak{F}_{\mathbb{R}^n} \left( \tau_0 \varphi \right) \right]\left( \textbf{P} \right) & = \int_{\mathbb{R}^n} \, \left( \tau_0 \varphi \right)\left( \textbf{x} \right) e^{-i \textbf{P} \cdot \textbf{x}} \frac{d \textbf{x}}{\left( 2\pi \right)^{\frac{n}{2}}} = \int_{\mathbb{R}^n} \, \varphi \left( 0, \textbf{x} \right) e^{-i \textbf{P} \cdot \textbf{x}} \frac{d \textbf{x}}{\left( 2\pi \right)^{\frac{n}{2}}} = \\
            & = \left[ \int_{\mathbb{R}^n} \, \varphi \left( y, \textbf{x} \right) e^{-i \textbf{P} \cdot \textbf{x}} \frac{d \textbf{x}}{\left( 2\pi \right)^{\frac{n}{2}}} \right]_{y = 0} = \left\{ \left[ \left( \mathds{1}_{\mathbb{R}} \otimes \mathfrak{F}_{\mathbb{R}^{n}} \right) \varphi \right] \left( y, \textbf{P} \right)  \right\}_{y=0} = \\
            & = \left\{ \left[\tau_0 \left( \mathds{1}_{\mathbb{R}} \otimes \mathfrak{F}_{\mathbb{R}^n} \right) \right] \varphi \right\} \left( \textbf{P} \right).
        \end{align*}
        \item For all $\textbf{P} \in \mathbb{R}^n$,
            \begin{align*}
                \left[ \mathfrak{F}_{\mathbb{R}^n} \left( \tau_0 \varphi \right) \right]\left( \textbf{P} \right) & = \int_{\mathbb{R}^n} \, \varphi \left( 0, \textbf{x} \right) e^{-i \textbf{P} \cdot \textbf{x}} \frac{d \textbf{x}}{\left( 2\pi \right)^{\frac{n}{2}}} = \int_{\mathbb{R}^n} \, \left[ \int_{\mathbb{R}} \, \left( \mathfrak{F}_{\mathbb{R}} \varphi \right) \left( p, \textbf{x} \right) \frac{dp}{\sqrt{2\pi}} \right] e^{-i \textbf{P} \cdot \textbf{x}} \frac{d\textbf{x}}{\left( 2\pi \right)^{\frac{n}{2}}} = \\
                & = \int_{\mathbb{R}} \left[ \mathfrak{F}_{\mathbb{R}^n} \left( \mathfrak{F}_{\mathbb{R}}\varphi \right) \right] \left(p, \textbf{P}\right) \frac{dp}{\sqrt{2\pi}} \equiv \int_{\mathbb{R}} \, \left(\mathfrak{F}_{\mathbb{R}^{n+1}} \varphi\right)\left(p, \textbf{P}\right) \frac{dp}{\sqrt{2\pi}} \equiv \\
                & \equiv \left[ \hat{\tau}_0 \left( \mathfrak{F}_{\mathbb{R}^{n+1}} \varphi \right) \right]\left( \textbf{P} \right)
            \end{align*}
            \item For all $\textbf{y} \in \mathbb{R}^n$,
            \begin{equation*}
                \left( \tau_0\varphi \right) \left( \textbf{y} \right) = \frac{1}{\sqrt{2\pi}} \int_{\mathbb{R}} \, \left( \mathfrak{F}_{\mathbb{R}} \varphi \right) \left( p, \textbf{y} \right) d\lambda^{(1)}\left(p\right) \equiv  \left[\hat{\tau}_0 \left( \mathfrak{F}_{\mathbb{R}} \otimes \mathds{1}_{\mathbb{R}^n} \right) \varphi \right] \left( \textbf{y} \right). 
            \end{equation*}
    \end{enumerate}
\end{proof}

\begin{remark}
    \justifying
    The Fourier trace operator is bounded. $\hfill \square$
\end{remark}


\newpage
\section*{Appendix 3 - Quadratic Form Investigations}
\justifying

\begin{theorem}
    \justifying
    Given $n \in \mathbb{N}: \, n \geq 2$, let $m_1, \ldots, m_n \in \mathbb{R}^+$ be and correspondingly $a_j = (2m_j)^{-1}, \, j = 1, \ldots, n$. Consider, further, $g \in \mathbb{R} \setminus \{0\}$; the map $\left( t, \mathcal{D}_t \right)$, such that $\mathcal{D}_t = H^1\left( \mathbb{R}^n \right)$ and $t: \, (\varphi,\psi) \in \mathcal{D}_t \times \mathcal{D}_t \longmapsto t(\varphi,\psi) \in \mathbb{C}$ with
    \small{\begin{equation*}
        t\left( \varphi, \psi \right) = \sum_{i=1}^n \, a_j \int_{\mathbb{R}^n} \, \left[ \overline{ \partial_j \varphi } \; \partial_j \psi \right] d\lambda^{(n)} - g \sum_{1\leq i < j \leq n} \int_{\mathbb{R}^{n-1}} \, \left[ \overline{ \tau_{(ij)} \varphi} \; \tau_{(ij)} \psi \right] d\lambda^{(n-1)}
    \end{equation*}}
    \normalsize 
    \noindent results in a sesquilinear, densely defined, hermitian, lower semi-bounded, closed form on $L^2\left( \mathbb{R}^n \right)$.
\end{theorem}

\begin{proof}
    \justifying
    Clearly $H^{1}\left( \mathbb{R}^n \right)$ is dense in $\left(L^2\left( \mathbb{R}^n\right), \, \norm{\cdot}_{L^2\left( \mathbb{R}^n\right)}\right)$. Now, given $\varphi, \psi \in H^{1}\left( \mathbb{R}^n \right)$,
    \begin{align*}
        t(\varphi,\psi) & = \sum_{i=1}^n \, a_j \int_{\mathbb{R}^n} \, \overline{\left( \partial_j \varphi \right)} \left( \partial_j \psi \right) - g \sum_{i<j} \, \int_{\mathbb{R}^{n-1}} \, \overline{\left( \tau_{(ij)} \varphi \right)} \left( \tau_{(ij)} \psi \right) = \\
        & = \overline{\sum_{i=1}^n \, a_j \int_{\mathbb{R}^n} \, \overline{\left( \partial_j \psi \right)} \left( \partial_j \varphi \right) - g \sum_{i<j} \, \int_{\mathbb{R}^{n-1}} \, \overline{\left( \tau_{(ij)} \psi \right)} \left( \tau_{(ij)} \varphi \right)} = \overline{t(\psi,\varphi)},
    \end{align*}
    \noindent i.e. $t$ is hermitian. To prove it is lower semi-bounded, let $\psi \in H^{1}\left( \mathbb{R}^n \right)$ be.
    \begin{align*}
        q_t(\psi) & = \sum_{j=1}^n \, a_j \int_{\mathbb{R}^n} \, \abs{\partial_j \psi}^2 - g \sum_{i<j} \, \int_{\mathbb{R}^{n-1}} \, \abs{\tau_{(ij)}\psi}^2 \geq \left( a \equiv \min \left\{ a_1, \ldots, a_n \right\} \right) \\ & \geq a \int_{\mathbb{R}^n} \, \norm{\nabla \psi}^2 - g \sum_{i<j} \, \norm{\tau_{(ij)} \psi}_{L^2\left( \mathbb{R}^{n-1}\right)}^2.
    \end{align*}
    If $g < 0$, then
    \begin{equation*}
        q_t(\psi) \geq a \int_{\mathbb{R}^n} \, \norm{\nabla \psi}^2 + \abs{g} \sum_{i<j} \, \norm{\tau_{(ij)} \psi}_{L^2\left( \mathbb{R}^{n-1}\right)}^2 \geq 0 = 0 \,\norm{\psi}_{H^1(\mathbb{R}^n)}^2.
    \end{equation*}
    On the other hand, if $g > 0$, given $\mu >0$, there exists $C_{\mu} > 0$ such that $\norm{\tau_{(ij)} \psi} \leq \mu \norm{ \, \norm{\nabla \psi} } + C_{\mu} \norm{\psi}$ for all $i,j \in \{1,\ldots,n\}: \, i<j$. Consequently
    \begin{equation*}
        g \sum_{i<j} \, \norm{\tau_{(ij)} \psi}^2 \leq g n(n-1)\left[ \mu^2 \norm{ \, \norm{\nabla \psi}}^2 + C_{\mu}^2 \norm{\psi}_{H^1(\mathbb{R}^n)}^2 \right]
    \end{equation*}
    and
    \begin{equation*}
        q_t(\psi) \geq \left[ a - g n(n-1)\mu^2 \right] \norm{ \, \norm{\nabla \psi}}^2 - g n(n-1)C_{\mu}^2 \norm{\psi}_{H^1(\mathbb{R}^n)}^2.
    \end{equation*}
    By choosing $\mu = \sqrt{a \left[ g n(n-1) \right]^{-1}}$,
    \begin{equation*}
        q_t(\psi) \geq \left[ - g n(n-1)C_{\mu}^2 \right] \norm{\psi}_{H^1}^2 \equiv m(g,n,a) \norm{\psi}_{H^1}^2,
    \end{equation*}
    allowing to state that, for all $g \in \mathbb{R} \setminus \{0\}$, $t$ is lower semi-bounded. Eventually, given $\psi \in H^{1}\left( \mathbb{R}^n\right)$, what follows holds.
    \begin{itemize}
        \item $\boxed{g < 0}$
        \begin{align*}
            q_t(\psi) & = \sum_{i=1}^n \, a_j \norm{\partial_j \psi}_{L^2\left( \mathbb{R}^n \right)}^2 + \abs{g} \sum_{i<j} \, \norm{\tau_{(ij)}\psi}_{L^2\left( \mathbb{R}^{n-1} \right)}^2 \leq \left( A \equiv \underset{j}{\max} \, a_j, \quad K = \underset{(ij)}{\max} \norm{\tau_{(ij)}}^2 \right)  \\ & \leq A \sum_{j=1}^n \, \norm{\partial_j \psi}_{L^2 \left( \mathbb{R}^n \right)}^2 + \frac{\abs{g}n(n-1) K}{2} \norm{\psi}_{H^1\left( \mathbb{R}^n \right)}^2 \leq \left( B \equiv \max\left\{A, \, \frac{\abs{g}n(n-1) K}{2} \right\} \right) \\
            & \leq \left(2 B\right) \, \norm{\psi}_{H^1\left( \mathbb{R}^n \right)}^2.
        \end{align*}
        However, by recalling that $a = \underset{j}{\min}  \, a_j$,
        \begin{align*}
            \norm{\psi}_{H^1 \left(\mathbb{R}^n \right)}^2 & = \norm{\psi}_{L^2 \left(\mathbb{R}^n \right)}^2 + \sum_{i=1}^n \, \norm{\partial_i \psi}_{L^2 \left(\mathbb{R}^n \right)}^2 \leq \norm{\psi}_{L^2\left(\mathbb{R}^n \right)}^2 + \frac{1}{a} \left[ \sum_{i=1}^n a_i \norm{\partial_i \psi}_{L^2\left(\mathbb{R}^n \right)}^2 \right] \leq \\
            & \leq \norm{\psi}_{L^2\left(\mathbb{R}^n \right)}^2 + \frac{1}{a} \left[ \sum_{i=1}^n \, a_i \norm{\partial_i \psi}_{L^2\left(\mathbb{R}^n \right)}^2 - g \sum_{i<j} \, \norm{\tau_{(ij)} \psi}_{L^2\left( \mathbb{R}^{n-1} \right)}^2 \right] \leq \\
            & \leq \norm{\psi}_{L^2 \left(\mathbb{R}^n \right)}^2 + \frac{1}{a} \,  q_t\left( \psi \right).
        \end{align*}
        \item $\boxed{g >0}$
        \begin{align*}
            q_t(\psi) & = \sum_{j=1}^n \, a_j \norm{\partial_j \psi}_{L^2\left( \mathbb{R}^n \right)}^2 - g \sum_{i<j} \, \norm{\tau_{(ij)} \psi}_{L^2\left( \mathbb{R}^{n-1} \right)}^2 \leq \\
            & \leq \sum_{j=1}^n \, a_j \norm{\partial_j \psi}_{L^2\left( \mathbb{R}^n \right)}^2 \leq A \, \norm{\psi}_{H^1\left( \mathbb{R}^n \right)}^2,
        \end{align*}
        Then,
        \begin{align*}
            \norm{\psi}_{H^1\left( \mathbb{R}^n \right)}^2 & = \norm{\psi}_{L^2\left( \mathbb{R}^n \right)}^2 + \sum_{j=1}^n \, \norm{\partial_j \psi}_{L^2\left( \mathbb{R}^n \right)}^2 \leq \norm{\psi}_{L^2\left( \mathbb{R}^n \right)}^2 + \frac{1}{a} \, \sum_{j=1}^n \, a_j \norm{\partial_j \psi}_{L^2\left( \mathbb{R}^n \right)}^2 \equiv \\
            & \equiv \norm{\psi}_{L^2\left( \mathbb{R}^n \right)}^2 + \frac{1}{a} \, \left[q_t\left(\psi\right) + g \sum_{i<j} \, \norm{\tau_{(ij)} \psi}_{L^2\left( \mathbb{R}^{n-1} \right)}^2\right] \leq \\
            & \leq \norm{\psi}_{L^2\left( \mathbb{R}^n \right)}^2 + \frac{1}{a} \, q_{t}\left(\psi\right) + g \frac{n(n-1)K}{2a} \norm{\psi}_{H^1\left( \mathbb{R}^n \right)}^2
        \end{align*}
        leading to $\left[1 - g \frac{Kn(n-1)}{2a}\right] \norm{\psi}_{H^1\left( \mathbb{R}^n \right)}^2 \leq \norm{\psi}_{L^2\left( \left( \mathbb{R}^n \right) \right)}^2 + \frac{1}{a} \, q_t(\psi)$.
    \end{itemize}
\end{proof}

\begin{remark}\label{Remark Hamiltonian}
    The foregoing proposition guarantees the existence of a unique self-adjoint operator $\left( H, \mathcal{D}_H \right)$ on $L^2\left( \mathbb{R}^n\right)$, to be understood as the Hamiltonian of the system considered in section 4, whose corresponding sesquilinear form is $\left( t, \mathcal{D}_t \right)$ indeed. $\hfill \square$
\end{remark}

\newpage
\section*{Appendix 4 - Boundedness Results}

\begin{theorem}\label{App. 4 Prima Prop.}
    Let $F: \, \psi \in L^2\left( \mathbb{R}^2, d\overline{R}_{\nu} d\overline{x}_2 \right) \longmapsto F\psi \in L^2\left( \mathbb{R}^2,dR_{\sigma}d\overline{x}_{\nu_2} \right)$ be the linear operator defined via the position
    \begin{align*}
        & \left[F\psi\right]\left( R_{\sigma},\overline{x}_{\nu_2} \right) \equiv \int_{\mathbb{R}^2} \, d\overline{R}_{\nu}^{\prime} d\overline{x}_{2}^{\prime} \left[ G_z^{(3)} \left( \sqrt{2m_1}\left( R_{\sigma} - \overline{R}_{\nu}^{\prime}, \right), \, \sqrt{2m_2} \left( R_{\sigma} - \overline{x}_2^{\prime} \right), \, \sqrt{2m_{\nu_2}} \left(\overline{x}_{\nu_2} - \overline{R}_{\nu}^{\prime} \right) \right) \right. \\
        \begin{split}
            &\phantom{ = { \left[F\psi\right]\left( R_{\sigma},\overline{x}_{\nu_2} \right) \equiv \int_{\mathbb{R}^2} \, d\overline{R}_{\nu}^{\prime} d\overline{x}_{2}^{\prime} \left[ \right. } } \left. \psi \left( \overline{R}_{\nu}^{\prime}, \overline{x}_2^{\prime} \right) \right]
        \end{split}
    \end{align*}
    \noindent with $z<0$, $m_1,m_2,m_{\nu_2} \in \mathbb{R}^+$. Then, $F$ is bounded.
\end{theorem}

\begin{proof}
    The Schur test will be employed.
    \begin{align*}
        & G_z^{(3)} \left( \sqrt{2m_1}\left( R_{\sigma} - \overline{R}_{\nu}^{\prime}, \right), \, \sqrt{2m_2} \left( R_{\sigma} - \overline{x}_2^{\prime} \right), \, \sqrt{2m_{\nu_2}} \left(\overline{x}_{\nu_2} - \overline{R}_{\nu}^{\prime} \right) \right) = \\
        & = \int_0^{\infty} \, e^{ - \frac{ 2m_1 \left( R_{\sigma} - \overline{R}^{\prime}_{\nu} \right)^2 +  2m_2 \left( R_{\sigma} - \overline{x}_2^{\prime} \right)^2 + 2m_{\nu_2} \left( \overline{x}_{\nu_2} - \overline{R}_{\nu}^{\prime} \right)^2 }{4t} + zt} \frac{dt}{\left(4\pi t\right)^{\frac{3}{2}}} \leq \left( m \equiv \min \left(m_1,m_2,m_{\nu_2}\right) \right) \leq \\ & \leq \int_0^{\infty} \, e^{ - \frac{ 2m \left( R_{\sigma} - \overline{R}^{\prime}_{\nu} \right)^2 +  2m \left( R_{\sigma} - \overline{x}_2^{\prime} \right)^2 + 2m \left( \overline{x}_{\nu_2} - \overline{R}_{\nu}^{\prime} \right)^2 }{4t} + zt} \frac{dt}{\left(4\pi t\right)^{\frac{3}{2}}} \equiv \\
        & \equiv \int_0^{\infty} \, e^{- \frac{\left(x-x^{\prime}\right)^2 + \left(x-y^{\prime}\right)^2 + \left(y-x^{\prime}\right)^2}{4t} + zt} \frac{dt}{\left(4\pi t\right)^{\frac{3}{2}}} \equiv K \left( x,y; x^{\prime}, y^{\prime} \right),
    \end{align*}

    \noindent by having set
    \begin{equation*}
        \begin{cases}
            x & = \sqrt{2m} \, R_{\sigma} \\
            x^{\prime} & = \sqrt{2m} \, \overline{R}_{\nu}^{\prime} \\
            y^{\prime} & = \sqrt{2m} \, \overline{x}_2^{\prime} \\
            y & = \sqrt{2m} \, \overline{x}_{\nu_2}
        \end{cases}.
    \end{equation*}

    \noindent Trivially, $ K\left( x,y; \, x^{\prime}, y^{\prime} \right) = K\left( x^{\prime}, y^{\prime}; \, x,y  \right)$. On the other hand, set $\alpha \equiv \sqrt{\abs{z}}$,
    \begin{equation*}
        \int_{\mathbb{R}^2} \, K\left( x,y; \, x^{\prime}, y^{\prime} \right) dx^{\prime} dy^{\prime} = \int_{\mathbb{R}^2} \, \frac{e^{- \alpha \sqrt{\left(x-x^{\prime}\right)^2 + \left( y-x^{\prime} \right)^2 + \left( x-y^{\prime}\right)^2} } }{4\pi \sqrt{\left(x-x^{\prime}\right)^2 + \left( y-x^{\prime} \right)^2 + \left( x-y^{\prime}\right)^2}} dx^{\prime}dy^{\prime},
    \end{equation*}
    \noindent i.e. it does not exist whenever $x=y=x^{\prime}=y^{\prime}$. Since $\left\{\left(x,y\right) \in \mathbb{R}^2 \vert x = y \right\}$ is a set of $\lambda^{(2)}-$measure zero, $x \neq y$ will be assumed. The coordinate transformation
    \begin{equation*}
        \begin{cases}
            \overline{x}^{\prime} = x^{\prime} - \frac{x+y}{2}\\
            \overline{y}^{\prime} = \frac{y^{\prime} - x}{\sqrt{2}}
        \end{cases} \iff
        \begin{cases}
            x^{\prime} = \overline{x}^{\prime} + \frac{x+y}{2}\\
            y^{\prime} = \sqrt{2} \overline{y}^{\prime} + x
        \end{cases},
    \end{equation*}
    that gives $dx^{\prime}dy^{\prime} = \sqrt{2}d\overline{x}^{\prime}d\overline{y}^{\prime}$ and $\sqrt{\left( x - x^{\prime} \right)^2 + \left( y - x^{\prime} \right)^2 + \left( x - y^{\prime} \right)^2} = \sqrt{2 \left[ \overline{x}^{\prime \, 2} + \overline{y}^{\prime \, 2} + \frac{\left(y-x\right)^2}{4} \right]}$, allows for
    
    \begin{align*}
        \int_{\mathbb{R}^2} \, \abs{K\left(x,y; x^{\prime},y^{\prime}\right)} dx^{\prime}dy^{\prime} & = \int_{\mathbb{R}^2} \, \frac{e^{ - \sqrt{2} \alpha  \sqrt{ \overline{x}^{\prime \, 2} + \overline{y}^{\prime \, 2} + \frac{\left(y-x\right)^2}{4}} }}{ \sqrt{\overline{x}^{\prime \, 2} + \overline{y}^{\prime \, 2} + \frac{\left(y-x\right)^2}{4}} } \frac{d\overline{x}^{\prime}d\overline{y}^{\prime}}{4\pi} = \left( \small{\text{by integrating in polar coordinates}}\right) = \\
        & = \frac{1}{4\pi} \int_0^{\infty} \int_{0}^{2\pi} \, \frac{e^{-\alpha\sqrt{2}\sqrt{\rho^2 + \frac{(y-x)^2}{4}} } }{\sqrt{\rho^2 + \frac{(y-x)^2}{4}}} \rho d\rho d\theta \equiv \frac{1}{2} \int_0^{\infty} \, \frac{e^{-\alpha\sqrt{2}\sqrt{\rho^2 + \frac{(y-x)^2}{4}} } }{\sqrt{\rho^2 + \frac{(y-x)^2}{4}}} \rho d\rho < \\ & < \frac{1}{2} \int_0^{\infty} \, e^{-\alpha\sqrt{2}\rho} d\rho = \frac{1}{2\sqrt{2\abs{z}}}.
    \end{align*}
    \noindent In the end, $\norm{F} \leq \frac{1}{2\sqrt{2\abs{z}}}$.
\end{proof}

\begin{theorem}
    Let $B: \varphi \in L^2\left( \mathbb{R}^3, \, d\overline{R}_{\nu}d\overline{x}_1 d\overline{x}_2 \right) \longmapsto B\varphi \in L^2\left( \mathbb{R}^3, dR_{\sigma} d\overline{x}_{\nu_1} d\overline{x}_{\nu_2} \right)$ be the linear operator defined by
    
    \small{\begin{align*}
        & \left[B\varphi\right] \left( R_{\sigma}, \overline{x}_{\nu_1}, \overline{x}_{\nu_2} \right) = \\
        \begin{split}
            & = \int_{\mathbb{R}^3} \, d\overline{R}_{\nu}^{\prime} d\overline{x}_1^{\prime} d\overline{x}_2^{\prime} \; \left[ G_z^{(4)} \left( \sqrt{2m_1} \left( R_{\sigma} - \overline{x}_{1}^{\prime} \right), \, \sqrt{2m_2} \left( R_{\sigma} - \overline{x}_2^{\prime} \right), \, \sqrt{2m_{\nu_1}} \left( \overline{x}_{\nu_1} - \overline{R}_{\nu}^{\prime} \right), \, \sqrt{2m_{\nu_2}} \left( \overline{x}_{\nu_2} - \overline{R}_{\nu}^{\prime} \right) \right) \right. \\
            &\phantom{ = {\int_{\mathbb{R}^3} \, d\overline{R}_{\nu}^{\prime} d\overline{x}_1^{\prime} d\overline{x}_2^{\prime} \; \left[\right. } } \left.     \varphi\left( \overline{R}_{\nu}^{\prime}, \, \overline{x}_1^{\prime}, \, \overline{x}_{2}^{\prime} \right) \right]
        \end{split}
    \end{align*}}
    \normalsize for all $z<0$, $m_1, m_2, m_{\nu_1}, m_{\nu_2} \in \mathbb{R}^+$. $B$ is a bounded operator.
\end{theorem}

\begin{proof}
    By proceeding as in proposition \ref{App. 4 Prima Prop.}, it does not harm generality focusing on
    \begin{equation*}
        K\left(x,y,w; \, x^{\prime},y^{\prime},w^{\prime} \right) = G_z^{(4)} \left( x-x^{\prime}, y-x^{\prime},w-y^{\prime}, w-w^{\prime} \right).
    \end{equation*}
    \noindent Then
    \begin{align*}
        & \int_{\mathbb{R}^3} \, dx^{\prime} dy^{\prime} dw^{\prime} \, \abs{G_z^{(4)} \left( x-x^{\prime}, y-x^{\prime},w-y^{\prime}, w-w^{\prime} \right)} = \\
        & = \int_{\mathbb{R}^3} \, dx^{\prime} dy^{\prime} dw^{\prime} \, \int_0^{\infty} \, \frac{dt}{\left(4\pi t\right)^2} \, \exp{ \left\{ - \frac{\left(x-x^{\prime}\right)^2 + \left(y-x^{\prime}\right)^2 + \left(w-y^{\prime}\right)^2 + \left(w-w^{\prime}\right)^2}{4t} + zt \right\} }.
    \end{align*}
    By considering the coordinate transformation
    \begin{equation*}
        \begin{cases}
            \overline{x}^{\prime} &= x^{\prime} - \frac{x+y}{2}\\
            \overline{y}^{\prime} &= \frac{y^{\prime} - w}{\sqrt{2}}\\
            \overline{w}^{\prime} &= \frac{w^{\prime} - w}{\sqrt{2}}
        \end{cases} \iff
        \begin{cases}
            x^{\prime} &= \overline{x}^{\prime} + \frac{x+y}{2}\\
            y^{\prime} &= \sqrt{2}\overline{y}^{\prime} + w\\
            w^{\prime} &= \sqrt{2}\overline{w}^{\prime} + w
        \end{cases},
    \end{equation*}
    \noindent $dx^{\prime} dy^{\prime} dz^{\prime} = 2 d\overline{x}^{\prime} d\overline{y}^{\prime} d\overline{z}^{\prime}$ and $\left(x-x^{\prime}\right)^2 + \left(y-x^{\prime}\right)^2 + \left(w-y^{\prime}\right)^2 + \left(w-w^{\prime}\right)^2 = 2 \left[ \overline{x}^{\prime \, 2} + \overline{y}^{\prime \, 2} + \overline{w}^{\prime \, 2} + \frac{\left( x-y \right)^2}{4} \right]$,  what follows holds
    \small{\begin{align*}
        & \int_{\mathbb{R}^{3}} \, dx^{\prime} dy^{\prime} dw^{\prime} \, \int_0^{\infty} \, \frac{dt}{\left(4\pi t\right)^2} \exp{ \left\{ - \frac{\left(x-x^{\prime}\right)^2 + \left(y-x^{\prime}\right)^2 + \left(w-y^{\prime}\right)^2 + \left(w-w^{\prime}\right)^2}{4t} + zt \right\} } = \\
        & = \int_{\mathbb{R}^3} \, d\overline{x}^{\prime}d\overline{y}^{\prime}d\overline{w}^{\prime} \, \int_0^{\infty} \, \frac{dt}{8\pi^2 t^2} \, \exp{\left\{ - \frac{\overline{x}^{\prime \, 2} + \overline{y}^{\prime \, 2} + \overline{w}^{\prime \, 2} + \frac{\left(y-x\right)^2}{4}}{2t} + zt \right\}} = \\
        & = \int_0^{\infty} \rho^2 d\rho \int_0^{\pi} \, \sin\theta d\theta \int_0^{2\pi}d\phi \int_0^{\infty} \, \frac{dt}{8\pi^2t^2} \, \exp{ \left\{ -\frac{\rho^2}{2t} \right\} } \, \exp{ \left\{ -\frac{\left(y-x\right)^2}{8t} \right\} } \exp{ \left\{zt\right\} } \leq \\
        & \leq \int_0^{\infty} \, \frac{dt}{2\pi t^2} \, e^{zt} \left[ \int_0^{\infty} \, d\rho \, \rho^2 \exp{\left\{ - \frac{\rho^2}{2t} \right\}}  \right] = \frac{\sqrt{\pi}}{4} \, \int_0^{\infty} \, \frac{dt}{2\pi} \, \frac{e^{zt}}{t^2} \, \left(2t\right)^{\frac{3}{2}} = \frac{1}{2\sqrt{2\abs{z}}} < \infty
    \end{align*}}
    \normalsize The Schur test then gives $\norm{B} \leq \left( 2\sqrt{2\abs{z}} \right)^{-1}$.
\end{proof}


\newpage
\section*{Appendix 5 - $\left(H,\mathcal{D}_{H}\right)$ is the unique self-adjoint extension of $\left( H_0, \ker \tau \right)$}
\justifying

\begin{lemma}
    \justifying
    Let $\mathcal{H}$ be a Hilbert space and let $\left( A, \mathcal{D}_{A} \right)$ be a closed operator. For all $z \in \rho\left(A\right)$, $R_{A}\left(z\right): \, \mathcal{H} \longrightarrow \left( \mathcal{D}_{A}, \norm{\cdot}_{A} \right)$ is a bounded operator.
\end{lemma}

\begin{proof}
    \justifying
    Let $\psi \in \mathcal{H}$ be arbitrary. 
    \begin{align*}
        & \norm{ R_{A}\left(z\right) \psi }_{A}^2 = \langle R_{A}\left(z\right) \psi, R_{A}\left(z\right) \psi \rangle_{A} = \langle R_{A}\left(z\right) \psi, R_{A}\left(z\right) \psi \rangle + \langle A R_{A}\left(z\right) \psi, A R_{A}\left(z\right) \psi \rangle = \\
        & = \langle R_{A}\left(z\right) \psi, R_{A}\left(z\right) \psi \rangle + \langle \left( A - z + z \right) R_{A}\left(z\right) \psi, \left( A - z + z \right) R_{A}\left(z\right) \psi \rangle = \\
        & = \left( 1 + \abs{z}^2 \right) \norm{ R_{A}\left( z \right) \psi }^2 + \norm{\psi}^2 + 2\Re \left( z \langle \psi, R_{A}\left(z\right) \psi \rangle \right) \leq \\
        & \leq \left( 1 + \abs{z}^2 \right) \norm{R_{A}\left(z\right)\psi}^2 + \norm{\psi}^2 + \abs{z} \left[ \norm{\psi}^2 + \norm{R_A\left(z\right)\psi}^2 \right] \leq C_{A}\left(z\right) \norm{\psi}^2.
    \end{align*}
\end{proof}

\begin{lemma}
    \justifying
    Given $n \in \mathbb{N}: n \geq 2$, let $\left( H_0, \mathcal{D}_{H_0} \right)$ be the free Hamiltonian on $L^2\left( \mathbb{R}^n \right)$. $\tau_0$ is $H_0-$bounded.
\end{lemma}

\begin{proof}
    \justifying
    Showing the statement amounts in proving that $\tau_0: \left( \mathcal{D}_{H_0}, \norm{\cdot}_{H_0} \right) \longrightarrow \left( L^2\left( \mathbb{R}^{n-1} \right), \norm{\cdot}_{L^2} \right)$ is continuous. First of all, it is observed that $\mathcal{D}_{H_0} = H^2\left( \mathbb{R}^n \right) \subset H^1\left( \mathbb{R}^n \right) \equiv \mathcal{D}_{\tau_0}$. Then, for all $\psi \in \mathcal{D}_{H_0}$,
    \begin{equation*}
        \norm{\tau_0 \psi}_{L^2}^2 \leq \norm{\psi}_{H^1}^2 = \norm{ \psi }_{L^2}^2 + \norm{ \, \norm{\Vec{\nabla}\psi} }_{L^2}^2.
    \end{equation*}
    In particular,
    \begin{equation*}
        \norm{ \, \norm{\Vec{\nabla}\psi} }_{L^2}^2 = \int_{\mathbb{R}^n} \, \overline{\Vec{\nabla}\psi} \cdot \Vec{\nabla} \psi = \left( H^1\left( \mathbb{R}^n \right) = H^1_0\left( \mathbb{R}^n \right) \right) = - \int_{\mathbb{R}^n} \, \overline{\psi} \Delta \psi \equiv C \int_{\mathbb{R}^n} \overline{\psi} H_0 \psi = C \langle \psi, H_0 \psi \rangle.
    \end{equation*}
    \noindent Since $\left(H_0, \mathcal{D}_{H_0}\right)$ is a positive operator,
    \begin{equation*}
        \langle \psi, H_0 \psi \rangle \leq \norm{\psi} \norm{H_0 \psi} \leq \norm{\psi}^2 + \norm{H_0\psi}^2,
    \end{equation*}
    therefore
    \begin{equation*}
        \norm{\tau_0 \psi}_{L^2}^2 \leq K \norm{ \psi }_{H_0}^2.
    \end{equation*}
\end{proof}

\begin{corollary}
    \justifying
    For all $\sigma \in \mathcal{I}$, $z \in \rho\left(H_0\right)$, the linear map
    \begin{equation*}
        G_{\sigma}\left(z\right) \doteq \tau_{\sigma}R_{H_0}\left( z \right): \,  L^2\left( \mathbb{R}^n,dx_1 \cdots dx_n \right) \longrightarrow \chi_{\sigma}^{\left( \text{red} \right)}
    \end{equation*}
    is bounded. Analogously, for all $\nu \in \mathcal{I}$,
    \begin{equation*}
        \tau_{\nu}G_{\sigma}^{\ast}\left(z\right): \, \chi_{\sigma}^{\left( \text{red} \right)} \longrightarrow \chi_{\nu}^{\left( \text{red} \right)}
    \end{equation*}
    is bounded. $\hfill \blacksquare$
\end{corollary}

\begin{lemma}
    \justifying
    For all $\sigma \in \mathcal{I}$, $z<0$, $S^{\sigma}\left( z \right) = v\left( \cdot \right) G_{\sigma}\left( z \right) \in \mathcal{B}\left( L^2\left(\mathbb{R}^n\right), \chi_{\sigma} \right)$.
\end{lemma}

\begin{proof}
    \justifying
    Given arbitrarily $\psi \in L^2\left( \mathbb{R},dr_{\sigma} \right) \otimes \tilde{\chi}_{\sigma}^{(\text{red})}$,
    \begin{align*}
        & \left[ T_{0, \underline{P}_{\sigma}}^{\sigma}\left(z\right) \psi \right] \left( r_{\sigma}, \underline{P}_{\sigma} \right) = v\left(r_{\sigma}\right) \left(2\mu_{\sigma}\right) \int_{\mathbb{R}} \, G^{(1)}_{(2\mu_{\sigma})\left(z-Q_{\sigma}\right)} \left( - r_{\sigma}^{\prime}\right) \psi\left( r_{\sigma}^{\prime}, \underline{P}_{\sigma} \right) dr_{\sigma}^{\prime} = \\
        & = v\left( r _{\sigma} \right) \left[ \left( 2\mu_{\sigma} \right) \int_{\mathbb{R}} \, G^{(1)}_{(2\mu_{\sigma})\left(z-Q_{\sigma}\right)} \left( r_{\sigma} - r_{\sigma}^{\prime}\right) \psi\left( r_{\sigma}^{\prime}, \underline{P}_{\sigma} \right) dr_{\sigma}^{\prime} \right]_{r_{\sigma} = 0} = v\left(r_{\sigma}\right) \left\{ \tau_0 \left[ R_{\tilde{H}_{0}^{\sigma}}\left(z\right) \psi \right] \right\}\left(\underline{P}_{\sigma}\right)
    \end{align*}
    where
    \begin{equation*}
        \tilde{H}_{0}^{\sigma} = -\frac{1}{2\mu_{\sigma}} \frac{\partial^2}{\partial r_{\sigma}^2} + Q_{\sigma}\mathds{1}.
    \end{equation*}
    Then
    \begin{align*}
        & \left[ T_{0, \underline{P}_{\sigma}}^{\sigma}\left(z\right) \psi \right] \left( r_{\sigma}, \underline{P}_{\sigma} \right) = v\left(r_{\sigma}\right) \left\{ \tau_0 \left[ \left( \mathds{1} \otimes \mathfrak{F}_{\underline{Y}_\sigma} \right) R_{H_{0}^{\sigma}}\left(z\right) \left( \mathds{1} \otimes \mathfrak{F}_{\underline{Y}_\sigma}^{-1} \right) \psi \right] \right\}\left(\underline{P}_{\sigma}\right) = \\
        & = v\left(r_{\sigma}\right)  \left[ \mathfrak{F}_{\underline{Y}_\sigma} \tau_0 R_{H_{0}^{\sigma}}\left(z\right) \left( \mathds{1} \otimes \mathfrak{F}_{\underline{Y}_\sigma}^{-1} \right) \psi \right] \left(\underline{P}_{\sigma}\right) \equiv \left[  \left( \mathds{1} \otimes \mathfrak{F}_{\underline{Y}_\sigma} \right) v\left( \cdot \right) \tau_0 R_{H_0^{\sigma}}\left(z\right) \left( \mathds{1} \otimes \mathfrak{F}_{\underline{Y}_\sigma}^{-1} \right) \right] \left( r_{\sigma}, \underline{P}_{\sigma} \right),
    \end{align*}
    hence, by definition, $T_0^{\sigma}\left(z\right) = v\left( \cdot \right) \tau_0 R_{H_0^{\sigma}}\left(z\right)$. In the end
    \begin{equation*}
        S^{\sigma}\left(z\right) = T_0^{\sigma}\left(z\right)U_{\sigma} = v\left(\cdot\right) \tau_0 R_{H_0^{\sigma}}\left(z\right)U_{\sigma} = v\left( \cdot \right) \tau_0 U_{\sigma} R_{H_0}\left(z\right) \equiv v\left(\cdot\right) G_{\sigma}\left( z \right),
    \end{equation*}
    i.e., for all $\psi \in L^2\left( \mathbb{R}^n, dx_1\cdots dx_n \right)$,
    \begin{equation*}
        \left[S^{\sigma}\left(z\right)\psi\right] \left( r_{\sigma}, R_{\sigma}, x_1, \ldots, \hat{x}_{\sigma_1}, \ldots, \hat{x}_{\sigma_2}, \ldots, x_n \right) = v\left(r_{\sigma}\right) \left[G_{\sigma}\left( z \right)\psi\right]\left( R_{\sigma}, x_1, \ldots, \hat{x}_{\sigma_1}, \ldots, \hat{x}_{\sigma_2}, \ldots, x_n \right).
    \end{equation*}
\end{proof}

\begin{lemma}\label{Factorization of Diagional Contribution}
    \justifying
    For all $\sigma \in \mathcal{I}$, $z<0$, $\phi_0^{\sigma}\left(z\right) = \langle v_{\sigma}, \cdot \rangle v_{\sigma} \otimes g D^{(\sigma)}\left( z \right)$, where
    \begin{equation*}
        D^{(\sigma)}\left( z \right) = \left( \sqrt{\frac{\mu_\sigma}{2}} \right) \mathfrak{F}_{\underline{Y}_\sigma}^{-1} \frac{\mathds{1}}{\sqrt{Q_{\sigma} - z}} \mathfrak{F}_{\underline{Y}_\sigma} \in \mathcal{B}\left( \chi_{\sigma}^{\left( \text{red} \right)} \right).
    \end{equation*}
\end{lemma}

\begin{proof}
    \justifying
    Arbitrarily given $\sigma \in \mathcal{I}$, $z<0$, for all $\varphi \in L^2\left( \mathbb{R},dr_{\sigma} \right) \otimes \tilde{\chi}_{\sigma}^{(\text{red})}$, 
    \begin{equation*}
        \left[ \phi_{0,\underline{P}_{\sigma}}^{\sigma} \left(z\right) \varphi \right] \left( r_{\sigma}, \underline{P}_{\sigma} \right) = g \left( \sqrt{\frac{\mu_{\sigma}}{2}} \right) v\left( r_{\sigma} \right) \int_{\mathbb{R}} \, \frac{v\left( r_{\sigma}^{\prime} \right) \varphi\left( r_{\sigma}^{\prime}, \underline{P}_{\sigma} \right)}{\sqrt{Q_{\sigma} - z}} dr_{\sigma}^{\prime}.
    \end{equation*}
    In particular, if $\varphi \equiv \alpha \otimes \xi$, $\alpha \in L^2\left( \mathbb{R},dr_{\sigma} \right)$, $\xi \in \Tilde{\chi}^{(\text{red})}_{\sigma}$, then
    \begin{align*}
        \left[ \phi_{0,\underline{P}_{\sigma}}^{\sigma} \left(z\right) \varphi \right] \left( r_{\sigma}, \underline{P}_{\sigma} \right) & = \left[v\left(r_{\sigma}\right) \int_{\mathbb{R}} v\left( r_{\sigma}^{\prime} \right) \alpha\left( r_{\sigma}^{\prime} \right) dr_{\sigma}^{\prime}\right] g \left( \sqrt{\frac{\mu_{\sigma}}{2}} \right) \frac{\xi\left( \underline{P}_{\sigma} \right)}{\sqrt{Q_{\sigma} - z}} \equiv \\
        & = \left[\langle v_{\sigma}, \cdot \rangle v_{\sigma} \otimes g \left( \sqrt{\frac{\mu_{\sigma}}{2}} \right) \frac{\mathds{1}}{\sqrt{Q_{\sigma} - z}}\right] \alpha \otimes \xi,
    \end{align*}
    therefore, by definition, $\phi_0^{\sigma}\left(z\right) = \langle v, \cdot \rangle v \otimes g D^{(\sigma)}\left( z \right)$. As a consequence,
    \begin{equation*}
        \left[ \Lambda_{0}\left( z \right)_{\text{diag}} \right]_{\sigma \sigma} = \mathds{1}_{\chi_{\sigma}} - \phi_0^{\sigma}\left(z\right) \equiv \mathds{1}_{\chi_{\sigma}} - \langle v_{\sigma},\cdot \rangle v_{\sigma} \otimes g D^{(\sigma)}\left(z\right)
    \end{equation*}
    and, by direct inspection,
    \begin{equation*}
        \left[ \Lambda_{0}\left( z \right)_{\text{diag}}^{-1} \right]_{\sigma \sigma} = \mathds{1}_{\chi_{\sigma}} + \langle v_{\sigma},\cdot \rangle v_{\sigma} \otimes \left\{ g D^{(\sigma)}\left(z\right) \left[ \mathds{1}_{\chi^{(\text{red})}_{\sigma}} - g D^{(\sigma)}\left(z\right) \right]^{-1} \right\}.
    \end{equation*}
    where the operator inversion in curly brackets is understood in $\mathcal{B}\left( \chi_{\sigma}^{(\text{red})} \right)$.
\end{proof}

\begin{lemma}\label{Factorization of Off Diagonal Elements}
    \justifying
    For all $\sigma, \nu \in \mathcal{I}$, $z<0$, $\left[\Lambda_{0}\left(z\right)_{\text{off}}\right]_{\sigma \nu} = \langle v_{\nu},\cdot \rangle v_{\sigma} \otimes \left[ - g \tau_{\sigma} G_{\nu}^{\ast}\left( \overline{z} \right) \right]$. 
\end{lemma}

\begin{proof}
    \justifying
    To simplify the notation, it will not harm generality assuming $\sigma = \left(12\right)$. The case $\nu = (1\nu_2), \; 3 \leq \nu_2 \leq n$ is considered first. It is then recalled that
    \begin{equation*}
        \left[ \Lambda_{0,\underline{P}_{\nu}} \left(z\right)_{\text{off}} \right]_{\sigma \nu}: \, L^2\left( \mathbb{R}^3, dr_{\nu}dR_{\nu}dx_2 \right) \otimes \tilde{\chi}_{\nu}^{-} \longrightarrow L^2\left( \mathbb{R}^3, dr_{\sigma}dR_{\sigma}dx_{\nu_2} \right) \otimes \tilde{\chi}_{\nu}^{-}
    \end{equation*}
    where, for all $\psi \in L^2\left( \mathbb{R}^3, dr_{\nu}dR_{\nu}dx_2 \right) \otimes \tilde{\chi}_{\nu}^{-}$,
    \begin{align*}
        & \left\{ \left[ \Lambda_{0,\underline{P}_{\nu}} \left(z\right)_{\text{off}} \right]_{\sigma \nu} \psi \right\}\left( r_{\sigma}, R_{\sigma}, x_{\nu_2}, \underline{P}_{\nu} \right) = \\
        & = - g \sqrt{2m_1 2m_2 2m_{\nu_2}} \, v\left(r_{\sigma}\right) \int_{\mathbb{R}^3} \, dr_{\nu}^{\prime}dR_{\nu}^{\prime}dx_2^{\prime} \; G_{\left( z - Q_{\nu} \right)}^{(3)}\left(X_{\sigma, \nu, \, 0}\right) v\left( r_{\nu}^{\prime} \right) \psi\left( r_{\nu}^{\prime}, R_{\nu}^{\prime}, x_2^{\prime}, \underline{P}_{\nu} \right).
    \end{align*}
    Should $\psi \equiv \alpha \otimes \delta$, where $\alpha \in L^2\left( \mathbb{R}, dr_{\sigma} \right)$, $\delta \in L^2\left( \mathbb{R}^2,dR_{\sigma} dx_{\nu_2} \right) \otimes \tilde{\chi}_{\nu}^{-}$, it would be
    \begin{align*}
        & \left\{ \left[ \Lambda_{0,\underline{P}_{\nu}} \left(z\right)_{\text{off}} \right]_{\sigma \nu} \alpha \otimes \delta \right\}\left( r_{\sigma}, R_{\sigma}, x_{\nu_2}, \underline{P}_{\nu} \right) = \\
        & = \langle v_{\nu}, \alpha \rangle v_{\sigma}\left( r_{\sigma} \right) \left\{ - g \int_{\mathbb{R}^2} \, dR_{\nu}^{\prime} dx_2^{\prime} \left[\sqrt{2m_1 2m_2 2m_{\nu_2}} G_{\left( z - Q_{\nu} \right)}\left( X_{\sigma \nu, 0} \right) \delta\left( R_{\nu}^{\prime}, x_2^{\prime}, \underline{P}_{\nu} \right) \right] \right\}.
    \end{align*}
    Concerning the integral in curly brackets, set
    \begin{equation*}
        \tilde{H}_0^{\nu} = - \frac{1}{2\mu_{\nu}} \frac{\partial^2}{\partial r_{\nu}^{2}} - \frac{1}{2M_{\nu}} \frac{\partial^2}{\partial R_{\nu}^2} - \frac{1}{2m_2} \frac{\partial^2}{\partial x_2^2} + Q_{\nu}\mathds{1},
    \end{equation*}
    for all $\varphi \in L^2\left( \mathbb{R}^{n-1}, dR_{\nu}dx_2 d\underline{P}_{\nu} \right)$, $\psi \in L^2\left( \mathbb{R}^n, dr_{\nu} dR_{\nu} dx_{2} d\underline{P}_{\nu} \right)$
    \small
    \begin{align*}
        & \langle \varphi, \left[\tau_0 R_{\tilde{H}_0^{\nu}} \left( z - Q_{\nu} \right) \right] \psi \rangle = \\
        & = \int_{\mathbb{R}^{n-1}} \, \overline{\varphi}\left( R_{\nu}, x_{2}, \underline{P}_{\nu} \right) \left[ \tau_0 R_{\tilde{H}_0^{\nu}}\left( z - Q_{\nu} \right) \psi \right] \left( R_{\nu}, x_2, \underline{P}_{\nu} \right) dR_{\nu} dx_2 d\underline{P}_{\nu} = \\
        & = \int_{\mathbb{R}^{n-1}} \, \overline{\varphi}\left( R_{\nu}, x_{2}, \underline{P}_{\nu} \right)  \int_{\mathbb{R}^3} \, R_{\tilde{H}_0^{\nu}}\left( z - Q_{\nu} \right) \left(r_{\nu}^{\prime}, R_{\nu} - R_{\nu}^{\prime}, x_2 - x_2^{\prime} \right) \psi \left( r_{\nu}^{\prime}, R_{\nu}^{\prime}, x_2^{\prime}, \underline{P}_{\nu} \right) dr_{\nu}^{\prime} dR_{\nu}^{\prime} dx_2^{\prime} dR_{\nu} dx_2 d\underline{P}_{\nu} = \\
        & = \int_{\mathbb{R}^n} \overline{\left[ \int_{\mathbb{R}^2} \overline{R_{\Tilde{H}_0^{\nu}}\left(z-Q_{\nu}\right)\left( r_{\nu}, R_{\nu} - R_{\nu}^{\prime}, x_2 - x_2^{\prime} \right)} \varphi\left( R_{\nu}^{\prime}, x_2^{\prime}, \underline{P}_{\nu} \right) dR_{\nu}^{\prime} dx_2^{\prime} \right]} \psi\left( r_{\nu}, R_{\nu}, x_2, \underline{P}_{\nu} \right) dr_{\nu} dR_{\nu} dx_2 d\underline{P}_{\nu} = \\
        & = \langle \left[ \tau_0 R_{\Tilde{H}_0^{\nu}} \left( z - Q_{\nu} \right) \right]^{\ast} \varphi, \psi \rangle.
    \end{align*}
    \normalsize This implies
    \small
    \begin{align*}
        \left\{ \left[ \tau_0 R_{\tilde{H}_{0}^{\nu}}\left( z - Q_{\nu} \right) \right]^{\ast} \varphi \right\}\left( r_{\nu}, R_{\nu}, x_2, \underline{P}_{\nu} \right) = \int_{\mathbb{R}^2} \overline{R_{\Tilde{H}_0^{\nu}}\left( z - Q_{\nu}\right)\left( r_{\nu}, R_{\nu} - R_{\nu}^{\prime}, x_2 - x_2^{\prime} \right)} \varphi\left( R_{\nu}^{\prime}, x_2^{\prime}, \underline{P}_{\nu} \right) dR_{\nu}^{\prime} dx_2^{\prime}.
    \end{align*}
    \normalsize The coordinate transformation
    \begin{equation*}
        \vec{g}: \, \left( x_1,x_2,x_{\nu_2}, \underline{P}_{\nu} \right) \in \mathbb{R}^n \longmapsto \vec{g}\left( x_1,x_2,x_{\nu_2}, \underline{P}_{\nu} \right) \in \mathbb{R}^n
    \end{equation*}
    where
    \begin{equation*}
        \vec{g}\left( x_1,x_2,x_{\nu_2}, \underline{P}_{\nu} \right) = 
        \begin{cases}
            R_{\nu} & = \frac{m_1 x_1 + m_{\nu_2} x_{\nu_2}}{m_1 + m_{\nu_2}}\\
            r_{\nu} & = x_{\nu_2} - x_1\\
            x_2 & \equiv x_2\\
            \underline{P}_{\nu} & \equiv \underline{P}_{\nu}
        \end{cases},
    \end{equation*}
    is then considered. Let $\overline{U}_{\nu}: \, L^2\left( \mathbb{R}^n, dx_1dx_2dx_{\nu_2}d\underline{P}_{\nu} \right) \longrightarrow L^2\left( \mathbb{R}^n, dr_{\nu}dR_{\nu}dx_2 d\underline{P}_{\nu} \right)$ be the unitary operator implementing such a coordinate transformation. By defining $\overline{U}_{\sigma}$ analogously,
    \footnotesize
    \begin{align*}
        & \left\{ \tau_0 \overline{U}_{\sigma} \overline{U}_{\nu}^{-1} \left[ \tau_0 R_{\tilde{H}_0^{\nu}}\left( z - Q_{\nu} \right) \right]^{\ast} \varphi \right\}\left( R_{\sigma}, x_{\nu_2}, \underline{P}_{\nu} \right) = \\
        & = \int_{\mathbb{R}^2} \, \overline{\left[ \sqrt{2m_1 2m_2 2m_{\nu_2}} G_{\left( z - Q_{\nu} \right)}^{(3)}\left( \sqrt{2m_1}\left(R_{\sigma} - R_{\nu}^{\prime}\right), \sqrt{2m_2}\left( R_{\sigma} - x_2^{\prime} \right), \sqrt{2m_{\nu_2}} \left( x_{\nu_2} - R_{\nu}^{\prime} \right) \right) \right]} \varphi\left( R_{\nu}^{\prime}, x_2^{\prime}, \underline{P}_{\nu} \right) dR_{\nu}^{\prime} dx_{2}^{\prime} \equiv \\
        & = \left\{ \left[\mathds{1}_{L^2\left(\mathbb{R}^2, dR_{\sigma}dx_{\nu_2}\right)} \otimes \mathfrak{F}_{\underline{Y}_{\nu}} \right] \tau_{\sigma}G_{\nu}^{\ast}\left( \overline{z} \right) \right\} \left( R_{\sigma}, x_{\nu_2}, \underline{P}_{\nu} \right).
    \end{align*}
    \normalsize By recalling the definition of $\left[\Lambda_0\left(z\right)_{\text{off}}\right]_{\sigma \nu}$ and collecting everything up here,
    \begin{equation*}
        \left[ \Lambda_0\left( z \right)_{\text{off}} \right]_{\sigma \nu} = \langle v_{\nu}, \cdot \rangle v_{\sigma} \otimes \left[ - g \tau_{\sigma} G_{\nu}^{\ast}\left( \overline{z} \right) \right].
    \end{equation*}
    The case $\nu = \left( \nu_1 \nu_2 \right)$, where $3 \leq \nu_1 < \nu_2 \leq n$ is analogously proved. The operator of interest is
    \begin{equation*}
        \left[ \Lambda_{0,\underline{P}_{\nu}}\left(z\right)_{\text{off}} \right]_{\sigma \nu}: \, L^2\left( \mathbb{R}^4, dx_1 dx_2 dr_{\nu} dR_{\nu} \right) \otimes \Tilde{\chi}_{\nu}^{-} \longrightarrow L^2\left( \mathbb{R}^4, dr_{\sigma} dR_{\sigma} dx_{\nu_1} dx_{\nu_2} \right) \otimes \Tilde{\chi}_{\nu}^{-},
    \end{equation*}
    such that for all $\psi \in L^2\left( \mathbb{R}^4, dx_1 dx_2 dr_{\nu} dR_{\nu} \right) \otimes \Tilde{\chi}_{\nu}^{-}$
    \begin{align*}
        & \left\{ \left[ \Lambda_{0,\underline{P}_{\nu}}\left(z\right)_{\text{off}} \right] \psi \right\}\left( r_{\sigma}, R_{\sigma}, x_{\nu_1}, x_{\nu_2}, \underline{P}_{\nu} \right) = \\
        & = v\left( r_{\sigma} \right) \left[ - g \sqrt{2m_1 2m_2 2m_{\nu_1} 2m_{\nu_2}} \int_{\mathbb{R}^4} \, dr_{\nu}^{\prime} dR_{\nu}^{\prime} dx_1^{\prime} dx_2^{\prime} G_{\left( z - Q_{\nu}\right)}^{(4)}\left( X_{\sigma\nu, 0} \right) v\left( r_{\nu}^{\prime} \right) \psi \left( r_{\nu}^{\prime}, R_{\nu}^{\prime}, x_1^{\prime}, x_2^{\prime}, \underline{P}_{\nu} \right) \right].
    \end{align*}
    For all $\varphi \in L^2\left( \mathbb{R}^{n-1}, dR_{\nu}dx_1dx_2d\underline{P}_{\nu} \right)$,
    \begin{align*}
        & \left\{ \left[ \tau_0 R_{\tilde{H}_0^{\nu}}\left( z - Q_{\nu} \right) \right]^{\ast} \right\}\left( r_{\nu}, R_{\nu}, x_1, x_2, \underline{P}_{\nu} \right) = \\
        & = \int_{\mathbb{R}^3} \, \overline{R_{\tilde{H}_0^{\nu}}\left( z - Q_{\nu} \right) \left( r_{\nu}, R_{\nu} - R_{\nu}^{\prime}, x_1 - x_1^{\prime}, x_2 - x_2^{\prime} \right)} \varphi\left( R_{\nu}^{\prime}, x_1^{\prime}, x_2^{\prime}, \underline{P}_{\nu} \right) dR_{\nu}^{\prime} dx_1^{\prime} dx_2^{\prime}.
    \end{align*}
    Analogously defined $\overline{U}_{\nu}$ and $\overline{U}_{\sigma}$ for the current case, it is possible to pass from
    \begin{equation*}
        2\mu_{\nu} r_{\nu}^2 + 2M_{\nu} \left( R_{\nu} - R_{\nu}^{\prime} \right)^2 + 2m_1 \left( x_1 - x_1^{\prime} \right)^2 + 2m_2 \left( x_2 - x_2^{\prime} \right)^2
    \end{equation*}
    to
    \begin{equation*}
        2m_1\left( R_{\sigma} - x_1^{\prime} \right)^2 + 2m_2\left( R_{\sigma} - x_2^{\prime} \right)^2 + 2m_{\nu_1}\left( x_{\nu_1} - R_{\nu}^{\prime} \right)^2 + 2m_{\nu_2} \left( x_{\nu_2} - R_{\nu}^{\prime} \right)^2
    \end{equation*}
    by acting with $\tau_0 \overline{U}_{\sigma} \overline{U}_{\nu}$ upon $\left[ \tau_0 R_{\tilde{H}_0^{\nu}}\left( z - Q_{\nu} \right) \right]^{\ast}$. In the end
    \begin{equation*}
        \left[ \Lambda_0\left(z\right)_{\text{off}} \right]_{\sigma \nu} = \langle v_{\nu}, \cdot \rangle v_{\sigma} \otimes \left[ - g \tau_{\sigma} G_{\nu}^{\ast}\left( \overline{z} \right) \right].
    \end{equation*}
\end{proof}

\begin{remark}
    \justifying
    The use of complex conjugation, even though irrelevant at the stage of the foregoing proposition, will be clearer in the following. $\hfill \square$
\end{remark}

\begin{definition}
    \justifying
    For all $z \in \rho\left(H_0\right)$, the linear operator
    \begin{equation*}
        G \left(z\right): \psi \in L^2\left( \mathbb{R}^n, dx_1 \cdots dx_n \right) \longmapsto G\left(z\right)\psi \doteq \left(G_{\sigma}\left(z\right) \psi \right)_{\sigma = 1}^{\frac{n(n-1)}{2}} \in \chi^{\left( \text{red} \right)} = \bigoplus_{\sigma} \chi^{\left( \text{red} \right)}_{\sigma} 
    \end{equation*}
    is introduced. Analogously,
    \begin{equation*}
        \tau: \psi \in H^1\left( \mathbb{R}^n, dx_1 \cdots dx_n \right) \longmapsto \tau\psi \doteq \left(\tau_{\sigma} \psi \right)_{\sigma = 1}^{\frac{n(n-1)}{2}} \in \chi^{\left( \text{red} \right)} = \bigoplus_{\sigma} \chi^{\left( \text{red} \right)}_{\sigma}
    \end{equation*}
    is considered. $\hfill \square$
\end{definition}

\begin{theorem}
    \justifying
    There exists a linear operator $\Theta\left( z \right) \in \mathcal{B}\left(\chi^{\left( \text{red} \right)}\right)$, $z<0$, such that
    \begin{equation}\label{Resolvent Formula for Self-Adjoint Extension}
        R_{H}\left( z \right) = R_{H_0}\left( z \right) + g G^{\ast}\left( \overline{z} \right) \Theta^{-1}\left( z \right) G\left( z \right)
    \end{equation}
    for all $z < z_0$.
\end{theorem}

\begin{proof}
    \justifying
    As long as $z<z_0$, it is first observed that
    \begin{align*}
        & \left[ \Lambda_0\left( z \right)_{\text{diag}}^{-1} \Lambda_0\left( z \right)_{\text{off}} \right]_{\sigma \nu} = \left( 1 - \delta_{\sigma \nu} \right) \left[ \Lambda_0\left( z \right)_{\text{diag}}^{-1} \right]_{\sigma \sigma} \left[ \Lambda_0\left( z \right)_{\text{off}} \right]_{\sigma \nu} = \\
        & = \left( 1 - \delta_{\sigma \nu} \right) \left\{ \mathds{1}_{\chi^{(\text{red})}_{\sigma}} + \langle v_{\sigma}, \cdot \rangle v_{\sigma} \otimes g D^{(\sigma)}\left( z \right) \left[ \mathds{1}_{\chi^{(\text{red})}_{\sigma}} - g D^{(\sigma)}\left( z \right) \right]^{-1} \right\} \left\{ \langle v_{\nu}, \cdot \rangle v_{\sigma} \otimes \left[ - g \tau_{\sigma} G_{\nu}^{\ast} \left( z \right) \right] \right\} = \\
        & = \langle v_{\nu}, \cdot \rangle v_{\sigma} \otimes \left( 1 - \delta_{\sigma \nu} \right) \left\{ \mathds{1}_{\chi^{(\text{red})}_{\sigma}} + g D^{(\sigma)}\left( z \right) \left[ \mathds{1}_{\chi^{(\text{red})}_{\sigma}} - g D^{(\sigma)}\left( z \right) \right]^{-1} \right\} \left[ - g \tau_{\sigma} G_{\nu}^{\ast}\left(z\right) \right].
    \end{align*}
    Given $z<0$, by introducing the linear operators
    \begin{align}
        \left[\Theta\left(z\right)_{\text{diag}}\right]_{\sigma \nu} & = \left[\mathds{1}_{\chi^{(\text{red})}_{\sigma}} - g D^{(\sigma)} \left(z\right) \right] \delta_{\sigma \nu} \in \mathcal{B}\left( \chi^{(\text{red})}_{\nu}, \chi^{(\text{red})}_{\sigma} \right) \\
        \left[\Theta\left(z\right)_{\text{off}}\right]_{\sigma \nu} & = \left[ - g \tau_{\sigma} G_{\nu}^{\ast}\left( z \right) \right] \left( 1 - \delta_{\sigma \nu}\right) \in \mathcal{B}\left( \chi^{(\text{red})}_{\nu}, \chi^{(\text{red})}_{\sigma} \right)
    \end{align}
    for all $\sigma, \nu \in \mathcal{I}$, as long as $z < z_0$, straightforward computations allow to state that
    \begin{equation*}
        \left\{ \left[ \Theta\left(z\right)_{\text{diag}} \right]_{\sigma \sigma} \right\}^{-1} = \mathds{1}_{\chi^{(\text{red})}_{\sigma}} + g D^{(\sigma)}\left( z \right) \left[ \mathds{1}_{\chi^{(\text{red})}_{\sigma}} - g D^{(\sigma)}\left( z \right) \right]^{-1},
    \end{equation*}
    therefore, by introducing the bounded operators
    \begin{align}
        \Theta\left(z\right)_{\text{diag}} & = \left( \Theta_{\text{diag}}\left(z\right)_{\sigma \nu} \right)_{\sigma, \nu \in \mathcal{I}} \in \mathcal{B}\left( \chi^{(\text{red})} \right)\\
        \Theta\left(z\right)_{\text{off}} & = \left( \Theta_{\text{off}}\left(z\right)_{\sigma \nu} \right)_{\sigma, \nu \in \mathcal{I}} \in \mathcal{B}\left( \chi^{(\text{red})} \right)
    \end{align}
    hence $\Theta\left(z\right) \doteq \Theta\left(z\right)_{\text{diag}} + \Theta\left(z\right)_{\text{off}}$,
    \begin{equation*}
        \left[ \Lambda_0\left( z \right)_{\text{diag}}^{-1} \Lambda_0\left( z \right)_{\text{off}} \right]_{\sigma \nu} = \langle v_{\nu}, \cdot \rangle v_{\sigma} \otimes \left[ \Theta\left(z\right)_{\text{diag}}^{-1} \Theta\left(z\right)_{\text{off}} \right]_{\sigma \nu} \equiv U_{\sigma \nu} \otimes \left[ \Theta\left(z\right)_{\text{diag}}^{-1} \Theta\left(z\right)_{\text{off}} \right]_{\sigma \nu}.
    \end{equation*}
    Further introduced $U \doteq \left( U_{\sigma \nu} \right)_{\sigma, \nu} \equiv \left( \langle v_{\nu}, \cdot \rangle v_{\sigma} \right)_{\sigma, \nu} \in \mathcal{B}\left( \bigoplus_{\nu} L^2\left( \mathbb{R}, dr_{\nu} \right) \right)$, denoted by $\circ$ the Hadamard product,
    \begin{equation}
        \left[ \Lambda_0\left( z \right)_{\text{diag}}^{-1} \Lambda_0\left( z \right)_{\text{off}} \right]_{\sigma \nu} = \left[ U \circ \Theta\left(z\right)_{\text{diag}}^{-1} \Theta\left(z\right)_{\text{off}} \right]_{\sigma \nu} \in \mathcal{B}\left( \chi \right).
    \end{equation}
    Since
    \begin{equation*}
        \Theta\left(z\right)^{-1} = \left[ \mathds{1}_{\chi^{(\text{red})}} +  \Theta\left(z\right)_{\text{diag}}^{-1} \Theta\left(z\right)_{\text{off}} \right]^{-1} \Theta\left(z\right)_{\text{diag}}^{-1},
    \end{equation*}
    a direct computation shows that
    \begin{equation*}
        \left\{ \mathds{1}_{\chi} + U \circ \left[ \Theta\left(z\right)_{\text{diag}}^{-1} \Theta\left(z\right)_{\text{off}} \right] \right\}^{-1} = \mathds{1}_{\chi} - U \circ \left[ \Theta\left(z\right)^{-1} \Theta\left(z\right)_{\text{off}} \right].
    \end{equation*}
    In the end
    \begin{equation}
        \left[ \Lambda_0\left( z \right)^{-1} \right]_{\sigma \nu} = \left[ U \circ \Theta\left( z \right)^{-1} \right]_{\sigma \nu} + \left\{ \left[ \mathds{1}_{L^2\left( \mathbb{R},dr_{\sigma} \right)} - \langle v_{\sigma}, \cdot \rangle v_{\sigma} \right] \otimes \mathds{1}_{\chi^{(\text{red})}_{\sigma}} \right\} \delta_{\sigma \nu},
    \end{equation}
    hence
    \begin{equation}
        \sum_{\sigma, \nu} S^{(\sigma)}\left( \overline{z} \right)^{\ast} \left[ \Lambda_0^{-1}\left( z \right) \right]_{\sigma \nu} S^{(\nu)}\left( z \right) = \sum_{\sigma, \nu} G_{\sigma}^{\ast}\left( \overline{z} \right) \left[\Theta\left(z\right)^{-1}\right]_{\sigma \nu} G_{\nu}\left( z \right) \equiv G^{\ast}\left( \overline{z} \right) \Theta\left(z\right)^{-1} G\left( z \right).
    \end{equation}
\end{proof}

\begin{remark}\label{H is the self-adjoint extension of H_0}
    \justifying
    (\cite{07_CFP}) thm. 2.19 states that, if $\Theta$ could be defined for all $z \in \rho\left(H_0\right)$ in a such a way that
    \begin{enumerate}
        \item $\mathcal{D}_{\Theta\left(z\right)}$ is independent on $z$,
        \item $\Theta\left(z\right)^{\ast} = \Theta\left( \overline{z} \right)$, for all $z \in \rho\left(H_0\right)$,
        \item $\Theta\left(z\right) = \Theta\left( w \right) + g \left(w-z\right) G\left(w\right) G\left(\overline{z}\right)^{\ast}$, for all $z,w \in \rho\left(H_0\right): \; z \neq w$,
        \item $0 \in \rho \left( \Theta\left(z\right) \right)$, for some $z \in \rho\left(H_0\right)$,
    \end{enumerate}
    (\ref{Resolvent Formula for Self-Adjoint Extension}) would then hold for all $z \in \rho\left(H\right) \cap \rho\left(H_0\right)$, allowing for $\left(H, \mathcal{D}_{H}\right)$ to be the unique self-adjoint extension of $\left( H_0, \ker \tau \right)$. What follows proves that this is the case. $\hfill \square$
\end{remark}

\begin{lemma}
    \justifying
    For all $\sigma \in \mathcal{I}$, $z \in \rho\left( H_0 \right)$, $\varphi \in \chi_{\sigma}^{(\text{red})}$
    \begin{align*}
        \left\{ \mathfrak{F}\left[ \tau_0 R_{H_0^{\sigma}}\left( \overline{z} \right) \right]^{\ast} \varphi \right\}\left( p_{\sigma}, P_{\sigma}, \underline{P}_{\sigma} \right) & = \mu_{\sigma} \sqrt{\frac{2}{\pi}} \cdot \frac{\left( \mathfrak{F}_{\underline{Y}_{\sigma}} \varphi \right) \left( P_{\sigma}, \underline{P}_{\sigma} \right) }{ p_{\sigma}^2 + 2\mu_{\sigma}\left( Q_{\sigma} - z \right) } \\ 
        & \equiv \left\{ \left[ \frac{p_{\sigma}^2}{2\mu_{\sigma}} + \left( Q_{\sigma} - z \right) \right]^{-1} \hat{\tau}_0^{\ast} \left( \mathfrak{F}_{\underline{Y}_{\sigma}} \varphi \right) \right\} \left( p_{\sigma}, P_{\sigma}, \underline{P}_{\sigma} \right).
    \end{align*}
\end{lemma}

\begin{proof}
    \justifying
    For all $\psi \in \chi_{\sigma}$, $\varphi \in \chi_{\sigma}^{(\text{red})}$
    \begin{align*}
        & \langle \varphi, \left[ \tau_0 R_{H_0^{\sigma}}\left( \overline{z} \right) \right] \psi \rangle = \int_{\mathbb{R}^{n-1}} dR_{\sigma} d\underline{Y}_{\sigma} \,  \overline{\varphi}\left( R_{\sigma}, \underline{Y}_{\sigma} \right) \left\{ \left[ \tau_0 R_{H_0^{\sigma}}\left( \overline{z} \right) \right] \psi \right\} \left( R_{\sigma}, \underline{Y}_{\sigma} \right) = \left( \text{Plancherel} \right) \equiv \\
        & = \int_{\mathbb{R}^{n-1}}dP_{\sigma}d\underline{P}_{\sigma} \overline{\left( \mathfrak{F}_{\underline{Y}_{\sigma}} \varphi \right)}\left( P_{\sigma}, \underline{P}_{\sigma} \right) \left\{ \mathfrak{F}_{\underline{Y}_{\sigma}}\left[ \tau_0 R_{H_0^{\sigma}}\left( \overline{z} \right) \psi \right] \right\}\left( P_{\sigma}, \underline{P}_{\sigma} \right) = \\
        & = \int_{\mathbb{R}^{n-1}}dP_{\sigma}d\underline{P}_{\sigma} \overline{\left( \mathfrak{F}_{\underline{Y}_{\sigma}} \varphi \right)}\left( P_{\sigma}, \underline{P}_{\sigma} \right) \hat{\tau_0}\left[ \frac{p_{\sigma}^2}{2\mu_{\sigma}} + \left( Q_{\sigma} - \overline{z} \right) \right]^{-1} \left( \mathfrak{F}\psi \right) \left( p_{\sigma}, P_{\sigma}, \underline{P}_{\sigma} \right) = \\
        & = \int_{\mathbb{R}^{n-1}}dP_{\sigma}d\underline{P}_{\sigma} \overline{\left( \mathfrak{F}_{\underline{Y}_{\sigma}} \varphi \right)}\left( P_{\sigma}, \underline{P}_{\sigma} \right) \left[ \int_{\mathbb{R}} \frac{\left( \mathfrak{F}\psi \right) \left( p_{\sigma}, P_{\sigma}, \underline{P}_{\sigma} \right)}{\frac{p_{\sigma}^2}{2\mu_{\sigma}} + \left( Q_{\sigma} - \overline{z} \right)} \frac{dp_{\sigma}}{\sqrt{2\pi}} \right] = \\
        & = \int_{\mathbb{R}^n} \overline{ \mathfrak{F}^{-1} \left\{ \frac{1}{\sqrt{2\pi}} \left[ \frac{ \left( \mathfrak{F}_{\underline{Y}_{\sigma}} \varphi \right) \left( P_{\sigma}, \underline{P}_{\sigma} \right) }{\frac{p_{\sigma}^2}{2\mu_{\sigma}} + \left( Q_{\sigma} - z \right)} \right] \right\} } \left( r_{\sigma}, R_{\sigma}, \underline{Y}_{\sigma} \right) \psi\left( r_{\sigma}, R_{\sigma}, \underline{Y}_{\sigma} \right) dr_{\sigma}dR_{\sigma}d\underline{Y}_{\sigma} = \\
        & = \langle \left[ \tau_0 R_{H_0^{\sigma}}\left( \overline{z} \right) \right]^{\ast} \varphi, \psi \rangle.
    \end{align*}
\end{proof}

\begin{corollary}
    \justifying
    For all $\sigma \in \mathcal{I},$ $z<0$, $D^{\left( \sigma \right)}\left( z \right) = \tau_{\sigma} G_{\sigma}^{\ast}\left( \overline{z} \right)$.
\end{corollary}

\begin{proof}
    \justifying
    Let $z \in \rho\left( H_0 \right)$ be arbitrary.
    \begin{align*}
        G_{\sigma}^{\ast}\left( \overline{z} \right) & = \left[ \tau_{\sigma} R_{H_0}\left( \overline{z} \right) \right]^{\ast} = R_{H_0}\left( \overline{z} \right)^{\ast} \tau_{\sigma}^{\ast} = R_{H_0}\left( z \right) \left( \tau_0 U_{\sigma} \right)^{\ast} = R_{H_0}\left( z \right) U_{\sigma}^{\ast} \tau_{0}^{\ast} = \left[ U_{\sigma} R_{H_0}\left( \overline{z} \right) \right]^{\ast} \tau_{0}^{\ast} = \\
        & = \left[ R_{H_0^{\sigma}}\left( \overline{z} \right) U_{\sigma} \right]^{\ast} \tau_0^{\ast} = U_{\sigma}^{\ast} R_{H_0^{\sigma}}\left( z \right) \tau_0^{\ast}.
    \end{align*}
    Consequently
    \begin{equation*}
        \tau_{\sigma} G_{\sigma}^{\ast}\left( \overline{z} \right) = \tau_0 U_{\sigma}U_{\sigma}^{\ast}R_{H_0^{\sigma}}\left( z \right)\tau_0^{\ast} = \tau_0 R_{H_0^{\sigma}}\left( z \right) \tau_0^{\ast} = \mathfrak{F}_{\underline{Y}_{\sigma}}^{-1} \left\{ \hat{\tau}_0 \left[ \frac{p_{\sigma}^2}{2\mu_{\sigma}} + \left( Q_{\sigma} - z \right) \right]^{-1} \hat{\tau}_0^{\ast} \right\} \mathfrak{F}_{\underline{Y}_{\sigma}},
    \end{equation*}
    hence, for all $\varphi \in \chi^{(\text{red})}_{\sigma}$
    \begin{align*}
        \left\{ \hat{\tau}_0 \left[ \frac{p_{\sigma}^2}{2\mu_{\sigma}} + \left( Q_{\sigma} - z \right) \right]^{-1} \hat{\tau}_0^{\ast} \left( \mathfrak{F}_{\underline{Y}_{\sigma}} \varphi \right) \right\} \left( P_{\sigma}, \underline{P}_{\sigma} \right) & = \frac{\mu_{\sigma}}{\pi} \int_{\mathbb{R}} \, \frac{ \left( \mathfrak{F}_{\underline{Y}_{\sigma}} \varphi \right)\left( P_{\sigma}, \underline{P}_{\sigma} \right) }{ \frac{p_{\sigma}^2}{2\mu_{\sigma}} + 2\mu_{\sigma} \left( Q_{\sigma} - z \right) } dp_{\sigma} = \\
        & = \left\{ \left[ \sqrt{\frac{\mu_\sigma}{2}} \frac{1}{\sqrt{Q_{\sigma} - z}} \right] \left( \mathfrak{F}_{\underline{Y}_{\sigma}}  \varphi \right) \right\} \left( P_{\sigma}, \underline{P}_{\sigma} \right).
    \end{align*}
\end{proof}

\begin{theorem}
    \justifying
    $\left( H, \mathcal{D}_{H} \right)$ is the unique self-adjoint extension of $\left( H_0, \ker \tau \right)$.
\end{theorem}

\begin{proof}
    \justifying
    By having established that
    \begin{align*}
        \left[ \Theta\left(z\right)_{\text{diag}} \right]_{\sigma \sigma} & = \mathds{1}_{\chi^{(\text{red})}_{\sigma}} - g \tau_{\sigma} G_{\sigma}^{\ast}\left( \overline{z} \right)\\
        \left[ \Theta\left(z\right)_{\text{off}} \right]_{\sigma \nu} & = - g \tau_{\sigma} G_{\nu}^{\ast}\left( \overline{z} \right)
    \end{align*}
    for $z<0$, the first resolvent formula allows to analytically extend these equalities to $\rho\left( H_0 \right)$ entirely, for all $\sigma, \nu \in \mathcal{I}$. The domain of the operators is independent on $z$ and, as long as $z<z_0$, $\Theta\left(z\right)$ is invertible in $\mathcal{B}\left( \chi^{(\text{red})} \right)$. Then,
    \begin{enumerate}
        \item for all $z \in \rho\left( H_0 \right)$, $\sigma \in \mathcal{I}$
        \begin{align*}
            \left[ \Theta\left(z\right)_{(\text{diag})} \right]^{\ast}_{\sigma \sigma} & = \mathds{1}_{\chi^{\text{red}}_{\sigma}} - g \left[ \tau_{\sigma} G_{\sigma}^{\ast}\left( \overline{z} \right) \right]^{\ast} = \mathds{1}_{\chi^{\text{red}}_{\sigma}} - g \tau_{\sigma} R_{H_0}\left( \overline{z} \right) \tau_{\sigma}^{\ast} = \mathds{1}_{\chi^{\text{red}}_{\sigma}} - g \tau_{\sigma} \left[ R_{H_0}\left( z \right)^{\ast} \tau_{\sigma}^{\ast} \right] = \\
            & = \mathds{1}_{\chi^{\text{red}}_{\sigma}} - g \tau_{\sigma} G_{\sigma}^{\ast}\left( z \right) = \left[ \Theta\left( \overline{z} \right)_{\text{diag}} \right]_{\sigma \sigma}.
        \end{align*}
        \item for all $z \in \rho\left( H_0 \right)$, $\sigma, \nu \in \mathcal{I}: \, \sigma \neq \nu$
        \begin{align*}
            \left[ \Theta\left( z \right)_{\text{off}} \right]_{\sigma \nu}^{\ast} & = - g G_{\nu}\left( \overline{z} \right) \tau_{\sigma}^{\ast} = - g \tau_{\nu} R_{H_0}\left( \overline{z} \right) \tau_{\sigma}^{\ast} = - g \tau_{\nu} \left[ R_{H_0}\left( z \right)^{\ast} \tau_{\sigma}^{\ast} \right] = - g \tau_{\nu} G_{\sigma}^{\ast}\left( z \right) = \\
            & = \left[ \Theta\left( \overline{z} \right)_{\text{off}} \right]_{\nu \sigma}.
        \end{align*}
    \end{enumerate}
    Further, given $\sigma \in \mathcal{I}$ arbitrary, let $z,w \in \rho\left( H_0 \right)$ be such that $z \neq w$. The first resolvent formula allows to prove that
    \begin{equation*}
        G_{\sigma}^{\ast}\left( w \right) - G_{\sigma}^{\ast}\left( z \right) = \left( \overline{w} - \overline{z} \right) R_{H_0}\left( \overline{w} \right) R_{H_0}\left( \overline{z} \right) \tau_{\sigma}^{\ast} \equiv \left( \overline{w} - \overline{z} \right) R_{H_0}\left( \overline{w} \right) G_{\sigma}^{\ast}\left( z \right). 
    \end{equation*}
    Then,
    \begin{enumerate}
        \item for all $\sigma \in \mathcal{I}$, $z,w \in \rho\left( H_0 \right)$ as above 
        \begin{equation*}
            \left[\Theta\left( z \right)_{\text{diag}}\right]_{\sigma \sigma} - \left[\Theta\left( w \right)_{\text{diag}}\right]_{\sigma \sigma} = g \tau_{\sigma} \left[ G_{\sigma}^{\ast}\left( \overline{w} \right) - G_{\sigma}^{\ast}\left( \overline{z} \right) \right] \equiv g \left( w -z  \right) G_{\sigma}\left( w \right) G_{\sigma}^{\ast}\left( \overline{z} \right).
        \end{equation*}
        \item for all $\sigma, \nu \in \mathcal{I}: \, \sigma \neq \nu$, $z,w \in \rho\left( H_0 \right)$ as above
        \begin{equation*}
            \left[\Theta\left( z \right)_{\text{off}}\right]_{\sigma \nu} - \left[\Theta\left( w \right)_{\text{off}}\right]_{\sigma \nu} = g \tau_{\sigma} \left[ G_{\nu}^{\ast}\left( \overline{w} \right) - G_{\nu}^{\ast}\left( \overline{z} \right) \right] = g \left( w - z  \right) G_{\sigma}\left( w \right) G_{\nu}^{\ast} \left( \overline{z} \right).
        \end{equation*}
    \end{enumerate}
    The statement then follows from remark \ref{H is the self-adjoint extension of H_0}.
\end{proof}

\begin{lemma}
    \justifying
    Under the foregoing hypothesis, fixed $z \in \rho\left( H_0 \right) \cap \rho\left( H \right)$
    \begin{equation*}
        \mathcal{D}_{H} = \left\{ \psi \in \mathcal{H} \, \bigg \vert \, \exists ! \, \varphi \in \mathcal{D}_{H_0} : \psi = \varphi + g G\left( \overline{z} \right)^{\ast} \Theta\left( z \right)^{-1} \tau \varphi \right\}.
    \end{equation*}
\end{lemma}

\begin{proof}
    \justifying
    It is well known that $\mathcal{D}_{H} = \text{Ran} \, R_{H}\left( z \right)$ for all $z \in \rho\left( H \right)$, independently of $z$. Infact, considered $z_1 \in \rho\left( H \right)$, let
    \begin{equation*}
        \mathcal{D}_{H}^{\left( z_1 \right)} = \left\{ \psi \in \mathcal{H} \bigg \vert \psi = R_{H}\left( z_1 \right)\varphi, \; \varphi \in \mathcal{H} \right\}
    \end{equation*}
    \noindent be. Let then $z_2 \in \rho\left( H \right)$ be such that $\abs{z_1 - z_2} \leq \norm{R_{H}\left( z_2 \right)}^{-1}$ and let $\psi \in \mathcal{D}_{H}^{\left( z_1 \right)}$ be arbitrary. There exists a unique $\varphi_{z_1} \in \mathcal{H}$ such that $\psi = R_{H}\left( z_1 \right) \varphi_{z_1}$. However, the Neumann expansion formula gives
    \begin{equation*}
        \psi = R_{H}\left( z_1 \right) \varphi_{z_1} = R_{H}\left( z_2 \right) \left[ \underset{n \in \mathbb{N}_0}{\sum} \, \left( z_1 - z_2 \right)^n R_{H}\left( z_2 \right)^n \varphi_{z_1} \right] \equiv R_{H}\left( z_2 \right) \varphi_{z_2},
    \end{equation*}
    \noindent i.e. $\psi \in \mathcal{D}_{H}^{\left( z_2 \right)}$, hence $\mathcal{D}_{H}^{\left( z_1 \right)} \subseteq \mathcal{D}_{H}^{\left( z_2 \right)}$. The roles of $z_1,z_2$ are nevertheless switchable, therefore $\mathcal{D}_{H}^{\left( z_1 \right)} = \mathcal{D}_{H}^{\left( z_2 \right)}$, proving the independence of $\mathcal{D}_{H}$ on $z \in \rho\left(H\right)$. The result then follows by arbitrarily fixing $z \in \rho\left( H \right) \cap \rho\left( H_0 \right)$ and by observing that, given $\psi \in \mathcal{D}_{H}$ generic, $\varphi = R_{H_0}\left( z \right) \left( H - z \right) \psi \in \mathcal{D}_{H_0}$.
\end{proof}


\end{document}